\PassOptionsToPackage{unicode}{hyperref}
\PassOptionsToPackage{hyphens}{url}
\PassOptionsToPackage{dvipsnames,svgnames,x11names}{xcolor}

\documentclass[12pt]{article}

\usepackage{xurl}
\usepackage{amsmath,amssymb}
\usepackage{iftex}
\ifPDFTeX
  \usepackage[T1]{fontenc}
  \usepackage[utf8]{inputenc}
  \usepackage{textcomp} 
\else 
  \usepackage{unicode-math}
  \defaultfontfeatures{Scale=MatchLowercase}
  \defaultfontfeatures[\rmfamily]{Ligatures=TeX,Scale=1}
\fi
\usepackage{lmodern}
\ifPDFTeX\else  
\fi
\IfFileExists{upquote.sty}{\usepackage{upquote}}{}
\IfFileExists{microtype.sty}{
  \usepackage[]{microtype}
  \UseMicrotypeSet[protrusion]{basicmath} 
}{}
\makeatletter
\@ifundefined{KOMAClassName}{
  \IfFileExists{parskip.sty}{%
    \usepackage{parskip}
  }{
    \setlength{\parindent}{0pt}
    \setlength{\parskip}{6pt plus 2pt minus 1pt}}
}{
  \KOMAoptions{parskip=half}}
\makeatother
\usepackage{xcolor}

\makeatletter
\ifx\paragraph\undefined\else
  \let\oldparagraph\paragraph
  \renewcommand{\paragraph}{
    \@ifstar
      \xxxParagraphStar
      \xxxParagraphNoStar
  }
  \newcommand{\xxxParagraphStar}[1]{\oldparagraph*{#1}\mbox{}}
  \newcommand{\xxxParagraphNoStar}[1]{\oldparagraph{#1}\mbox{}}
\fi
\ifx\subparagraph\undefined\else
  \let\oldsubparagraph\subparagraph
  \renewcommand{\subparagraph}{
    \@ifstar
      \xxxSubParagraphStar
      \xxxSubParagraphNoStar
  }
  \newcommand{\xxxSubParagraphStar}[1]{\oldsubparagraph*{#1}\mbox{}}
  \newcommand{\xxxSubParagraphNoStar}[1]{\oldsubparagraph{#1}\mbox{}}
\fi
\makeatother

\usepackage{longtable,booktabs,array}
\usepackage{calc} 
\usepackage{etoolbox}
\makeatletter
\patchcmd\longtable{\par}{\if@noskipsec\mbox{}\fi\par}{}{}
\makeatother
\IfFileExists{footnotehyper.sty}{\usepackage{footnotehyper}}{\usepackage{footnote}}
\makesavenoteenv{longtable}
\usepackage{graphicx}
\makeatletter
\def\maxwidth{\ifdim\Gin@nat@width>\linewidth\linewidth\else\Gin@nat@width\fi}
\def\maxheight{\ifdim\Gin@nat@height>\textheight\textheight\else\Gin@nat@height\fi}
\makeatother
\setkeys{Gin}{width=\maxwidth,height=\maxheight,keepaspectratio}
\makeatletter
\def\fps@figure{htbp}
\makeatother

\makeatletter
\@ifpackageloaded{caption}{}{\usepackage{caption}}
\AtBeginDocument{%
\ifdefined\contentsname
  \renewcommand*\contentsname{Table of contents}
\else
  \newcommand\contentsname{Table of contents}
\fi
\ifdefined\listfigurename
  \renewcommand*\listfigurename{List of Figures}
\else
  \newcommand\listfigurename{List of Figures}
\fi
\ifdefined\listtablename
  \renewcommand*\listtablename{List of Tables}
\else
  \newcommand\listtablename{List of Tables}
\fi
\ifdefined\figurename
  \renewcommand*\figurename{Figure}
\else
  \newcommand\figurename{Figure}
\fi
\ifdefined\tablename
  \renewcommand*\tablename{Table}
\else
  \newcommand\tablename{Table}
\fi
}
\@ifpackageloaded{float}{}{\usepackage{float}}
\floatstyle{ruled}
\@ifundefined{c@chapter}{\newfloat{codelisting}{h}{lop}}{\newfloat{codelisting}{h}{lop}[chapter]}
\floatname{codelisting}{Listing}

\makeatother
\makeatletter
\@ifpackageloaded{caption}{}{\usepackage{caption}}
\@ifpackageloaded{subcaption}{}{\usepackage{subcaption}}
\makeatother

\ifLuaTeX
  \usepackage{selnolig}  
\fi
 \usepackage[round]{natbib}
\usepackage{bookmark}

\IfFileExists{xurl.sty}{\usepackage{xurl}}{} 
\definecolor{mycyan}{HTML}{007FA3}
\hypersetup{
  pdftitle={Title},
  pdfauthor={Author 1; Author 2},
  pdfkeywords={3 to 6 keywords, that do not appear in the title},
  colorlinks=true,
  linkcolor={mycyan},
  filecolor={Maroon},
  citecolor={mycyan},
  urlcolor={mycyan},
  pdfcreator={LaTeX via pandoc}}

\usepackage{hyperref}

  \usepackage{graphicx,psfrag,epsf}
\usepackage{enumerate}
\usepackage{url} 
\usepackage{tikz, adjustbox, float} 
\usepackage{amsmath, amsfonts, amssymb, amsthm, dsfont} 
\usepackage{mathrsfs} 
\usepackage{soul} 
\usepackage{bm} 
\usepackage{booktabs}  
\usepackage{enumitem} 
\usepackage{mathtools}
\usepackage{mathabx}
\usepackage{nicefrac}
\usepackage{tabularx}
\usepackage{algorithm}
\usepackage{algpseudocode}
\usepackage{cleveref}

\usepackage{comment}

\usepackage{bbm} 

\usepackage{float}

\usepackage{wrapfig}

\usepackage{setspace}

\usepackage{tabularx}
\usepackage{makecell}

\usepackage{graphicx}

\usepackage{placeins}

\usepackage[letterpaper, top=2cm, bottom=2cm, left=2.05cm, right=2.05cm]{geometry}

\usepackage{url}
\usepackage{hyperref}

\def\boxit#1{\vbox{\hrule\hbox{\vrule\kern6pt  \vbox{\kern6pt#1\kern6pt}\kern6pt\vrule}\hrule}}

\def\real{{\rm I\!R}}

\DeclareMathOperator*{\argmin}{argmin}
\DeclareMathOperator*{\tr}{tr}

\DeclareMathOperator*{\cov}{cov}
\def\0{{\bf 0}}

\DeclareMathOperator*{\var}{var}

\DeclareMathOperator*{\diag}{diag}

\DeclareMathOperator*{\op}{op}

\def\bOmega{{\bm{\Omega}}}

\def\X{{\bf X}}
\def\x{{\bf x}}
\def\I{{\bf I}}

\def\A{{\bf A}}
\def\B{{\bf B}}
\def\C{{\bf C}}
\def\D{{\bf D}}
\def\G{{\bf G}}
\def\H{{\bf H}}
\def\h{{\bf h}}
\def\M{{\bf M}}
\def\P{{\bf P}}
\def\Q{{\bf Q}}
\def\R{{\bf R}}
\def\r{{\bf r}}
\def\S{{\bf S}}
\def\V{{\bf V}}
\def\W{{\bf W}}
\def\Y{\mathbf{Y}}
\def\y{{\bf y}}

\def\Z{{\bf Z}}

\def\bnu{{\boldsymbol\nu}}

\def\bgamma{\boldsymbol\gamma}
\def\bg{\boldsymbol\gamma}

\def\trans{^{\rm T}}

\newtheoremstyle{mytheoremstyle} 
    {0.3cm}                      
    {0cm}                        
    {\itshape}                   
    {}                           
    {\bf}                   
    {: }                          
    {0em}                       
    {}  

\theoremstyle{mytheoremstyle}
\newtheorem{Theorem}{Theorem}

\newtheorem{Lemma}{Lemma}
\newtheorem*{Lemma*}{Lemma}

\newtheorem{Proposition}{Proposition}

\newtheoremstyle{myExampleRemarkstyle} 
    {0.3cm}                    
    {0cm}                           
    {\itshape}                   
    {}                           
    {\bf}                   
    {: }                          
    {0em}                       
    {}  

\theoremstyle{myExampleRemarkstyle}
 
\newtheorem{Remark}{Remark}
\newtheorem{Assumption}{Assumption}

\newtheorem{innercustomgeneric}{\customgenericname}
\providecommand{\customgenericname}{}
\newcommand{\newcustomtheorem}[2]{%
  \newenvironment{#1}[1]
  {%
   \renewcommand\customgenericname{#2}%
   \renewcommand\theinnercustomgeneric{##1}%
   \innercustomgeneric
  }
  {\endinnercustomgeneric}
}

\newcustomtheorem{customthm}{Theorem}
\newcustomtheorem{customlemma}{Lemma}

\newcommand{\anon}{1}

\begin{document}

\def\spacingset#1{\renewcommand{\baselinestretch}%
{#1}\small\normalsize} \spacingset{1}

\date{}

\if1\anon
{
\title{\vspace{-1.5cm} \bf Towards Open Science: Monitoring Crustal Deformations in North America}
  \author{
 Lionel~Voirol$^{1}$, Haotian~Xu$^{2}$,  Yuming Zhang$^{3}$, Luca Insolia$^{4, 5, 6}$,\\ Roberto Molinari$^{2}$ \& St\'ephane~Guerrier$^{4, 5, 6}$
 \vspace{.2cm}\\
{\small  $^{1}$Geneva School of Economics and Management, University of Geneva, Switzerland}\\
{\small $^{2}$Department of Mathematics and Statistics, Auburn University, United States}\\
{\small  $^{3}$Department of Biostatistics, Harvard T.H. Chan School of Public Health, United States}\\
{\small $^{4}$Institute of Pharmaceutical Sciences
of Western Switzerland, University of Geneva, Switzerland}\\
{\small $^{5}$School of Pharmaceutical Sciences, University of Geneva, Switzerland}\\
{\small $^{6}$Department of Earth Sciences, University of Geneva, Switzerland}
    }
  \maketitle
} 
\fi



\begin{abstract}
\noindent
The study of the Earth's behavior has greatly benefited from the widespread deployment of Global Navigation Satellite Systems (GNSS), enabling large-scale monitoring of crustal deformation and long-term geophysical trends. In this work, we focus on the North American region, where complex tectonic activity, particularly along the western margin, requires methods capable of processing and analyzing large collections of GNSS time series distributed across extensive spatial domains.
Analyzing the full GNSS network provides a coherent view of deformation across multiple scales, allowing detection of long-wavelength signals, subtle intraplate strain and regionally consistent velocity fields that are difficult to capture through local or subsampled analyses. Despite the availability of such data, existing methodologies remain computationally prohibitive for large-scale analyses across this region (or others). To address this limitation, we introduce a highly scalable framework (implemented in open-source software) that enables inference on crustal deformation across the full North American GNSS network using standard computational resources. The proposed method achieves substantial computational gains while maintaining inferential performance comparable to existing approaches and confirms existing tectonic trends, thereby supporting efficient large-scale monitoring and contributing to ongoing efforts toward ``Open Science''.
\end{abstract}

\noindent%
{\it Keywords:} Wavelet Variance, Generalized Method of Moments, GNSS Time Series, Missing Data, Model Selection, Long-Memory Processes
\vfill

\newpage
\spacingset{1.8} 

\section{Introduction}
\label{sec:intro}

The study of crustal deformations is one of the most important and impactful areas of research in Earth sciences, mainly focused on understanding and monitoring changes in the shape, position and structure of the Earth's lithosphere caused by tectonic forces, gravitational stresses, volcanic activity and surface processes. These deformations occur over a wide range of spatial and temporal scales, from slow continental drift to sudden earthquakes \citep{stein2003introduction, turcotte2014geodynamics} which, among others, manifest as uplift and horizontal displacements of the Earth's crust. As a result, aside from providing insights into the Earth's dynamic evolution \citep[see e.g.,][]{turcotte2014geodynamics}, the study of these phenomena is essential to perform efficient hazard mapping and risk reduction \citep{scholz2002mechanics} as well as sustainable resource exploration, extraction and management strategies \citep{fossen2016structural}. Within this context, a region of particular interest is North America which is characterized by a diverse tectonic setting that includes the active Pacific-North America plate boundary along California and Alaska, intraplate regions such as the stable continental interior, and postglacial rebound processes in the north \citep{calais2006continental,argus2010glacial}. These processes interact across large spatial scales, producing a deformation field that cannot be fully understood by focusing solely on localized or randomly sampled Global Navigation Satellite System (GNSS) networks. More precisely, these networks consist of permanent geodetic stations that continuously record precise three-dimensional positions over time, allowing surface motions to be measured at millimeter-level accuracy across local, continental or global scales. While small GNSS networks provide precise constraints on local deformation (e.g., along major faults or in urban regions), they often miss long-wavelength signals that are critical in regions such as North America, with postglacial isostatic adjustments across the Canadian Shield \citep{argus2010glacial, lambeck2014sea}, far-field plate-boundary effects extending from California into the Basin and Range and central United States (U.S.) \citep{calais2006continental}, and subtle intraplate strain accumulation associated with areas like the New Madrid seismic zone \citep{calais2005strain}.

To address this need for crustal deformation monitoring (as well as others) at larger scales, there has been an important increase in installation of GNSS stations across the globe which record relevant data. 
Indeed, the North American region has experienced rapid growth in GNSS ground stations used to observe signals from satellite constellations such as Global Positioning System (GPS), which is operated by the U.S. Space Force \citep{zawacki2025history}. Beyond navigation, these data provide essential information for monitoring crustal deformation at continental scales \citep[e.g.,][]{montillet2024big}. Such analyses rely on models that separate deterministic components (which capture signals of interest such as long-term velocities) from stochastic components whose correct specification is crucial for valid inference and uncertainty quantification \citep{williams2003effect, he2017review}. Because crustal deformation exhibits both short- and long-term variability, exploiting long and spatially extensive records is key to improving estimates of deformation rates, particularly velocities and their associated uncertainties \citep[e.g.,][]{lv2025investigating, kleinherenbrink2018comparison, mudelsee2019trend, aydin2021effect, maddanu2023trends, kermarrec2024modeling}. Confidence Intervals (CIs) and hypothesis tests play a central role in applications ranging from sea-level change to seismic hazard assessment \citep{wang202295, hohensinn2018stand, gokdacs2021velocity}. Nevertheless, despite this data expansion, analytic capabilities have not kept pace. Most current approaches rely on Gaussian likelihood-based frameworks that jointly estimate deterministic and stochastic components, requiring repeated evaluation of likelihoods involving large covariance matrices \citep{williams2003effect, serpelloni2022surface, sun2023relationship}. As record lengths and network sizes grow, these methods become computationally prohibitive, motivating the use of approximations that reduce cost \citep{Bos2008, bos2013fast}. However, while still remaining computationally demanding for practical purposes, such approximations often also come with limited or unclear statistical guarantees thereby highlighting the need for methods that are both scalable and inferentially reliable.

As a consequence, while the access to data has been opened to the wider research community, the result of the above-described setting is an important bottleneck in the actual capacity of researchers to benefit from these data. Indeed, while severely affecting the study and responsiveness to these natural phenomena, this aspect constitutes a fundamental barrier to achieving (or at least getting closer to) ``Open Science'', a global movement advocating for a democratization of science through the sharing of data, software and computational resources (among others) to make science more accessible, transparent and reproducible for everyone \citep{das2021unesco}. This movement has been particularly relevant for Earth sciences where, for example, National Aeronautics and Space Administration's  (NASA) ``Transform to Open Science'' program and its funded broad-access ``CryoCloud'' aim to share data and computing resources to accelerate research and more efficiently monitor global natural processes \citep{snow2023cryocloud}. These solutions indeed allow to increase access to the use of data and computational resources, but they still require practitioners and researchers to have access to high-performance computing clusters and are steered towards collaborative analyses rather than towards providing all researchers with locally-usable solutions that do not require additional costs or learning barriers. Moreover, the increased reliance on clusters and computing clouds has important environmental impacts due to their energy-consumption and heat-production, thereby directly affecting many natural phenomena they are being used to monitor \citep{achar2022cloud}.
\\
\textbf{Contributions}: To contribute to the goal of Open Science in this domain, this work puts forward a comprehensive statistical framework that provides an \textit{alternative to likelihood-based approaches} which, in practical terms, can deliver outputs in \textit{linear computational time} with \textit{well-defined statistical properties}. In contrast, exact likelihood-based methods generally exhibit quadratic or cubic computational complexity, whereas scalable likelihood-based implementations typically rely on additional approximations whose statistical guarantees depend on the approximation adopted. More specifically:
\begin{enumerate}
    \item We develop a framework that explicitly \textit{accounts for missing observations} in GNSS time series while \textit{preserving strong theoretical guarantees} for estimation and inference.
    
    \item We provide \textit{reliable uncertainty quantification}, demonstrating through extensive simulations that the proposed approach achieves \textit{inferential performance comparable to, or better than, likelihood-based methods}.
    
    \item We deliver these results with substantially \textit{improved computational efficiency}, enabling large-scale analyses to be performed in seconds rather than hours, thereby supporting more accessible and reproducible research.

    \item We introduce a \textit{computationally efficient model selection procedure for the stochastic component}, achieving performance comparable to likelihood-based methods such as the Akaike Information Criterion (AIC), while avoiding costly likelihood evaluations.

    \item We further develop \textit{efficient computational tools} for key quantities underlying the method, which are \textit{broadly applicable to wavelet-based analyses} and large time series.
    
    \item We provide an \textit{open-source \texttt{R} package and an interactive application} that support reproducible, end-to-end analysis and visualization of large-scale geodetic datasets.
\end{enumerate}

Based on these contributions, researchers monitoring the velocity of crustal deformations over the North American region can process and analyze all the data on standard laptops or, if needed, on small-scale computing resources which are commonly accessible to research communities (including those already made available through current Open Science initiatives). The proposed framework is built on the Generalized Method of Wavelet Moments (GMWM) which, in a two-step fashion, takes advantage of scalable wavelet decompositions and their relations to the covariance structure to model complex dependence structures (see e.g., \citealp{percival1995estimation,guerrier2013wavelet, cucci2023generalized}). More specifically, we fully deliver the properties of the GMWM when working with estimated residuals and missing values, thereby providing inference for regression coefficients on-par (or better) than the Maximum-Likelihood Estimator (MLE) as well as proposing a new efficient model selection method for the covariance structure. As highlighted later in this work, we achieve these results while observing computational speed-ups 100 times faster than the best benchmark approaches (which are computationally fast approximations to the standard likelihood) bringing times down from many hours to the order of seconds without losing inferential accuracy (but often improving it). 

\textbf{Organization}: To perform the case study of interest, Section~\ref{sec:data} provides an overview of the data used for the analysis of velocity of crustal deformations in North America, followed by Section~\ref{sec:framework} which describes the common statistical model used to analyze this data and the contributions of the proposed approach, including how it addresses limitations of existing methods. Section~\ref{sec:methodology} introduces the proposed methodology, including a new modelling framework to address missing data (Section~\ref{sec:missing_data}), with corresponding inferential and model selection properties (Sections~\ref{sec:inference} and \ref{sec:model_selection}) that are supported by Section~\ref{sec:simulations} in which we provide a brief overview of some simulation studies. Section \ref{sec:case_study} presents the first ever joint analysis of crustal deformation using data covering the entire GNSS network in the North American region thereby (i) delivering the first continental-scale, uncertainty-quantified velocity field with calibrated confidence intervals at every station, 
(ii) highlighting how this is achieved through standard computational resources and, consequently, (iii) allowing the needed timely updates for this highly active region. Finally, concluding remarks are provided in Section~\ref{sec:conclusion}. Supplementary materials include (i) a \texttt{GitHub} repository with all code used for the simulations and the case study, (ii) an open-source \texttt{R} package implementing the proposed method, and (iii) a web application illustrating the case-study results. Complete derivations of the theoretical results, together with additional simulation studies, are provided in the appendices.

\section{Geodetic Time Series Data}
\label{sec:data}

As a result of rapid technological advances, there has been a widespread global deployment of multiple satellite constellations (e.g., GPS, the Global Navigation Satellite System (GLONASS), Galileo, BeiDou) with a corresponding expansion of geodetic infrastructure where thousands of permanent GNSS stations now operate worldwide, continuously collecting vast amounts of data. Several web-based archives, including the International GNSS Service and the University NAVSTAR Consortium (UNAVCO), offer extensive collections of GNSS data products that support geoscience research and applications. The Nevada Geodetic Laboratory (NGL) further streamlines access by harvesting data from more than 130 online archives, providing a comprehensive set of GNSS products for over $20,000$ stations worldwide (\citealp{blewitt2018harnessing}). Among the available data products provided by the NGL are the daily spatial positions of geodetic stations expressed in graticule distance (see \citealp{blewitt2024improved}) within the IGS20 reference frame (see \citealp{rebischung2023reference}). The daily position solutions are derived from raw GNSS data processed with precise orbit, clock and atmospheric models. For each station and at each day, the position is described by three components: northing (coordinate in the north-south direction), easting (coordinate in the east-west direction), and height (vertical coordinate). Although numerous new stations have been installed in recent years, many GNSS stations have been operating continuously for several decades. Hence, for each station, this results in a time series (signal) for each of the three coordinates, each with thousands of observations, hence representing a substantial volume of data. These time series frequently include missing values due to factors such as power outages, hardware malfunctions, maintenance periods or data transmission failures, which complicate analysis and modeling (see e.g., \citealp{bao2021filling, liu2022missing}). The daily position time series are available at the NGL data online portal which contains the time series expressed in meters for each coordinate and for each station. Figure~\ref{fig:example_gnss_network_gnss_station_and_geodetic_time_series} illustrates a subset of the NGL GNSS network along the U.S. east coast, together with a schematic representation of a typical continuously operating GNSS station and representative northing, easting and height position time series.
\begin{figure}[]
    \centering
    \includegraphics[width=0.95\linewidth]{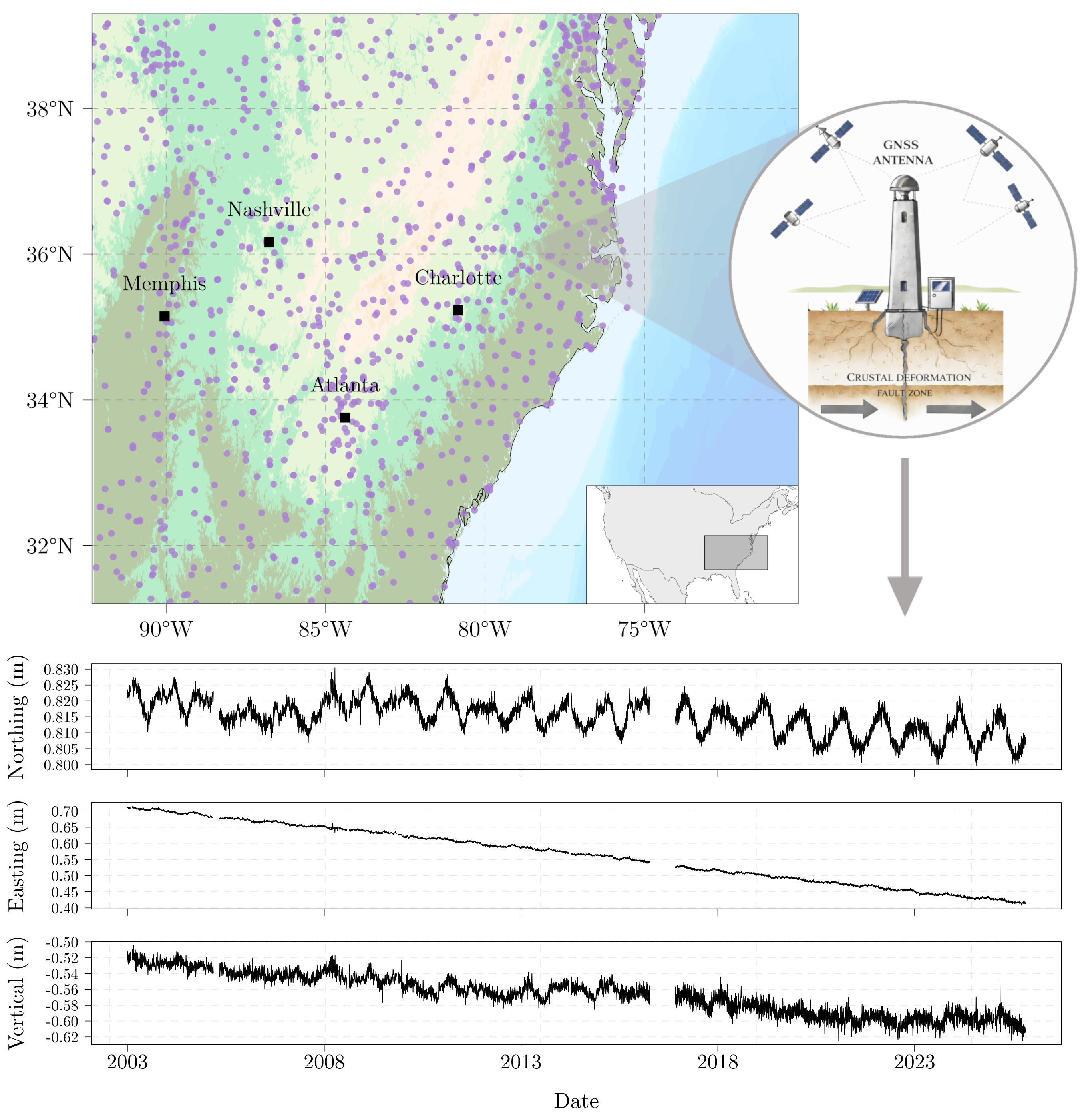}
    \caption{Illustration of a subset of the GNSS station network provided by the Nevada Geodetic Laboratory along the east coast of the United States. Purple dots represent station locations, with selected major cities indicated for geographical reference. The inset schematic illustrates a typical continuously operating GNSS station measuring crustal deformation. The lower panels display an example of geodetic position time series (North, East, and Vertical components) over multiple years for a representative station.}
    \label{fig:example_gnss_network_gnss_station_and_geodetic_time_series}
\end{figure}
The focus of this study is on the North American region, which is of particular interest due to its complex tectonic activity along the western margin, where the Pacific Plate and the North American Plate interact \citep{freymueller1999kinematics}. This plate boundary is characterized by significant compressional forces, leading to the development of subduction zones such as the Cascadia Subduction Zone, as well as intense shear stress that manifests in major transform faults, notably the San Andreas Fault system. The selected stations for this region are located between $180^\circ$W and $45^\circ$W longitude and between $10^\circ$N and $75^\circ$N latitude, and the criteria to select these stations from the NGL database consist in their signals having long observation periods and relatively low amounts of missing data. More in detail, we retained only stations in this area with more than 10 years of observations and less than 40\% missing data to ensure sufficient temporal coverage to reliably estimate long-term trends, seasonal components and temporally correlated noise. These thresholds represent a compromise between maximizing network coverage and ensuring stable and valid inference from each time series. These selection criteria deliver a network of $4,793$ GNSS stations, implying a sample of $4,793 \times 3 = 14,379$ signals including all three axes of measurement for each station, with an effective number of daily observations per station varying between $2,231$ and $11,667$, while the proportion of missing data in these time series ranges from $0\%$ to $39.99\%$. These portions of missing data can pose considerable challenges to adequately analyze the data when also considering the length of these signals. Indeed, there is a need to consider the entire signal in order to better characterize long-range dependence structures in these data and, for this reason, there is the consequent need to continuously update the models when new data come in. Aside from the computational strain, the main challenge from this perspective is characterizing the time-dependence structure in these data, considering that this is commonly constituted by a mix of different short- and long-term stochastic (noise) components which make the estimation of this structure particularly complex. In fact, while the estimation of the deterministic component can be relatively straightforward, it is the estimation of these stochastic components (which are essential for appropriate inference on the deterministic parameters of interest) that poses the major statistical and computational challenges.

More in detail, this scale of data analysis is (largely) prohibitive in terms of computational time and/or resources for many geoscientists. Indeed, the state-of-the-art software to analyze this data is represented by \texttt{Hector} \citep{bos2013fast}: this software provides the most efficient implementation of the MLE for these analyses which has been continuously optimized over the last decade, allowing it to become 10-100 times faster than the classical MLE, and is a standard tool for analyzing geodetic time series. Without considering the optimization procedure implemented in \texttt{Hector}, storing the covariance matrix of the shortest time series considered (i.e., $n = 2,231$) is already non-negligible, and additional temporary copies are typically required during optimization and matrix factorizations. More importantly, the computational burden of repeated likelihood evaluations and covariance matrix factorizations becomes substantial when model selection must be performed across multiple candidate stochastic models and thousands of long GNSS position time series. Although likelihood-based inference is practical for individual signals, scaling such analyses to large GNSS networks often requires dedicated high-performance computing resources. For example, the recent study of \cite{tunini2024global} relied on a dedicated computing cluster to analyze a relatively small GNSS network spanning two decades and required approximately one week of computation. Aside from the required time, many research teams cannot afford these types of computational resources, leading to limitations in terms of \textit{practical} access to data and in terms of Open Science in this line of research.

\section{Current Modeling Approaches}
\label{sec:framework}

The data described in Section \ref{sec:data} is commonly analyzed by considering a standard trajectory model for each signal described earlier: this model includes an intercept, a velocity (or trend) parameter, seasonal and half-seasonal components, offsets and post-seismic relaxations to model earthquakes near the stations. For this reason we also employ this model for our study since it adequately incorporates the main effects coming from such processes (see e.g., \citealp{bevis2014trajectory,he2017review,ren2021analysis}). More specifically, with $t_i$ representing the $i^{\text{th}}$ ordered time point of the time series (per station and per axis) and $t_0$ representing the reference epoch, the model for the mean function (deterministic component) can be represented as follows:
\begin{equation}
\label{eq:func_model}
    \begin{aligned}
        \mathbb{E}[Y_i] &= a + b \left(t_{i} - t_0\right) +  \sum_{h=1}^{m}\Big\{c_{h} \sin \left(2 \pi f_{h} t_{i}\right)  + d_h \cos \left(2 \pi f_h t_i\right)\Big\} \\& \;\;\;\;+ \sum_{j=1}^{r}e_j H\left(t_i - t_j\right) + \sum_{k = 1 }^{s} l_k \left\{1- \exp\left(\frac{-(t_i-t_k)}{\tau_k}\right)\right\}H\left(t_{i}-t_k\right)\text{,}
    \end{aligned}
\end{equation}
for $i \in \{1, \ldots, n\}$, where $a$ is the initial position at the reference epoch $t_0$, $b$ is the trend (or velocity) parameter and $c_h$ and $d_h$ are the periodic motion parameters. The common choice for $m$ (the number of periodic components) is $2$ thereby delivering $f_1 = \nicefrac{1}{365.25}$ and $f_2 = \nicefrac{2}{365.25}$ as the annual and semiannual seasonal terms respectively. The coefficients $e_j$ represent the magnitude of the $r$ offset terms (step change in the mean function) that model earthquakes, equipment changes or human interventions, where $H(\cdot)$ is the Heaviside step function defined as $H(x) = \mathbbm{1}(x\geq 0)$, and with $\mathbbm{1}(\cdot)$ being the indicator function. The coefficients $l_k$ model the magnitude of the exponential decay function which itself models post-seismic relaxations (e.g., \citealp{shen1994postseismic, ren2021analysis}), where $t_k$ and $\tau_k$ represent respectively the start time of the $k^{\text{th}}$ post-seismic relaxation and its relaxation time (these are both user-provided). In our case, the start time of each post-seismic relaxation associated with each station are provided by the NGL and we estimate the amplitude of the offsets and the amplitude of the exponential decay function used to model post-seismic relaxations. The relaxation time $\tau_k$ is fixed to $1$ year, in line with the findings of \cite{ren2021analysis} who highlight that the seismic relaxation time does not significantly affect the overall trajectory or variation characteristics of individual GNSS coordinate time series. Hence, the model assumed for these data takes on the standard linear regression form as follows:
\begin{equation}
    \Y = \X \boldsymbol{\beta} + \boldsymbol{\varepsilon}, \quad \boldsymbol{\varepsilon} \sim \mathcal{N}\left\{\0, \boldsymbol{\Sigma}\left(\boldsymbol{\gamma}\right)\right\},
    \label{eq.reg_complete}
\end{equation}
where $\X \in \real^{n \times p}$ is a design matrix of observed/known predictors described earlier and $\boldsymbol{\beta} \in \real^p$ is the regression parameter vector corresponding to: $\boldsymbol{\beta}=\left(a, b, c_1, \ldots , c_m, d_1, \ldots, d_m, e_1, \ldots, e_r, l_1, \ldots, l_s\right)^{\trans}$. In addition to the deterministic component, $\boldsymbol{\varepsilon} = (\varepsilon_{1}, \dots, \varepsilon_n)^{\trans}$ represents a zero-mean error process following a Gaussian distribution $\mathcal{N}$ with positive-definite covariance function $\boldsymbol{\Sigma}(\boldsymbol{\gamma}) \in \real^{n\times n}$ characterizing the second-order dependence structure of the process (induced by measurement errors and unobserved variables) and parameterized by the vector $\boldsymbol{\gamma} \in \boldsymbol{\Gamma} \subset \real^q$. The Gaussian assumption is common in the analysis of this data, however we must underline that many results for our proposed method do not require Gaussianity and remain valid under a broad class of non-Gaussian distributions. Given the above, the goal is therefore to adequately estimate $\boldsymbol{\beta}$ as well as $\boldsymbol{\gamma}$, the latter being needed to perform inference on the former. More specifically, in the above setting we know that to perform inference on $\bm{\beta}$ we make use of the estimated parameter vector $\hat{\bm{\beta}}$ as well as its corresponding covariance matrix. In this work we will employ least-squares regression to estimate $\bm{\beta}$ and use the notation $\bm{\Phi}_n$ for its covariance matrix to underline its dependence on the design matrix $\X$ and therefore on the sample size $n$. Since this design matrix $\X$ is fixed (known) for the modeling of GNSS stations, the consequent structure of the deterministic (functional) component in \eqref{eq:func_model} is generally unproblematic since it is defined via domain knowledge. However, the challenge mainly comes from the choice of the stochastic model for the error process $\boldsymbol{\varepsilon}$ and the consequent estimation of its parameter $\bm{\gamma}$ which itself impacts the estimation of $\bm{\Phi}_n$. Indeed, the latter is used to create CIs for $\bm{\beta}$ and the importance of correctly identifying and estimating $\bm{\gamma}$ for this purpose is widely supported by the literature which highlights how neglecting long-range dependence noise in the error process, for example, can lead to an underestimation of uncertainties on $\hat{\bm{\beta}}$ by factors ranging from $5$ to $11$ if not appropriately addressed \citep{mao1999noise, wang2012noise}. Considering this, in the following sections we provide an overview of limitations of current methodologies to estimate $\bm{\beta}$, and more importantly $\bm{\gamma}$, thereby motivating the use of the proposed methodology for our case study of interest.
\subsection{Existing Methodologies}
\label{sec:related_work}

While the MLE is the standard methodology for framework described in Section~\ref{sec:framework}, estimating the parameters $\boldsymbol{\beta}$ and $\boldsymbol{\gamma}$ by directly maximizing the likelihood function is commonly computationally demanding with respect to both time and resources (see e.g.,~\citealp{guerrier2013wavelet,proietti2022modelling}). For example, in the analysis of GNSS time series, most likelihood-based methods have a computational complexity of $\mathcal{O}(n^\delta)$, where $\delta \in [2, 3]$, depending among others on the degree and structure of missing data as well as the stochastic model (see e.g.,~\citealp{bos2013fast}). There have been improvements in this direction such as the restricted MLE in \cite{tehranchi2021fast} or faster approximations in \cite{Bos2008, bos2013fast} which can reduce computation times by factors of 10 to 100 compared to traditional methods \citep{williams2008cats}. However, while some methods can accommodate missing observations, they generally do not explicitly model the missing-data mechanism and the challenges of parameter estimation remain significant due to matrix storage limitations and complex noise structures which are very common for the study of crustal deformations \citep[see][]{tunini2024global}.

To address these challenges, several strategies have been proposed. Some methods sparsify the covariance matrix by partitioning the domain or modifying correlations to induce near-independence \citep{fuentes2002spectral, sang2011covariance}, while others impose sparsity in the precision matrix through conditional modeling or spectral approaches \citep{stein2013stochastic, guinness2019spectral}. Variational low-rank approximations and predictive processes reduce computational cost by using a smaller set of knots or basis functions \citep{gibbs2000variational, lazaro2011variational}, and distribution-free approaches such as Gapfill rely on quantile regression \citep{gerber2018predicting}. A benchmark study by \cite{heaton2019case} highlights that, while these approaches can achieve substantial computational gains, they often do so at the expense of accuracy or well-characterized statistical guarantees. As a consequence, existing methods generally fail to simultaneously provide scalability and reliable inference, particularly in the presence of long-range dependence and missing data (two common characteristics of GNSS time series). Moreover, extending these approaches to include model selection for the stochastic component remains computationally prohibitive in practice \citep{he2017review}, further limiting their applicability to large-scale analyses. Alternative step-wise procedures are also available. A common approach is Feasible Generalized Least Squares, where an initial estimate of $\boldsymbol{\beta}$ is obtained, the residuals are used to estimate the covariance parameters $\boldsymbol{\gamma}$, and the resulting covariance estimate is then used to update $\boldsymbol{\beta}$. More generally, moment-based variants of this strategy replace likelihood-based covariance estimation by moment-matching procedures, following the same philosophy as the Generalized Method of Moments (GMM) framework of \cite{hansen1982large}. Following this approach, \cite{cucci2023generalized} proposed a two-step method combining least-squares regression with the GMWM for large-scale GNSS data. However, this framework suffers from several important limitations: (i) it treats the error process as observed, leading to biased estimation and inference, (ii) it relies on unrealistic assumptions on the missing data mechanism, and (iii) it cannot accommodate non-stationary noise processes, which are common in GNSS applications over the North American region. In addition, it provides no theoretical guarantees for valid inference under realistic settings, as it assumes the stochastic model to be known. In contrast, the methodology proposed in this work is specifically designed to overcome these limitations by jointly addressing missing observations, complex stochastic structures and large-scale computations (an overview of these contributions and their supporting results is provided in Appendix~\ref{sec.appendix_structure}).

\section{Wavelet Moment Regression}
\label{sec:methodology}

Given the interest in inference on $\bm{\beta}$, and therefore on the error process $\bm{\varepsilon}$ (i.e., the stochastic component), in this section we firstly put forward a new modeling approach to adequately address the missing data mechanism (including assumptions and notation). We then present the proposed WAMORE framework which builds on this approach to deliver scalable inferential tools for constructing CIs and model selection (among others).

\subsection{New Modeling Approach with Missing Data Mechanism}
\label{sec:missing_data}

Recalling the modeling framework defined in \eqref{eq.reg_complete}, we address the missing data process by equipping the response of interest $\Y$ (i.e., a position time series of a station for a given axis) with a random variable $\Z = (Z_{1}, \dots, Z_n)^{\trans}$ which describes this mechanism. In particular, we assume $\Z$ to be a binary-valued stationary process \textit{independent} of $\Y$ defined as
\begin{equation}
\label{eq:markov_missinigness_in_text}
  Z_i =
\begin{cases}
    Z_{i-1}, & \text{with probability } \rho, \\
    W_i \sim \mathrm{Bernoulli}(\mu(\bm{\vartheta})), & \text{with probability } 1-\rho,
\end{cases}
\end{equation}
with $\rho \in [0,1)$ such that its expectation is $\mu(\boldsymbol{\vartheta}) = \mathbb{E}[Z_i] \in (0, \, 1]$, for all $i$, with covariance matrix $\boldsymbol{\Lambda}(\boldsymbol{\vartheta}) = \var\left(\Z\right) \in \real^{n\times n}$  whose structure is assumed known up to the parameter vector $\boldsymbol{\vartheta} \in \boldsymbol{\Upsilon} \subset \real^k$. Following this, we assume that we only observe $\tilde{\Y} = \Z \odot \Y$, where $\odot$ denotes the element-wise product. Considering the missing mechanism as \textit{random and independent} of $\Y$ is reasonable because we generally do not know when data will be observed in advance, since most of the time data are not observed due to mechanical or software failure in the system. Moreover, given the dependence over time, the presence/absence of data will be conditional on the presence/absence of previous observations. Finally, the stationarity of this process is assumed plausible since, in practice, human intervention is often required to fix potential issues with corresponding equipment and practices for these procedures remaining relatively stable over time. Even in the case where such assumptions were (slightly) violated, we consider this approach to be closer to reality than those underlying current alternative strategies. Nevertheless, other missing data mechanisms could easily be adopted with minimal changes to the proposed procedure. With this in mind, $\Z$ therefore determines which elements of $\Y$ are observed (i.e., ~when $Z_i = 1$). As a result of the missing observations mechanism, the (overall) observed process $\tilde{\Y}$ is now parametrized by the vector:
$$\boldsymbol{\theta} = \left(\boldsymbol{\beta}^{\trans}, \boldsymbol{\gamma}^{\trans}, \boldsymbol{\vartheta}^{\trans}\right)^{\trans} \in \boldsymbol{\Theta},$$
where $\boldsymbol{\Theta} = \real^p \times \boldsymbol{\Gamma} \times \boldsymbol{\Upsilon}$. This new parametrization implies that we can also define $\tilde{\X}=\left(\Z \otimes \bm{1}^{\trans}\right) \odot \X \in \real^{n \times p}$ as the design matrix $\X$ with zero-valued vectors for the rows where observations are missing in $\Y$, where $\otimes$ to denotes the Kronecker product and $\bm{1}$ represents a vector of ones of dimension $p$. It is important to note that, following the model in \eqref{eq.reg_complete}, throughout this work we will assume that all rows in $\X$ are observed (or can be deterministically derived/computed) implying that, as defined earlier, the missing data mechanism $\Z$ is solely associated to $\bm{\varepsilon}$ and hence $\Y$ (therefore not to $\X$). Following this, the missing data mechanism does not have a sizeable impact on the estimator of $\bm{\beta}$, defined as
\begin{equation}
\label{eq.beta_hat_miss}
    \hat{\boldsymbol{\beta}} = \left(\tilde{\X}^{\trans} \tilde{\X}\right)^{-1} \tilde{\X}^{\trans} \tilde{\Y},
\end{equation}
except in terms of statistical efficiency since the number of observations is reduced. Therefore, hereinafter the estimator $\hat{\boldsymbol{\beta}}$ will refer to the definition given in \eqref{eq.beta_hat_miss}. This being said, the impact of the missing data mechanism is instead significant on the covariance matrix of $\hat{\bm{\beta}}$ which, as shown in Appendix \ref{app:cov_beta}, is now given by:
\begin{equation}
\label{eq.beta_cov_miss}
    \bm{\Phi}_n = \mu(\bm{\vartheta})^{-2}(\X^{\trans}\X)^{-1} \X^{\trans} \bm{\Sigma}(\bm{\gamma}, \bm{\vartheta}) \X (\X^{\trans}\X)^{-1},
\end{equation}
where $\bm{\Sigma}(\bm{\gamma}, \bm{\vartheta}) = \left\{\bm{\Lambda}(\bm{\vartheta}) + \mu(\bm{\vartheta})^2 \bm{1}\bm{1}^{\trans}\right\} \odot \bm{\Sigma}(\bm{\gamma})$. Similarly to $\hat{\bm{\beta}}$, hereinafter $\bm{\Phi}_n$ will therefore refer to the matrix defined in \eqref{eq.beta_cov_miss}. Notice that in the analysis of GNSS data, as mentioned earlier, the design matrix $\X$ is commonly known and, while only the rows for the observed values in $\Y$ can be used to estimate $\hat{\bm{\beta}}$, the full design can be used in \eqref{eq.beta_cov_miss}, with $\tilde{\Y}$ being used to estimate $\bm{\vartheta}$ and $\bm{\gamma}$ as discussed in the next section.

\subsection{Inference}
\label{sec:inference}

While the quantities to obtain the estimate $\hat{\bm{\beta}}$ are directly observable in \eqref{eq.beta_hat_miss}, to perform inference on $\bm{\beta}$ one needs its covariance matrix $\bm{\Phi}_n$ which depends on two unknown quantities as mentioned earlier: $\bm{\vartheta}$ and $\bm{\gamma}$. Given the independence between $\Z$ and $\bm{\varepsilon}$ (and hence $\Y$) we can firstly estimate the missing data mechanism parameter $\bm{\vartheta}$ for which many statistically valid approaches already exist \citep[see e.g.,~][and references therein]{chib2001markov}. More specifically, in this work we will employ the MLE for the estimation of Markov processes. Once the estimate $\hat{\bm{\vartheta}}$ is available, this can be plugged into \eqref{eq.beta_cov_miss} as well as into other quantities that depend on this parameter. However the estimation of the parameter characterizing the stochastic error model, i.e., $\bm{\gamma}$, is generally the most computationally challenging component to obtain. For this reason we build upon the GMWM which was put forward in \cite{guerrier2013wavelet} and later developed in different directions \citep[see e.g.,][]{guerrier2014estimation, guerrier2015automatic, xu2019multivariate, guerrier2022robust, voirol2023accounting}. More specifically, with the goal of obtaining statistically valid estimates of the stochastic parameter $\bm{\gamma}$, the GMWM is defined as follows:
\begin{equation}
    \hat{\bm{\gamma}} = \argmin_{\bm{\gamma} \in \bm{\Gamma}} \|\hat{\bm{\nu}} - \bm{\nu}(\bm{\gamma})\|_{\bm{\Omega}}^2,
    \label{eq:gmwm}
\end{equation}
where (i) $\|\mathbf{a}\|_{\mathbf{B}}^2 = \mathbf{a}^{\trans}\mathbf{B}\mathbf{a}$, (ii) $\hat{\bm{\nu}} = [\hat{\nu}_j]_{j = 1, \hdots, J} \in \real_+^J$ is the estimated Wavelet Variance (WV) vector from the error process (with $J < \log_2(n)$ being the scales of wavelet decomposition, see e.g., \citealp{serroukh2000statistical}), (iii) $\bm{\nu}(\bm{\gamma}) = [\nu_j(\bm{\gamma})]_{j = 1, \hdots, J} \in \real_+^J$ is the theoretical WV vector implied by the assumed stochastic model for the process, and (iv) $\bm{\Omega}$ is any positive-definite weighting matrix whose choice is made to achieve statistical efficiency of the estimator $\hat{\bm{\gamma}}$. The (asymptotically) optimal choice for this weighting matrix is given by $\bm{\Omega} = \V^{-1}$, where $\V = \var\left(\hat{\bm{\nu}}\right)$ is the covariance matrix of the WV estimator \citep[see][]{hansen1982large, guerrier2013wavelet}.

In order to make use of the GMWM within the setting of this work, we however need to integrate \eqref{eq:gmwm} with the missing data mechanism parametrized by $\bm{\vartheta}$. In particular, the quantity that requires this integration is the theoretical WV since this is a function of the covariance matrix $\bm{\Sigma}(\bm{\gamma}, \bm{\vartheta})$. However, existing forms of the theoretical WV are available only for a restricted class of time series models and a general form of this function does not exist to enable the use of a broader class of models such as those considered in GNSS signals. To address this limitation, we derive a general expression for the theoretical WV at scale $j$, denoted as $\nu_j$, that holds for any zero-mean process with finite covariance matrix $\boldsymbol{\Sigma} \in \real^{n\times n}$, i.e., $\nu_j = \mathrm{tr}\!\left\{\mathbf{A}_j \boldsymbol{\Sigma} \right\}$. The proof of this result and the form of the matrix $\mathbf{A}_j$ are given in Appendix~\ref{app:wv_theo_forms}. It must be noted that this explicit form can be of considerable importance beyond the scope of this work since it can be used in many other natural and economic science applications. A potential issue with its computation however are the matrix operations involved in this form, and we therefore deliver a computationally efficient $\mathcal{O}(n^{-1})$-accurate approximation to compute $\nu_j$ in Appendix~\ref{app:fast_theo_wv}.
In our setting, we consider the error process with missing data $\tilde{\boldsymbol{\varepsilon}} = \tilde{\mathbf{Y}} - \tilde{\mathbf{X}}\boldsymbol{\beta} = \boldsymbol{\varepsilon} \odot \mathbf{Z}$ whose theoretical WV can therefore be written as
\begin{equation}
    \nu_j\left(\boldsymbol{\gamma}, \boldsymbol{\vartheta} \right) = \mathrm{tr}\!\left\{\mathbf{A}_j \boldsymbol{\Sigma}(\boldsymbol{\gamma}, \boldsymbol{\vartheta}) \right\}.
    \label{eq:wv_missing}
\end{equation}
 This being said, if we define $\bm{\nu}(\bm{\gamma}, \bm{\vartheta}) = [\nu_j(\bm{\gamma}, \bm{\vartheta})]_{j = 1, \hdots, J}$ and use an appropriate estimator for $\bm{\vartheta}$, we can consequently define the estimator for $\boldsymbol{\gamma}$ with missing observations as follows:
\begin{equation}
    \hat{\bm{\gamma}} = \argmin_{\bm{\gamma} \in \bm{\Gamma}} \,\|\hat{\bm{\nu}} - \bm{\nu}(\bm{\gamma}, \hat{\bm{\vartheta}})\|_{\bm{\Omega}}^2,
    \label{eq:gmwm_new}
\end{equation}
where $\hat{\bm{\nu}}$ is the estimated WV computed on the estimated residuals $\hat{\bm{\varepsilon}}$ and $\hat{\bm{\vartheta}}$ is the pre-computed estimator of the parameter $\bm{\vartheta}$, with $\bm{\Omega}$ being any positive-definite matrix (e.g., identity). It must be noted at this point that the empirical WV $\hat{\bm{\nu}}$ is computed on the \textit{estimated} error (residual) process $\hat{\boldsymbol{\varepsilon}} = {\tilde{\Y}} -\tilde{\X} \hat{\boldsymbol{\beta}}$ which is not a direct realization of the true error process with missingness, $\tilde{\boldsymbol{\varepsilon}}$. For this reason, we denote the covariance matrix of the estimated error process as $\bm{\Sigma}_{\hat{\bm{\varepsilon}}}(\bm{\gamma}, \bm{\vartheta})$ which is therefore different from $\bm{\Sigma}(\bm{\gamma}, \bm{\vartheta})$ in finite samples. Considering the scale of GNSS data, a first option would consist in ignoring the estimated nature of this error process and make a large sample approximation by assuming $\bm{\Sigma}_{\hat{\bm{\varepsilon}}}(\bm{\gamma}, \bm{\vartheta}) \approx \bm{\Sigma}(\bm{\gamma}, \bm{\vartheta})$. However, since the latter approximation can induce large biases with consequent negative impacts on inference, we propose a computationally efficient finite-sample approximation of $\bm{\Sigma}_{\hat{\bm{\varepsilon}}}(\bm{\gamma}, \bm{\vartheta})$ in Appendix~\ref{app:IminusH} and study its performance in Section \ref{sec:simulations}. This approximation is plugged into \eqref{eq:wv_missing} instead of $\bm{\Sigma}(\bm{\gamma}, \bm{\vartheta})$ to estimate $\bm{\gamma}$ in \eqref{eq:gmwm_new}, following which we can obtain the two parameters needed to estimate the required covariance matrix $\bm{\Phi}_n$. With the latter, we now have all the quantities needed to estimate and perform inference on $\bm{\beta}$ and the entire pipeline for the analysis is summarized in Algorithm \ref{algo:gmwmx}.

\begin{algorithm}[H]
\caption{Wavelet Moment Regression (WAMORE)}
\begin{algorithmic}[1]

\vspace{0.4em}
\State \textbf{Input:} Complete design matrix $\X$; response vector with missing observations $\tilde{\Y}$; missingness indicator vector $\Z$ 

\vspace{1em}
\State Obtain $\tilde{\X}= \left(\Z \otimes \mathbf{1}^{\trans}\right) \odot \X$
\State  Compute $ \hat{\boldsymbol{\beta}} = (\tilde{\X}^{\trans} \tilde{\X})^{-1} \tilde{\X}^{\trans} \tilde{\Y}$

\vspace{0.4em}
\State Compute  $\hat{\boldsymbol{\varepsilon}} = {\tilde{\Y}} -\tilde{\X} \hat{\boldsymbol{\beta}}$

\vspace{0.4em}
\State Compute $\hat{\boldsymbol{\vartheta}}$ based on $\Z$ and obtain $\mu(\hat{\boldsymbol{\vartheta}})$ and $\boldsymbol{\Lambda}(\boldsymbol{\hat{\vartheta}})$

\vspace{0.4em}
\State Compute $\hat{\bm{\gamma}}$ through \eqref{eq:gmwm_new} by using the covariance matrix $\bm{\Sigma}_{\hat{\bm{\varepsilon}}}(\bm{\gamma}, \bm{\vartheta})$ and a diagonal $\bm{\Omega}$

\vspace{0.4em}
\State \textit{Optional}: Use $\hat{\bm{\gamma}}$ to compute $\hat{\V}$ (see Eq. \eqref{eq:var_wv} later) and define $\bm{\Omega} = \hat{\V}^{-1}$. Then repeat step 6

\vspace{0.4em}
\State Compute $\hat{ \boldsymbol{\Phi}}_n  = \mu(\hat{\bm{\vartheta}})^{-2}(\X^{\trans}\X)^{-1} \X^{\trans} \bm{\Sigma}(\hat{\bm{\gamma}}, \hat{\bm{\vartheta}})\X (\X^{\trans}\X)^{-1}$

\vspace{1em}
\State \textbf{Output:} Estimated parameter vector $\hat{\bm{\beta}}$ and its estimated covariance matrix $\hat{\bm{\Phi}}_n$

\end{algorithmic}
\label{algo:gmwmx}
\end{algorithm}

To ensure validity of the inferential procedure using estimates obtained from Algorithm~\ref{algo:gmwmx}, we need to obtain the statistical properties of the quantities of interest, starting from the asymptotic distribution of $\hat{\bm{\beta}}$, under different memory/dependence regimes (common in the signals coming from crustal deformations), to the consistency of $\hat{\bm{\gamma}}$ which is then used in $\bm{\Phi}_n$ for inference on $\bm{\beta}$. In particular, Appendices \ref{app:short_memory} and \ref{app:long_memory} deliver the asymptotic distributions of $\hat{\bm{\beta}}$ under short- and long-term memory regimes respectively, while Appendix~\ref{app:consistency_nu_hat} proves the consistency of the WV estimator $\hat{\bm{\nu}}$ under these same regimes, which is then used in Appendix \ref{app:consistency_of_gamma_hat} to deliver the consistency of $\hat{\bm{\gamma}}$. In addition to these theoretical properties, in Appendix~\ref{app:fast_theo_wv} we put forward computational approaches to efficiently compute the theoretical WV in \eqref{eq:wv_missing} as well as a computationally efficient approximation of $\bm{\Sigma}_{\hat{\bm{\varepsilon}}}(\hat{\bm{\gamma}}, \hat{\bm{\vartheta}})$ which can be found in Appendix \ref{app:IminusH}.

\subsection{Model Selection}
\label{sec:model_selection}

The inferential procedure presented in the previous section assumes that the model (covariance structure) for the stochastic error component, parametrized by $\bm{\gamma}$, is \textit{known}. While the models for the error processes can indeed usually be determined rather accurately by experts and practitioners, any misspecification of the error model can however lead to an underestimation of parameter uncertainties by factors ranging from $5$ to $11$ \citep{mao1999noise, wang2012noise}. There is consequently a strong need to accurately identify the stochastic model but, since current inferential methods for the error process rely either on spectral decompositions (in the frequency domain) or likelihood-based approaches, the associated model selection tools cannot reliably capture the lowest-frequency components \citep{langbein1997correlated,zhang1997southern} or are particularly computationally demanding due to the need to compute the MLE for each considered stochastic model \citep{Bos2008, amiri2007assessment,bos2013fast}. 

We therefore put forward a model selection criterion built on the WAMORE framework. More specifically, following the covariance penalty selection in \cite{efron2004estimation}, we consider a model selection criterion based on an out-of-sample expectation $\mathbb{E}_0[\cdot]$ of the GMWM loss as follows:
\begin{equation}
\label{def:model_selection_criterion}
    C = \mathbb{E}_0\left[ \mathbb{E}\left\{\left\|\hat{\bm{\nu}}_0 - \bm{\nu}(\hat{\bm{\gamma}}, \hat{\boldsymbol{\vartheta}}) \right\|_{\bm{\Omega}}^2 \right\}\right],
\end{equation}
where $\hat{\bm{\gamma}}$ and $\hat{\bm{\vartheta}}$ are computed on the estimated residuals $\hat{\bm{\varepsilon}}$ and $\hat{\bm{\nu}}_0$ is computed on $\hat{\bm{\varepsilon}}_0$ (an independent copy of $\hat{\bm{\varepsilon}}$ satisfying $\hat{\bm{\varepsilon}} \overset{d}{=} \hat{\bm{\varepsilon}}_0$, where $\overset{d}{=}$ denotes equality in distribution). This criterion is similar to Mallow's $C_p$ developed for least-squares estimation \citep{mallows2000some} and similar ideas are also used by \cite{andrews1999consistent} and \cite{hansen1982large} where goodness-of-fit tests (e.g.,~Sargan-Hansen test or
$J$-test) provide model selection criteria for GMM frameworks. Indeed, computationally-intensive approximations to directly estimate this criterion had already been considered for the GMWM framework in \cite{guerrier2015automatic} and in \cite{radi2019multisignal} where they compute the outer expectation $\mathbb{E}_0[\cdot]$ via a cross-validation approach. To avoid a computationally expensive cross-validation approach, we put forward an explicit plug-in estimator for this criterion which, up to an additive constant and assuming that the missing-data mechanism parametrized by $\boldsymbol{\vartheta}$ is kept fixed across all candidate models, is given by
\begin{equation}
\label{eq.mod_crit}
    \hat{C} = \left\|\hat{\bnu} - \bnu(\hat{\bgamma},\hat{\boldsymbol{\vartheta}}) \right\|_{\hat{\mathbf{V}}^{-1}}^2 + 2q,
    \end{equation}
where $\hat{\V}$ is a consistent estimator for $\V$. The derivation of this estimator can be found in Appendix \ref{app:model_selection}. 

In this perspective, the only quantity missing to compute the estimator in \eqref{eq.mod_crit} is $\hat{\V}$ (the estimator of the covariance matrix of the observed WV $\hat{\bm{\nu}}$). To date there are not even explicit theoretical forms for this covariance matrix. However, adopting the same approach with which we put forward the theoretical WV in \eqref{eq:wv_missing}, we employ the \textit{variance of quadratic forms} (of Gaussian processes) to deliver a compact explicit form for the $(j,l)^{\text{th}}$ element of the covariance matrix $\V$ of the empirical WV computed on a zero-mean Gaussian process with covariance matrix $\boldsymbol{\Sigma}\in \real^{n \times n}$ which is given by:
\begin{equation}
\label{eq:var_wv}
v_{j,l} = 2 \tr\left\{\mathbf{A}_j \boldsymbol{\Sigma}\mathbf{A}_l \boldsymbol{\Sigma} \right\},
\end{equation}
where the matrix $\mathbf{A}_j$ is the same as in \eqref{eq:wv_missing} and defined in Appendix~\ref{app:wv_theo_forms}, where the derivation of this result is also provided. This is the \textit{first explicit form} of the covariance matrix of the empirical WV computed on Gaussian processes: this result is of significant importance also for applications and statistical inference beyond the scope of this work (see Remark \ref{remark.v} further on). In our setting, we consider the error process with missing data $\boldsymbol{\varepsilon}\odot \Z$ which therefore allows us to define these quantities as $v_{j,l} = 2 \tr\left\{\mathbf{A}_j \boldsymbol{\Sigma}(\boldsymbol{\gamma}, \boldsymbol{\vartheta})\mathbf{A}_l \boldsymbol{\Sigma}(\boldsymbol{\gamma}, \boldsymbol{\vartheta}) \right\}$. To estimate the elements of this matrix we therefore simply need to plug-in a consistent estimator for $\bm{\gamma}$ which is already given in \eqref{eq:gmwm_new} and of $\boldsymbol{\vartheta}$ that can be easily estimated. Although the process with missing observations is not Gaussian, we exploit the covariance structure of the error process with missing data through the corresponding covariance expression. While a direct evaluation of \eqref{eq:var_wv} would require the multiplication of large matrices and quickly become computationally prohibitive, this representation allows us to leverage the computational properties of the wavelet decomposition (e.g., the pyramid algorithm, \citealp{percival2000wavelet}) to derive an exact and computationally efficient algorithm for computing $\hat{\mathbf V}$ without explicitly forming these matrices. The details of this algorithm are provided in Appendix~\ref{app:compute_v}.

\begin{Remark}
\label{remark.v}
    The proposal of an efficient and consistent estimator of $\mathbf{V}$ is not only important for the case study of interest in this work, but is also of relevance for all procedures that employ the WV to perform statistical inference in different settings. For example, the WV is broadly employed within the natural and physical sciences for model building and prediction \citep[see e.g.,][for an overview]{percival2000wavelet}, and within other fields of research \citep[see e.g.,][to mention a few]{gallegati2012wavelet, xie2013wavelet,foufoula2014wavelets, jia2015correlations, ziaja2016fault, abry2018wavelet}. Other examples include inferential approaches for unit-root tests \citep{fan2010unit}, isotropy in random fields \citep{thon2014multiscale} or Portmanteau tests \citep{gencay2015multi}. Moreover, by plugging-in an initial estimate of $\bm{\gamma}$ we can obtain $\hat{\V}$ allowing us to choose $\hat{\bm{\Omega}} = \hat{\V}^{-1}$. The latter can then be used in a second run of \eqref{eq:gmwm_new} to obtain a more statistically efficient estimator for $\bm{\gamma}$ (see step 7 in Algorithm \ref{algo:gmwmx}).
\end{Remark}

Now that we have an efficient way of computing $\V$ from $\boldsymbol{\Sigma}\left(\boldsymbol{\gamma}, \bm \vartheta \right)$, this can be plugged into the proposed estimator in \eqref{eq.mod_crit}. 
In Appendix~\ref{app:model_selection} we show the consistency of the proposed estimator for the criterion in \eqref{def:model_selection_criterion} and discuss its overfitting and underfitting behavior and connections to the AIC. Importantly, the estimator for \eqref{def:model_selection_criterion} can be computed through the quantities $\hat{\bm{\nu}}$, $\hat{\V}$ and $\hat{\bm{\Omega}} = \hat{\V}^{-1}$, without relying on the likelihood function. While $\hat{\bm{\nu}}$ is computed directly from the data, $\hat{\V}$ and $\hat{\bm{\Omega}}$ are estimated only once using the largest candidate stochastic model, providing a low-bias estimate of the covariance structure. As a consequence, the only quantity requiring an update when evaluating the criterion is $\hat{\bm{\gamma}}$, which is obtained via the GMWM based on these same pre-computed quantities. This implies that the computational complexity of the selection procedure only depends on the evaluation of $\bm{\nu}\left(\bm{\gamma}, \hat{\bm{\vartheta}}\right)$ which relies on the computational efficiency of the procedure proposed in Appendix~\ref{app:fast_theo_wv}. Moreover, in Appendix~\ref{app:model_selection_simulation}, we present a simulation study comparing the performance of the model selection criterion presented in this section with the AIC implemented in the \texttt{Hector} software \citep{bos2013fast} for the task of selecting the appropriate model for the stochastic error of a simulated GNSS time series.

\section{Simulations}
\label{sec:simulations}

We evaluate the proposed inferential framework against the MLE approximation implemented in \texttt{Hector} v2.0 \citep{Bos2008} under varying conditions of (i) time series length, (ii) stochastic model, and (iii) proportion of missing data. Specifically, we consider three noise models: white noise with stationary power-law noise (WN + PL), white noise with flicker noise (WN + FL), and white noise with a Matérn process and a random walk (WN + MAT + RW). The first two correspond to commonly used models for GNSS time series \citep{zhang1997southern, calais1999continuous, bock2000instantaneous}, while the latter represents a more complex but practically relevant setting \citep{kermarrec2014matern}. For all scenarios, we adopt a simplified deterministic model including an intercept, a linear trend, and annual and semi-annual components, with parameters calibrated from the empirical distributions observed in the case study of Section~\ref{sec:case_study}. The results for two additional noise settings are given in Appendix~\ref{app:main_simu_results}.

For each stochastic model, we generate synthetic daily time series with lengths ranging from $10$ to $30$ years and introduce missing data via a Markov process (see Appendix~\ref{app:main_simu_results}), yielding missingness proportions between $0\%$ and $40\%$. For each configuration, $1{,}000$ signals are simulated and confidence intervals for the trend parameter are computed using both WAMORE and \texttt{Hector}, recording empirical coverage and computational time. For settings where the MLE takes more than 2 hours and 20 minutes per signal on average, we evaluate the MLE only on $5$ Monte Carlo realizations due to computational constraints. The results in Figures~\ref{fig:wn_matern_rw_running_time} and \ref{fig:wn_matern_rw_type_1_error} show that WAMORE achieves substantial computational gains while maintaining, or improving upon, the inferential performance of the MLE across all considered settings.

\begin{figure}[]
    \centering
    \includegraphics[width=1\linewidth]{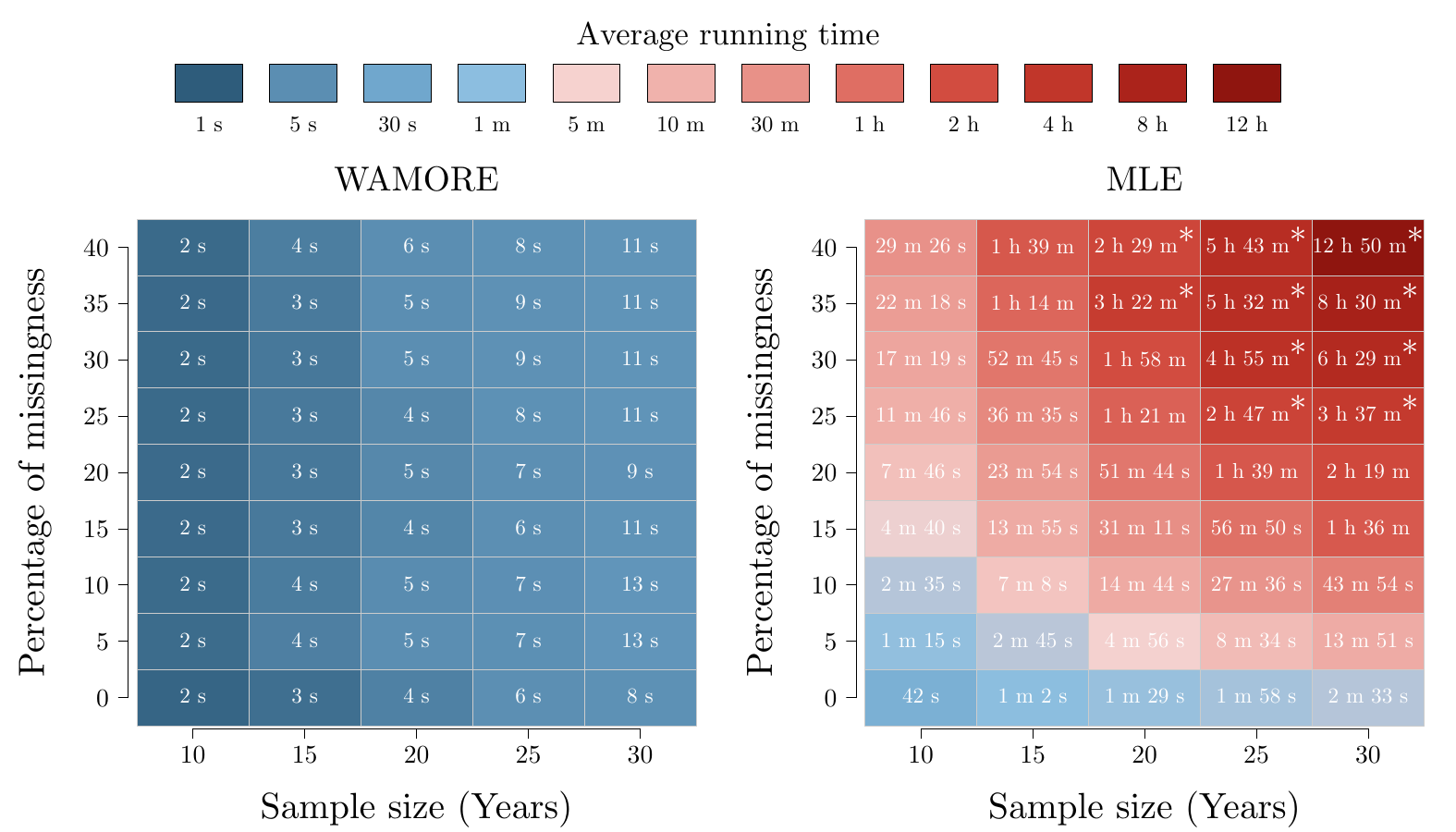}
    \caption{Average running time of the WAMORE (left panel) vs the MLE (right panel) for the setting with the WN + MAT + RW model. Settings annotated with a ``$\raisebox{-0.9ex}{\scalebox{1.3}{*}}$'' were evaluated only on $5$ Monte Carlo realizations due to computational constraints.}
    \label{fig:wn_matern_rw_running_time}
\end{figure}
\begin{figure}[]
    \centering
    \includegraphics[width=1\linewidth]{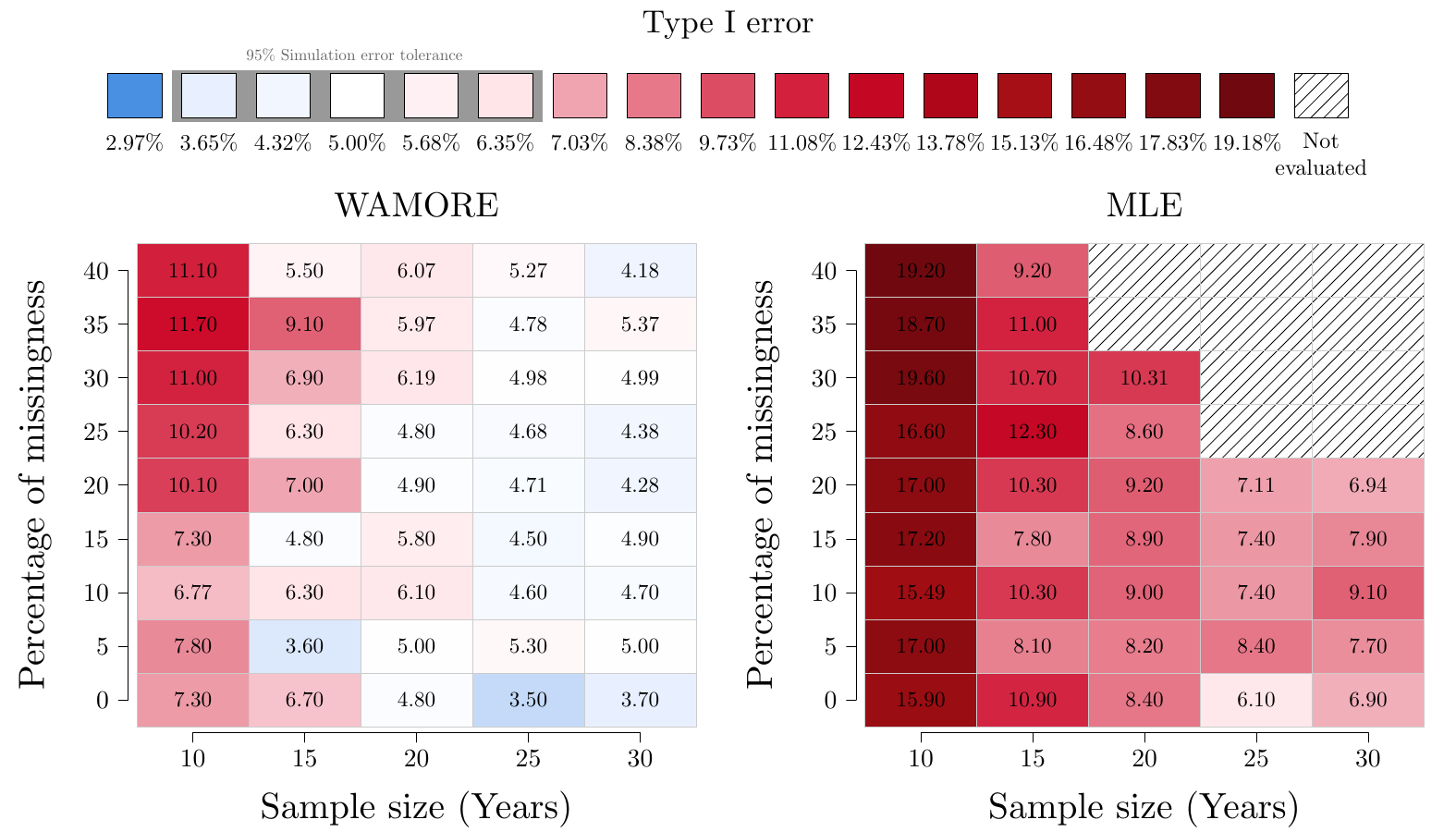}
    \caption{Empirical type I error (in \%) of the WAMORE (left panel) vs the MLE (right panel) for the setting with the WN + MAT + RW model.}
    \label{fig:wn_matern_rw_type_1_error}
\end{figure}

We additionally assess the model selection criterion proposed in Section~\ref{sec:model_selection} through a dedicated simulation study (see Appendix~\ref{app:model_selection_simulation} for details). Considering a WN + FL baseline, we evaluate the ability to detect an additional random walk component which is a common challenge in GNSS analysis \citep{kaczmarek2018identification, he2019investigation, he2021analysis}. Across $1{,}000$ simulations per setting and varying the random walk variance, both the proposed criterion and the AIC in \texttt{Hector} exhibit similar selection performance for time series of $20$ and $40$ years. However, the proposed method is significantly faster, with a median runtime of about $4$ minutes compared to $29$ minutes for the AIC on $40$-year series.

\FloatBarrier
\section{Velocity of Crustal Deformations in North America}
\label{sec:case_study}

We now employ our proposed WAMORE methodology to study velocity of crustal deformations based on the data described in Section \ref{sec:data}. As noted before, while the structure of the deterministic (functional) component in \eqref{eq:func_model} is generally known, the uncertainty mainly comes from the choice of the stochastic model for the error process $\boldsymbol{\varepsilon}$. For this reason, for each signal we consider six stochastic models to select from, each consisting of white noise (WN) combined with either (i) a flicker noise (FN), (ii) a stationary power-law process (PL) or (iii) a Matérn process (MAT), as well as the same three models with an additional random walk (RW) component. These correspond to the most commonly considered stochastic models used to analyze GNSS position time series (see e.g., \citealp{kermarrec2014matern,he2019investigation,wang2019impact}). For the WAMORE, we modify the procedure slightly by adding some steps to Algorithm \ref{algo:gmwmx}: first we estimate all models with the fixed diagonal matrix $\bm{\Omega}$ and then use the model with the smallest GMWM loss to construct a low-bias estimator of $\V$ (the covariance matrix of the WV). Based on this we can re-estimate all stochastic models using the inverse of this low-bias estimator of $\V$ as the weighting matrix $\boldsymbol{\Omega}$ in the GMWM step, which also allows us to compute the selection criterion presented in Section~\ref{sec:model_selection}. Based on the stochastic model selected from the latter criterion, we then compute the covariance matrix $\boldsymbol{\Phi}_n$ of the estimated functional parameters to extract the standard error of the estimated trend parameter $\hat{b}$ in \eqref{eq:func_model}. Using this we can construct 95\% CIs based on the derived asymptotic distribution of the estimated parameters $\hat{\boldsymbol{\beta}}$. The trends estimated by this procedure for all stations are represented in Figure~\ref{fig:plot_network_north_america} with their corresponding uncertainties (represented in gradient color scale).

\begin{figure}[H]
    \centering
    \includegraphics[width=1\linewidth]{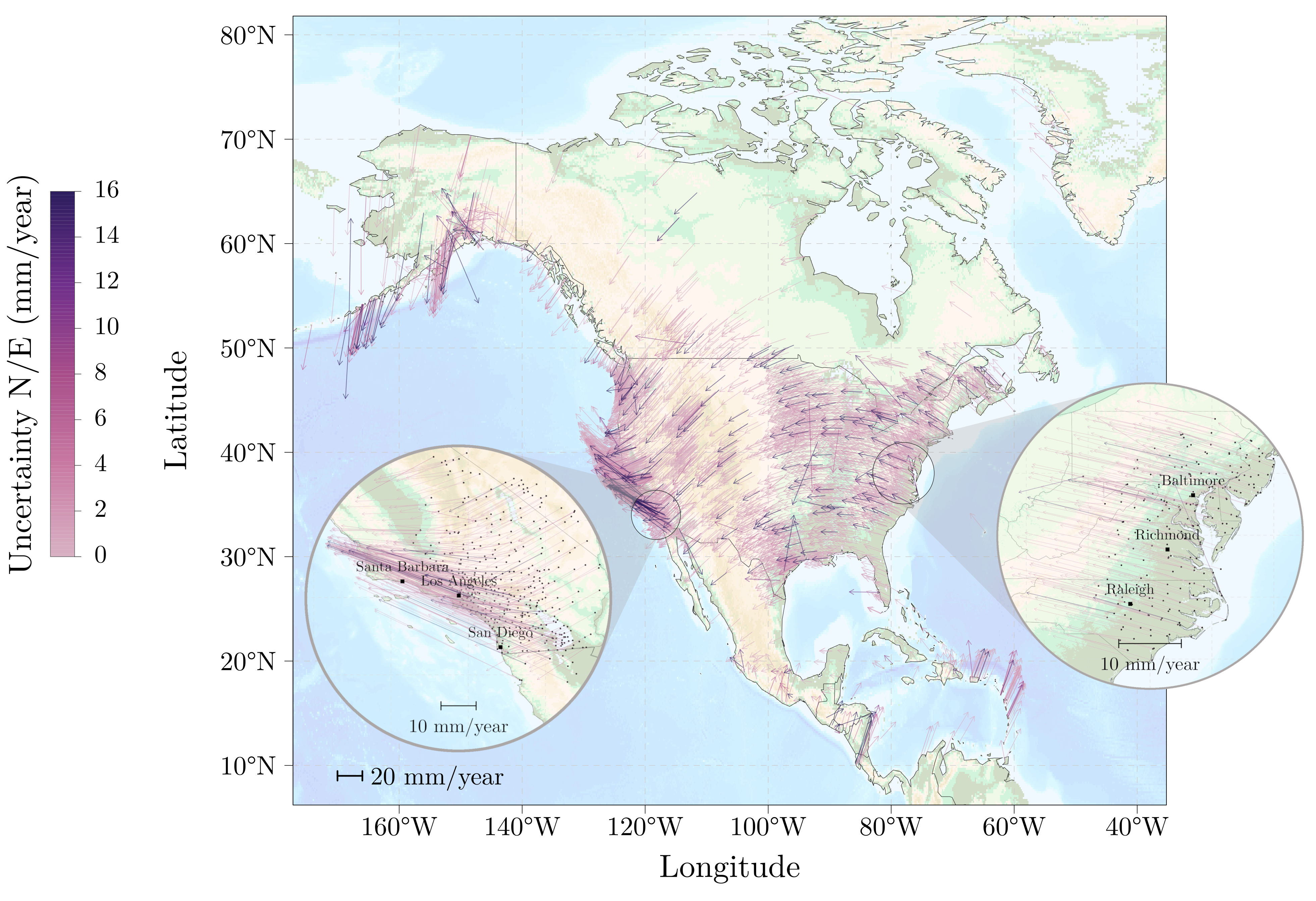}
    \caption{Estimated tectonic velocities and associated uncertainty using the WAMORE methodology. Arrow length encodes velocity magnitude, and color represents uncertainty, computed as the $\ell_2$-norm of the north and east component standard errors.}
    \label{fig:plot_network_north_america}
\end{figure}

It can be noticed that the eastern region generally reports similar trends with low uncertainty across all stations, whereas the western region has trends moving in different directions (with greater uncertainty) highlighting the known crustal behaviors for these areas. Indeed, at the continental scale, the inferred velocity field is consistent with the first-order kinematic structure of North American tectonics established by previous geodetic studies based on GNSS observations and plate motion models \citep{argus2010absolute, kreemer2014geodetic}. In particular, the results clearly distinguish the rigid behavior of the North American plate interior from the actively deforming plate boundary zones along the western margin. Stations located in eastern and central North America exhibit small-magnitude horizontal velocities with a high degree of directional coherence and uniformly low estimated uncertainty. This behavior is characteristic of the stable plate interior, where long-term strain rates are known to be very low and GNSS-derived residual velocities are typically on the order of a few tenths of a millimeter per year once reference-frame effects are accounted for \citep{mccaffrey2013active,demets2016high,kreemer2018robust}. The low uncertainty reported by the WAMORE indicates that the available observations provide relatively precise estimates of the trend parameters in these regions. This finding is consistent with the long observation records available at many stations in the stable interior of the North American plate and with previous studies showing that GNSS velocities can often be estimated with high precision when temporally correlated noise is appropriately modeled \citep{martin2017real}.

In contrast, the western portion of the North American region displays substantially larger velocities and pronounced spatial variability in both magnitude and direction. This pattern reflects distributed deformation along the Pacific-North American plate boundary, including transform faulting, block rotations and strain partitioning across multiple fault systems \citep{yang2013tectonic,rollins2018interseismic}. Similar velocity gradients and directional changes have been widely documented in regional and continental GNSS velocity fields \citep{mccaffrey2013active,demets2016high, kreemer2018robust}. Importantly, the WAMORE also reports systematically higher uncertainty for stations located within or adjacent to these tectonically active regions. From a geophysical perspective, elevated uncertainty in deforming regions is expected. GNSS time series in such environments are often affected by colored noise, transient deformation events, post-seismic relaxation and other non-stationary processes that complicate the estimation of long-term secular trends \citep{hammond2016gps,he2019investigation, duchnowski2024robust}.

Overall, the velocity field inferred using the WAMORE aligns closely with established local models of North American tectonics. Such a behavior is essential for downstream geophysical applications, including strain-rate estimation, block modeling and seismic hazard analysis, where reliable point estimates and realistic uncertainty quantification play a central role \citep{kreemer2014geodetic,verard2019plate}. The complete procedure, including model selection for the $4{,}793$ stations (corresponding to $14{,}379$ signals and $90{,}122{,}118$ data points), requires $170$ hours. In contrast, if the model were assumed to be known, the procedure would take only $20$ hours. To support the validity of the WAMORE analysis, we take a random sub-sample of 300 stations from the complete network that geographically covers the entire region, and then run the \texttt{Hector} software on these to compute the MLE. More specifically, we take a sub-sample since otherwise computing the MLE for the complete network is computationally prohibitive, thereby justifying the need for a more scalable method such as WAMORE. With this in mind, when comparing WAMORE to the MLE on these 300 stations, the results in terms of significance of the trend parameter $b$ are the same for \(85.11\%\) of the \(900\) signals analyzed. As shown in Figure~\ref{fig_agreement_test} in Appendix~\ref{app:case_study_additional_graphs}, agreement is the highest when both methods select the same stochastic model, exceeding \(92\%\) in nearly all such cases. When model selection coincides, CI lengths for the trend parameter $b$ are also similar, with median WAMORE-to-MLE ratios between \(1.05\) and \(1.98\); larger discrepancies arise primarily when one method selects a model including a random walk while the other does not. Since the true noise structure is unknown, we further investigate these differences through an emulation study which can be found in Appendix~\ref{app:emulation_study}. This analysis is designed to emulate realistic GNSS time series characteristics using parameter values estimated from an actual station analyzed in this section. We are therefore able to assess empirical coverage of CIs for the trend parameter under both direct model estimation and model selection. These results support the agreement between the WAMORE and the MLE in the real data analysis, especially for model selection criteria, while highlighting better coverage of the WAMORE for non-stationary cases when a random walk is included.

\subsection{Open Science Implications}

Having established the statistical validity of the WAMORE, we now consider its primary objective: enabling a broader community of researchers to analyze large-scale GNSS datasets using standard computational resources and within practical time frames. Expanding access to such analyses is essential for Open Science, as it allows a wider range of scientists to independently investigate crustal deformation processes and to contribute to the continuous monitoring of tectonic activity at regional and global scales. In turn, this facilitates the timely production of information that is critical for scientific synthesis and for decision-makers concerned with geophysical hazards.

In our case study, likelihood-based inference proves computationally prohibitive when access to computational resources is limited. Indeed, using the \texttt{Hector} software, estimation and model selection across six candidate stochastic models for $900$ signals required over $49$ days of computing time. Reducing this burden by limiting the model set would increase the risk of misspecification and compromise automation. In contrast, the WAMORE completed the same analysis in approximately $12.7$ hours, achieving nearly a two-order-of-magnitude speedup while delivering estimates and inferential conclusions largely consistent with those obtained via the MLE. These gains are conservative relative to simulation results, where the WAMORE can be up to $1{,}000$ times faster for longer GNSS records and higher proportions of missing data. To further demonstrate its scalability, we apply the WAMORE to a global analysis that is impractical (or infeasible/impossible) with existing methods. From over $20{,}000$ GNSS stations processed by the NGL, we analyze $6{,}313$ stations ($18{,}939$ signals across three components). Despite heterogeneous record lengths and missing data, the full analysis was completed in under $254$ hours using the WAMORE. Such large-scale processing is essential for maintaining global reference frames and studying deformation processes, yet would require months of computation with traditional likelihood-based approaches.

\section{Conclusion}
\label{sec:conclusion}

In this work, we proposed and applied WAMORE to the complete GNSS network covering North America and showed that it recovers a velocity field fully consistent with the first-order tectonic structure of this region, clearly distinguishing the stable plate interior from the actively deforming western plate boundary. The spatial patterns in both estimated velocities and their uncertainties reflect well-known geophysical regimes, with coherent, low-uncertainty signals in the continental interior and larger, more heterogeneous motions accompanied by increased uncertainty in tectonically complex regions. At the same time, WAMORE delivers inferential conclusions largely consistent with likelihood-based methods while reducing computational costs by orders of magnitude. These results demonstrate that WAMORE enables statistically sound, scalable and practically accessible analysis of large geodetic datasets, making it a valuable tool for routine monitoring of crustal deformation, the maintenance of reference frames and the timely assessment of geophysical hazards.

\begingroup
    {
    \small
     \setstretch{.875} 
\bibliography{biblio.bib}

@article{cucci2023generalized,
  title={{The Generalized Method of Wavelet Moments with eXogenous inputs: A Fast Approach for the Analysis of GNSS Position Time Series}},
  author={Cucci, Davide A and Voirol, Lionel and Kermarrec, Ga{\"e}l and Montillet, Jean-Philippe and Guerrier, St{\'e}phane},
  journal={Journal of Geodesy},
  volume={97},
  number={2},
  pages={14},
  year={2023},
  publisher={Springer}
}

@article{guerrier2022robust,
  title={{Robust two-step wavelet-based Inference for Time Series Models}},
  author={Guerrier, St{\'e}phane and Molinari, Roberto and Victoria-Feser, Maria-Pia and Xu, Haotian},
  journal={Journal of the American Statistical Association},
  volume={117},
  number={540},
  pages={1996--2013},
  year={2022},
  publisher={Taylor \& Francis}
}

@article{serroukh2000statistical,
  title={{Statistical Properties and Uses of the Wavelet Variance Estimator for the Scale Analysis of Time Series}},
  author={Serroukh, Abdeslam and Walden, Andrew T and Percival, Donald B},
  journal={Journal of the American Statistical Association},
  volume={95},
  number={449},
  pages={184--196},
  year={2000},
  publisher={Taylor \& Francis}
}

@article{jia2015correlations,
  title={{How do Correlations of Crude Oil Prices co-move? A Grey correlation-based Wavelet Perspective}},
  author={Jia, X. and An, H. and Fang, W. and Sun, X. and Huang, X.},
  journal={Energy Economics},
  volume={49},
  pages={588--598},
  year={2015},
  publisher={Elsevier}
}

@article{gallegati2012wavelet,
  title={{A wavelet-based Approach to Test for Financial Market Contagion}},
  author={Gallegati, M.},
  journal={Computational Statistics \& Data Analysis},
  volume={56},
  number={11},
  pages={3491--3497},
  year={2012},
  publisher={Elsevier}
}

@article{xie2013wavelet,
  title={{Wavelet-based Sparse Functional Linear Model with Applications to EEGs Seizure Detection and Epilepsy Diagnosis}},
  author={Xie, S. and Krishnan, S.},
  journal={Medical \& Biological Engineering \& Computing},
  volume={51},
  number={1-2},
  pages={49--60},
  year={2013},
  publisher={Springer}
}

@book{foufoula2014wavelets,
  title={{Wavelets in Geophysics}},
  author={Foufoula-Georgiou, E. and Kumar, P.},
  volume={4},
  year={2014},
  publisher={Academic Press}
}

@article{abry2018wavelet,
  title={{Wavelet Estimation for Operator Fractional Brownian Motion}},
  author={Abry, P. and Didier, G. and others},
  journal={Bernoulli},
  volume={24},
  number={2},
  pages={895--928},
  year={2018},
  publisher={Bernoulli Society for Mathematical Statistics and Probability}
}

@article{ziaja2016fault,
  title={{Fault Detection in Rolling Element Bearings Using wavelet-based Variance Analysis and Novelty Detection}},
  author={Ziaja, A. and Antoniadou, I. and Barszcz, T. and Staszewski, W. J. and Worden, K.},
  journal={Journal of Vibration and Control},
  volume={22},
  number={2},
  pages={396--411},
  year={2016},
  publisher={SAGE Publications Sage UK: London, England}
}

@article{fan2010unit,
  title={{Unit Root Tests with Wavelets}},
  author={Fan, Yanqin and Gen{\c{c}}ay, Ramazan},
  journal={Econometric Theory},
  volume={26},
  number={5},
  pages={1305--1331},
  year={2010},
  publisher={Cambridge University Press}
}

@article{thon2014multiscale,
  title={{A Multiscale wavelet-based Test for Isotropy of Random Fields on a Regular Lattice}},
  author={Thon, Kevin and Geilhufe, Marc and Percival, Donald B},
  journal={IEEE Transactions on Image Processing},
  volume={24},
  number={2},
  pages={694--708},
  year={2014},
  publisher={IEEE}
}

@article{gencay2015multi,
  title={{Multi-scale Tests for Serial Correlation}},
  author={Gencay, Ramazan and Signori, Daniele},
  journal={Journal of Econometrics},
  volume={184},
  number={1},
  pages={62--80},
  year={2015},
  publisher={Elsevier}
}

@article{zawacki2025history,
  title={{A History of UNAVCO: Four Decades of Advancing Geodesy}},
  author={Zawacki, Emily E and Charlevoix, Donna J and Meertens, Charles M and Freymueller, Jeffrey T and van Dam, Tonie},
  journal={Perspectives of Earth and Space Scientists},
  volume={6},
  number={1},
  pages={e2025CN000276},
  year={2025},
  publisher={Wiley Online Library}
}

@article{williams2003effect,
  title={{The Effect of Coloured Noise on the Uncertainties of Rates Estimated from Geodetic Time Series}},
  author={Williams, SDP},
  journal={Journal of Geodesy},
  volume={76},
  number={9},
  pages={483--494},
  year={2003},
  publisher={Springer}
}

@article{sun2023relationship,
  title={{The Relationship of Time Span and Missing Data on the Noise Model Estimation of GNSS Time Series}},
  author={Sun, Xiwen and Lu, Tieding and Hu, Shunqiang and Huang, Jiahui and He, Xiaoxing and Montillet, Jean-Philippe and Ma, Xiaping and Huang, Zhengkai},
  journal={Remote Sensing},
  volume={15},
  number={14},
  pages={3572},
  year={2023},
  publisher={MDPI}
}

@article{serpelloni2022surface,
  title={{Surface Velocities and strain-rates in the Euro-Mediterranean Region from Massive GPS Data Processing}},
  author={Serpelloni, Enrico and Cavaliere, Adriano and Martelli, Leonardo and Pintori, Francesco and Anderlini, Letizia and Borghi, Alessandra and Randazzo, Daniele and Bruni, Sergio and Devoti, Roberto and Perfetti, Paolo and others},
  journal={Frontiers in Earth Science},
  volume={10},
  pages={907897},
  year={2022},
  publisher={Frontiers Media SA}
}

@article{lv2025investigating,
  title={{Investigating Surface Loading Effect on Seasonal Crustal Deformation Observed by GNSS in Hong Kong}},
  author={Lv, Hongli and He, Xiaoxing and Hu, Shunqiang},
  journal={Scientific Reports},
  volume={15},
  number={1},
  pages={2742},
  year={2025},
  publisher={Nature Publishing Group UK London}
}

@article{liu2022missing,
  title={{Missing Data Imputation in GNSS Monitoring Time Series Using Temporal and Spatial Hankel Matrix Factorization}},
  author={Liu, Hanlin and Li, Linchao},
  journal={Remote Sensing},
  volume={14},
  number={6},
  pages={1500},
  year={2022},
  publisher={MDPI}
}

@Book{claeskens2008model,
publisher={Cambridge University Press},
series={Cambridge Books},
author={Claeskens,Gerda and Hjort,Nils Lid},
title={Model Selection and Model Averaging},
year={2008},
month={August},
url={https://ideas.repec.org/b/cup/cbooks/9780521852258.html},
}

@book{mcquarrie1998regression,
  title={{Regression and Time Series Model Selection}},
  author={McQuarrie, Allan DR and Tsai, Chih-Ling},
  year={1998},
  publisher={World Scientific}
}

@article{efron2004estimation,
  title={{The Estimation of Prediction Error: Covariance Penalties and cross-validation}},
  author={Efron, Bradley},
  journal={Journal of the American Statistical Association},
  volume={99},
  number={467},
  pages={619--632},
  year={2004},
  publisher={Taylor \& Francis}
}

@article{das2021unesco,
  title={{UNESCO Recommendation on Open Science: An Upcoming Milestone in Global Science}},
  author={Das, Anup Kumar},
  journal={Science Diplomacy},
  pages={39},
  year={2021}
}

@article{achar2022cloud,
  title={{Cloud Computing: Toward Sustainable Processes and Better Environmental Impact}},
  author={Achar, Sandesh},
  journal={Journal of Computer Hardware Engineering (JCHE)},
  volume={1},
  number={1},
  year={2022}
}

@article{andrews1999consistent,
  title={{Consistent Moment Selection Procedures for Generalized Method of Moments Estimation}},
  author={Andrews, Donald WK},
  journal={Econometrica},
  volume={67},
  number={3},
  pages={543--563},
  year={1999},
  publisher={Wiley Online Library}
}

@article{mallows2000some,
  title={{Some Comments on Cp}},
  author={Mallows, Colin L},
  journal={Technometrics},
  volume={42},
  number={1},
  pages={87--94},
  year={2000},
  publisher={Taylor \& Francis}
}

@article{guerrier2015automatic,
  title={{Automatic Identification and Calibration of Stochastic Parameters in Inertial Sensors}},
  author={Guerrier, St{\'e}phane and Molinari, Roberto and Skaloud, Jan},
  journal={Navigation: Journal of The Institute of Navigation},
  volume={62},
  number={4},
  pages={265--272},
  year={2015},
  publisher={Wiley Online Library}
}

@article{radi2019multisignal,
  title={{A Multisignal Wavelet variance-based Framework for Inertial Sensor Stochastic Error Modeling}},
  author={Radi, Ahmed and Bakalli, Gaetan and Guerrier, St{\'e}phane and El-Sheimy, Naser and Sesay, Abu B and Molinari, Roberto},
  journal={IEEE Transactions on Instrumentation and Measurement},
  volume={68},
  number={12},
  pages={4924--4936},
  year={2019},
  publisher={IEEE}
}

@article{hansen1982large,
  title={{Large Sample Properties of Generalized Method of Moments Estimators}},
  author={Hansen, Lars Peter},
  journal={Econometrica},
  pages={1029--1054},
  year={1982},
  publisher={JSTOR}
}

@inproceedings{snow2023cryocloud,
  title={{CryoCloud JupyterHub for NASA Cryosphere Communities: Open science in the Cloud as a Process, not a Product}},
  author={Snow, Tasha and Millstein, Joanna D and Sauthoff, Wilson and Scheick, Jessica and Leong, Wei Ji and Colliander, James and Munroe, James and P{\'e}rez, Fernando and Felikson, Denis and Sutterley, Tyler C and others},
  booktitle={AGU Fall Meeting Abstracts},
  volume={2023},
  pages={U24B--03},
  year={2023}
}

@article{fuentes2002spectral,
  title={{Spectral Methods for Nonstationary Spatial Processes}},
  author={Fuentes, Montserrat},
  journal={Biometrika},
  volume={89},
  number={1},
  pages={197--210},
  year={2002},
  publisher={Oxford University Press}
}

@article{gerber2018predicting,
  title={{Predicting Missing Values in Spatio-Temporal Remote Sensing Data}},
  author={Gerber, Florian and de Jong, Rogier and Schaepman, Michael E and Schaepman-Strub, Gabriela and Furrer, Reinhard},
  journal={IEEE Transactions on Geoscience and Remote Sensing},
  volume={56},
  number={5},
  pages={2841--2853},
  year={2018},
  publisher={IEEE}
}

@article{guinness2019spectral,
  title={{Spectral Density Estimation for Random Fields via Periodic Embeddings}},
  author={Guinness, Joseph},
  journal={Biometrika},
  volume={106},
  number={2},
  pages={267--286},
  year={2019},
  publisher={Oxford University Press}
}

@article{stein2013stochastic,
  title={{Stochastic Approximation of Score Functions for Gaussian Processes}},
  author={Stein, Michael L and Chen, Jie and Anitescu, Mihai},
  year={2013},
    journal = {The Annals of Applied Statistics},
    volume  = {7},
  number  = {2},
    pages   = {1162--1191}
}

@inproceedings{lazaro2011variational,
  title={{Variational Heteroscedastic Gaussian Process Regression}},
  author={L{\'a}zaro-Gredilla, Miguel and Titsias, Michalis K},
  booktitle={ICML},
  pages={841--848},
  year={2011}
}

@article{gibbs2000variational,
  title={{Variational Gaussian Process Classifiers}},
  author={Gibbs, Mark N and MacKay, David JC},
  journal={IEEE Transactions on Neural Networks},
  volume={11},
  number={6},
  pages={1458--1464},
  year={2000},
  publisher={IEEE}
}

@article{sang2011covariance,
  title={{Covariance Approximation for Large Multivariate Spatial Data Sets with an Application to Multiple Climate Model Errors}},
  author={Sang, Huiyan and Jun, Mikyoung and Huang, Jianhua Z},
  journal={The Annals of Applied Statistics},
  pages={2519--2548},
  year={2011},
  publisher={JSTOR}
}

@article{chib2001markov,
  title={{Markov Chain Monte Carlo Methods: Computation and Inference}},
  author={Chib, Siddhartha},
  journal={{Handbook of Econometrics}},
  volume={5},
  pages={3569--3649},
  year={2001},
  publisher={Elsevier}
}

@article{guerrier2013wavelet,
  title={{Wavelet-Variance-based Estimation for Composite Stochastic Processes}},
  author={Guerrier, St{\'e}phane and Skaloud, Jan and Stebler, Yannick and Victoria-Feser, Maria-Pia},
  journal={Journal of the American Statistical Association},
  volume={108},
  number={503},
  pages={1021--1030},
  year={2013},
  publisher={Taylor \& Francis}
}

@article{Bos2008,
author = {Bos, M. S. and Fernandes, R. M. S. and Williams, S. D. P. and Bastos, L.},
 year = {2008},
 title = {{Fast Error Analysis of Continuous GPS Observations}},
 pages = {157--166},
 volume = {82},
 number = {3},
 issn = {1432-1394},
 journal = {Journal of Geodesy},
 doi = {10.1007/s00190-007-0165-x}
}

@article{bos2013fast,
  title={{Fast Error Analysis of Continuous GNSS Observations with Missing Data}},
  author={Bos, MS and Fernandes, RMS and Williams, SDP and Bastos, L},
  journal={Journal of Geodesy},
  volume={87},
  number={4},
  pages={351--360},
  year={2013},
  publisher={Springer}
}

@article{heaton2019case,
  title={{A Case Study Competition Among Methods for Analyzing Large Spatial Data}},
  author={Heaton, Matthew J and Datta, Abhirup and Finley, Andrew O and Furrer, Reinhard and Guinness, Joseph and Guhaniyogi, Rajarshi and Gerber, Florian and Gramacy, Robert B and Hammerling, Dorit and Katzfuss, Matthias and others},
  journal={Journal of Agricultural, Biological and Environmental Statistics},
  volume={24},
  pages={398--425},
  year={2019},
  publisher={Springer}
}

@article{tehranchi2021fast,
  title={{Fast Approximation Algorithm to Noise Components Estimation in long-term GPS Coordinate Time Series}},
  author={Tehranchi, R and Moghtased-Azar, K and Safari, A},
  journal={Journal of Geodesy},
  volume={95},
  number={2},
  pages={1--16},
  year={2021},
  publisher={Springer}
}

@article{he2021analysis,
  title={{Analysis and Discussion on the Optimal Noise Model of Global GNSS long-term Coordinate Series Considering Hydrological Loading}},
  author={He, Yuefan and Nie, Guigen and Wu, Shuguang and Li, Haiyang},
  journal={Remote Sensing},
  volume={13},
  number={3},
  pages={431},
  year={2021},
  publisher={MDPI}
}

@article{kaczmarek2018identification,
  title={{Identification of the Noise Model in the Time Series of GNSS Stations Coordinates using Wavelet Analysis}},
  author={Kaczmarek, Adrian and Kontny, Bernard},
  journal={Remote Sensing},
  volume={10},
  number={10},
  pages={1611},
  year={2018},
  publisher={MDPI}
}

@article{shen1994postseismic,
  title={{Postseismic Deformation Following the Landers Earthquake, California, 28 June 1992}},
  author={Shen, Zheng-Kang and Jackson, David D and Feng, Yanjie and Cline, Michael and Kim, Mercedes and Fang, Peng and Bock, Yehuda},
  journal={Bulletin of the Seismological Society of America},
  volume={84},
  number={3},
  pages={780--791},
  year={1994},
  publisher={The Seismological Society of America}
}

@article{ren2021analysis,
  title={{Analysis of Seismic Deformation from Global three-decade GNSS Displacements: Implications for a three-dimensional Earth GNSS Velocity Field}},
  author={Ren, Yingying and Lian, Lizhen and Wang, Jiexian},
  journal={Remote Sensing},
  volume={13},
  number={17},
  pages={3369},
  year={2021},
  publisher={MDPI}
}

@article{he2019investigation,
  title={{Investigation of the Noise Properties at Low Frequencies in Long GNSS Time Series}},
  author={He, X and Bos, MS and Montillet, JP and Fernandes, RMS},
  journal={Journal of Geodesy},
  volume={93},
  number={9},
  pages={1271--1282},
  year={2019},
  publisher={Springer}
}

@article{kermarrec2024modeling,
  title={{Modeling Trends and Periodic Components in Geodetic Time Series: A Unified Approach}},
  author={Kermarrec, Ga{\"e}l and Maddanu, Federico and Klos, Anna and Proietti, Tommaso and Bogusz, Janusz},
  journal={Journal of Geodesy},
  volume={98},
  number={3},
  pages={17},
  year={2024},
  publisher={Springer}
}

@article{freymueller1999kinematics,
  title={{Kinematics of the Pacific-North America Plate Boundary Zone, Northern California}},
  author={Freymueller, Jeffrey T and Murray, Mark H and Segall, Paul and Castillo, David},
  journal={Journal of Geophysical Research: Solid Earth},
  volume={104},
  number={B4},
  pages={7419--7441},
  year={1999},
  publisher={Wiley Online Library}
}

@article{mudelsee2019trend,
  title={{Trend Analysis of Climate Time Series: A Review of Methods}},
  author={Mudelsee, Manfred},
  journal={Earth-Science Reviews},
  volume={190},
  pages={310--322},
  year={2019},
  publisher={Elsevier}
}

@article{newey1994large,
  title={{Large Sample Estimation and Hypothesis Testing}},
  author={Newey, Whitney K and McFadden, Daniel},
  journal={Handbook of Econometrics},
  volume={4},
  pages={2111--2245},
  year={1994},
  publisher={Elsevier}
}

@article{xu2019multivariate,
  title={{Multivariate Signal Modeling with Applications to Inertial Sensor Calibration}},
  author={Xu, Haotian and Guerrier, St{\'e}phane and Molinari, Roberto Carlo and Karemera, Mucyo},
  journal={IEEE Transactions on Signal Processing},
  volume={67},
  number={19},
  pages={5143--5152},
  year={2019},
  publisher={IEEE}
}

@article{bao2021filling,
  title={{Filling Missing Values of Multi-station GNSS Coordinate Time Series Based on Matrix Completion}},
  author={Bao, Zhi and Chang, Guobin and Zhang, Laihong and Chen, Guoliang and Zhang, Siyu},
  journal={Measurement},
  volume={183},
  pages={},
  year={2021},
  publisher={Elsevier}
}

@article{bock2000instantaneous,
  title={{Instantaneous Geodetic Positioning at Medium Distances with the Global Positioning System}},
  author={Bock, Yehuda and Nikolaidis, Rosanne M and de Jonge, Paul J and Bevis, Michael},
  journal={Journal of Geophysical Research: Solid Earth},
  volume={105},
  number={B12},
  pages={28223--28253},
  year={2000},
  publisher={Wiley Online Library}
}

@article{calais1999continuous,
  title={{Continuous GPS Measurements Across the Western Alps, 1996--1998}},
  author={Calais, Eric},
  journal={Geophysical Journal International},
  volume={138},
  number={1},
  pages={221--230},
  year={1999},
  publisher={Blackwell Publishing Ltd Oxford, UK}
}

@article{zhang1997southern,
  title={{Southern California Permanent GPS Geodetic Array: Error Analysis of Daily Position Estimates and Site Velocities}},
  author={Zhang, Jie and Bock, Yehuda and Johnson, Hadley and Fang, Peng and Williams, Simon and Genrich, Joachim and Wdowinski, Shimon and Behr, Jeff},
  journal={Journal of Geophysical Research: Solid Earth},
  volume={102},
  number={B8},
  pages={18035--18055},
  year={1997},
  publisher={Wiley Online Library}
}

@article{hohensinn2018stand,
  title={{Stand-alone GNSS Sensors as Velocity Seismometers: Real-time Monitoring and Earthquake Detection}},
  author={Hohensinn, Roland and Geiger, Alain},
  journal={Sensors},
  volume={18},
  number={11},
  pages={3712},
  year={2018},
  publisher={MDPI}
}

@article{gokdacs2021velocity,
  title={{Velocity Estimation Performance of GNSS Online Services (APPS and AUSPOS)}},
  author={G{\"o}kda{\c{s}}, {\"O}mer and {\"O}zl{\"u}demir, M Tevfik},
  journal={Survey Review},
  volume={53},
  number={378},
  pages={280--288},
  year={2021},
  publisher={Taylor \& Francis}
}

@article{aydin2021effect,
  title={{Effect of Stochastic Model Errors on Significance Test for Velocities in Analysis of GPS Position Time Series}},
  author={Aydin, C{\"u}neyt and Duman, H{\"u}seyin and G{\"u}nes, {\"O}zge and Ugur Sanli, Dogan},
  journal={Journal of Surveying Engineering},
  volume={147},
  number={1},
  pages={},
  year={2021},
  publisher={American Society of Civil Engineers}
}

@article{wang202295,
  title={{The 95\% Confidence Interval for GNSS-derived Site Velocities}},
  author={Wang, Guoquan},
  journal={Journal of Surveying Engineering},
  volume={148},
  number={1},
  pages={},
  year={2022},
  publisher={American Society of Civil Engineers}
}

@article{bevis2014trajectory,
  title={{Trajectory Models and Reference Frames for Crustal Motion Geodesy}},
  author={Bevis, Michael and Brown, Abel},
  journal={Journal of Geodesy},
  volume={88},
  pages={283--311},
  year={2014},
  publisher={Springer}
}

@article{kermarrec2014matern,
  title={{On the M{\'a}tern Covariance Family: A Proposal for Modeling Temporal Correlations Based on Turbulence Theory}},
  author={Kermarrec, Ga{\"e}l and Sch{\"o}n, Steffen},
  journal={Journal of Geodesy},
  volume={88},
  pages={1061--1079},
  year={2014},
  publisher={Springer}
}

@article{wang2019impact,
  title={{Impact of Estimating Position Offsets on the Uncertainties of GNSS Site Velocity Estimates}},
  author={Wang, Lei and Herring, Thomas},
  journal={Journal of Geophysical Research: Solid Earth},
  volume={124},
  number={12},
  pages={13452--13467},
  year={2019},
  publisher={Wiley Online Library}
}

@article{amiri2007assessment,
  title={{Assessment of Noise in GPS Coordinate Time Series: Methodology and Results}},
  author={Amiri-Simkooei, Ali Reza and Tiberius, Christian CJM and Teunissen, Peter JG},
  journal={Journal of Geophysical Research: Solid Earth},
  volume={112},
  number={B7},
  year={2007},
  publisher={Wiley Online Library}
}

@article{kleinherenbrink2018comparison,
  title={{A Comparison of Methods to Estimate Vertical Land Motion Trends from GNSS and Altimetry at Tide Gauge Stations}},
  author={Kleinherenbrink, Marcel and Riva, Riccardo and Frederikse, Thomas},
  journal={Ocean Science},
  volume={14},
  number={2},
  pages={187--204},
  year={2018},
  publisher={Copernicus GmbH}
}

@article{akaike1974new,
  title={{A New Look at the Statistical Model Identification}},
  author={Akaike, Hirotugu},
  journal={IEEE Transactions on Automatic Control},
  volume={19},
  number={6},
  pages={716--723},
  year={1974},
  publisher={Ieee}
}

@article{williams2008cats,
  title={{CATS: GPS Coordinate Time Series Analysis Software}},
  author={Williams, Simon DP},
  journal={GPS Solutions},
  volume={12},
  number={2},
  pages={147--153},
  year={2008},
  publisher={Springer}
}

@article{voirol2023accounting,
  title={{Accounting for Vibration Noise in Stochastic Measurement Errors of Inertial Sensors}},
  author={Karemera, Mucyo and Voirol, Lionel and Cucci, Davide A and Chu, Wenfei and Molinari, Roberto and Guerrier, St{\'e}phane},
  journal={IEEE Transactions on Signal Processing},
  volume={72},
  pages={2117--2129},
  year={2024},
  publisher={IEEE}
}

@article{guerrier2014estimation,
  title={{Estimation of Time Series Models via Robust Wavelet Variance}},
  author={Guerrier, St{\'e}phane and Molinari, Roberto and Victoria-Feser, Maria-Pia},
  journal={Austrian Journal of Statistics},
  volume={43},
  number={4},
  pages={267--277},
  year={2014}
}

@article{wang2012noise,
  title={{Noise Analysis of Continuous GPS Coordinate Time Series for CMONOC}},
  author={Wang, Wei and Zhao, Bin and Wang, Qi and Yang, Shaomin},
  journal={Advances in Space Research},
  volume={49},
  number={5},
  pages={943--956},
  year={2012},
  publisher={Elsevier}
}

@article{he2017review,
  title={{Review of Current GPS Methodologies for Producing Accurate Time Series and their Error Sources}},
  author={He, Xiaoxing and Montillet, Jean-Philippe and Fernandes, Rui and Bos, Machiel and Yu, Kegen and Hua, Xianghong and Jiang, Weiping},
  journal={Journal of Geodynamics},
  volume={106},
  pages={12--29},
  year={2017},
  publisher={Elsevier}
}

@article{mao1999noise,
  title={{Noise in GPS Coordinate Time Series}},
  author={Mao, Ailin and Harrison, Christopher GA and Dixon, Timothy H},
  journal={Journal of Geophysical Research: Solid Earth},
  volume={104},
  number={B2},
  pages={2797--2816},
  year={1999},
  publisher={Wiley Online Library}
}

@article{langbein1997correlated,
  title={{Correlated Errors in Geodetic Time Series: Implications for time-dependent Deformation}},
  author={Langbein, John and Johnson, Hadley},
  journal={Journal of Geophysical Research: Solid Earth},
  volume={102},
  number={B1},
  pages={591--603},
  year={1997},
  publisher={Wiley Online Library}
}

@article{blewitt2024improved,
  title={{An Improved Equation of Latitude and a Global System of Graticule Distance Coordinates}},
  author={Blewitt, Geoffrey},
  journal={Journal of Geodesy},
  volume={98},
  number={1},
  pages={6},
  year={2024},
  publisher={Springer}
}

@article{blewitt2018harnessing,
  title={{Harnessing the GPS Data Explosion for Interdisciplinary Science}},
  author={Blewitt, Geoffrey and Hammond, William and Kreemer, Corn},
  journal={Eos},
  volume={99},
  number={2},
  pages={},
  year={2018}
}

@article{percival1995estimation,
  title={{On Estimation of the Wavelet Variance}},
  author={Percival, Donald P},
  journal={Biometrika},
  volume={82},
  number={3},
  pages={619--631},
  year={1995},
  publisher={Oxford University Press}
}

@book{percival2000wavelet,
  title={{Wavelet Methods for Time Series Analysis}},
  author={Percival, Donald B and Walden, Andrew T},
  volume={4},
  year={2000},
  publisher={Cambridge University Press}
}

@article{montillet2024big,
  title={{How Big Data Can Help to Monitor the Environment and to Mitigate Risks due to Climate Change: A Review}},
  author={Montillet, J-P and Kermarrec, Ga{\"e}l and Forootan, Ehsan and Haberreiter, Margit and He, Xiaoxing and Finsterle, Wolfgang and Fernandes, Rui and Shum, CK},
  journal={IEEE Geoscience and Remote Sensing Magazine},
  year={2024},
  publisher={IEEE}
}

@article{proietti2022modelling,
  title={{Modelling Cycles in Climate Series: The Fractional Sinusoidal Waveform Process}},
  author={Proietti, Tommaso and Maddanu, Federico},
  journal={Journal of Econometrics},
  pages={105299},
  year={2022},
  publisher={Elsevier}
}

@article{maddanu2023trends,
  title={{Trends in Atmospheric Ethane}},
  author={Maddanu, Federico and Proietti, Tommaso},
  journal={Climatic Change},
  volume={176},
  number={5},
  pages={53},
  year={2023},
  publisher={Springer}
}

@article{tunini2024global,
  title={{Global Navigation Satellite System (GNSS) Time Series and Velocities About a Slowly Convergent Margin Processed on High-Performance Computing (HPC) Clusters: Products and Robustness Evaluation}},
  author={Tunini, Lavinia and Magrin, Andrea and Rossi, Giuliana and Zuliani, David},
  journal={Earth System Science Data},
  volume={16},
  number={2},
  pages={1083--1106},
  year={2024},
  publisher={Copernicus Publications G{\"o}ttingen, Germany}
}

@article{xu2017study,
  title={{A Study of the Allan Variance for constant-mean Nonstationary Processes}},
  author={Xu, Haotian and Guerrier, St{\'e}phane and Molinari, Roberto and Zhang, Yuming},
  journal={IEEE Signal Processing Letters},
  volume={24},
  number={8},
  pages={1257--1260},
  year={2017},
  publisher={IEEE}
}

@article{andrews2001consistent,
  title={{Consistent Model and Moment Selection Procedures for GMM Estimation with Application to Dynamic Panel Data Models}},
  author={Andrews, Donald WK and Lu, Biao},
  journal={Journal of Econometrics},
  volume={101},
  number={1},
  pages={123--164},
  year={2001},
  publisher={Elsevier}
}

@article{banna2016bernstein,
  title={{Bernstein-type Inequality for a Class of Dependent Random Matrices}},
  author={Banna, Marwa and Merlev{\`e}de, Florence and Youssef, Pierre},
  journal={Random Matrices: Theory and Applications},
  volume={5},
  number={02},
  pages={1650006},
  year={2016},
  publisher={World Scientific}
}

@article{bradley2005basic,
  title={{Basic Properties of Strong Mixing Conditions. A Survey and Some Open Questions}},
  author={Bradley, Richard C},
  journal={Probability Surveys},
  volume={2},
  pages={107--144},
  year={2005}
}

@article{pollard1991asymptotics,
  title={{Asymptotics for Least Absolute Deviation Regression Estimators}},
  author={Pollard, David},
  journal={Econometric Theory},
  volume={7},
  number={2},
  pages={186--199},
  year={1991},
  publisher={Cambridge University Press}
}

@book{beran2013long,
  title={{Long-memory Processes}},
  author={Beran, Jan and Feng, Yuanhua and Ghosh, Sucharita and Kulik, Rafal},
  year={2013},
  publisher={Springer}
}

@article{pipiras2000convergence,
  title={{Convergence of Weighted Sums of Random Variables with Long-range Dependence}},
  author={Pipiras, Vladas and Taqqu, Murad S},
  journal={Stochastic Processes and their Applications},
  volume={90},
  number={1},
  pages={157--174},
  year={2000},
  publisher={Elsevier}
}

@article{rosenthal1995convergence,
  title={{Convergence Rates for Markov Chains}},
  author={Rosenthal, Jeffrey S},
  journal={SIAM Review},
  volume={37},
  number={3},
  pages={387--405},
  year={1995},
  publisher={SIAM}
}

@article{xu2021online,
  title={{Online Network Change Point Detection with Missing Values and Temporal Dependence}},
  author={Xu, Haotian and Dubey, Paromita and Yu, Yi},
  journal={Journal of Time Series Analysis},
  year={2025},
  publisher={Wiley Online Library}
}

@book{fan2008nonlinear,
  title={{Nonlinear Time Series: Nonparametric and Parametric Methods}},
  author={Fan, Jianqing and Yao, Qiwei},
  year={2008},
  publisher={Springer Science \& Business Media}
}

@article{hannan1973asymptotic,
  title={{The Asymptotic Theory of Linear time-series Models}},
  author={Hannan, Edward J},
  journal={Journal of Applied Probability},
  volume={10},
  number={1},
  pages={130--145},
  year={1973},
  publisher={Cambridge University Press}
}

@article{pipiras2001classes,
  title={{Are Classes of Deterministic Integrands for Fractional Brownian Motion on an Interval Complete?}},
  author={Pipiras, Vladas and Taqqu, Murad S},
  year={2001},
    journal = {Bernoulli},
  volume   = {7},
  number   = {6},
  pages    = {873--897},
  year     = {2001}
}

@book{stein2003introduction,
  title={Introduction to Seismology, Earthquakes, and Earth Structure},
  author={Stein, Seth and Wysession, Michael},
  year={2003},
  publisher={Blackwell Publishing}
}

@book{turcotte2014geodynamics,
  title={Geodynamics},
  author={Turcotte, Donald L and Schubert, Gerald},
  year={2014},
  publisher={Cambridge University Press}
}

@book{scholz2002mechanics,
  title={The Mechanics of Earthquakes and Faulting},
  author={Scholz, Christopher H},
  year={2002},
  publisher={Cambridge University Press}
}

@book{fossen2016structural,
  title={Structural Geology},
  author={Fossen, Haakon},
  year={2016},
  publisher={Cambridge University Press}
}

@article{calais2006continental,
  title={{Continental Deformation in North America: The Role of Plate Boundary Forces, Intraplate Stresses, and Mantle Structure}},
  author={Calais, Eric and Stein, Seth and Newman, Andrew},
  journal={Journal of Geophysical Research: Solid Earth},
  volume={111},
  number={B6},
  year={2006},
  publisher={Wiley Online Library}
}

@article{argus2010glacial,
  title={{The Viscosity Structure of the Upper Mantle Inferred from GPS Measurements of Glacial Isostatic Adjustment}},
  author={Argus, Donald F and Peltier, W Richard and Drummond, Rachel and Moore, Angelyn W},
  journal={Journal of Geophysical Research: Solid Earth},
  volume={115},
  number={B9},
  year={2010},
  publisher={Wiley Online Library}
}

@article{calais2005strain,
  title={{Strain Accumulation in the New Madrid Seismic Zone from GPS Observations, 1997–2004}},
  author={Calais, Eric and Han, Jiaxing and DeMets, Charles and Nocquet, Jean-Mathieu},
  journal={Geophysical Research Letters},
  volume={32},
  number={23},
  year={2005},
  publisher={Wiley Online Library}
}

@article{lambeck2014sea,
  title={{Sea Level and Global Ice Volumes from the Last Glacial Maximum to the Holocene}},
  author={Lambeck, Kurt and Rouby, H{\'e}l{\`e}ne and Purcell, Anthony and Sun, Y and Sambridge, Malcolm},
  journal={Proceedings of the National Academy of Sciences},
  volume={111},
  number={43},
  pages={15296--15303},
  year={2014},
  publisher={National Acad Sciences}
}

@article{argus2010absolute,
  title={{The Angular Velocities of the Plates and the Velocity of Earth's Centre from Space Geodesy}},
  author={Argus, Donald F and Gordon, Richard G and Heflin, Michael B and Ma, Chopo and Eanes, Richard J and Willis, Pascal and Peltier, W Richard and Owen, Susan E},
  journal={Geophysical Journal International},
  volume={180},
  number={3},
  pages={913--960},
  year={2010},
  publisher={Blackwell Publishing Ltd Oxford, UK}
}

@article{kreemer2014geodetic,
  title={{A Geodetic Plate Motion and Global Strain Rate Model}},
  author={Kreemer, Corn{\'e} and Blewitt, Geoffrey and Klein, Elliot C},
  journal={Geochemistry, Geophysics, Geosystems},
  volume={15},
  number={10},
  pages={3849--3889},
  year={2014},
  publisher={Wiley Online Library}
}

@article{kreemer2018robust,
  title={{A Robust Estimation of the 3-D Intraplate Deformation of the North American Plate from GPS}},
  author={Kreemer, Corn{\'e} and Hammond, William C and Blewitt, Geoffrey},
  journal={Journal of Geophysical Research: Solid Earth},
  volume={123},
  number={5},
  pages={4388--4412},
  year={2018},
  publisher={Wiley Online Library}
}

@article{demets2016high,
  title={{High-resolution Reconstructions of Pacific--North America Plate Motion: 20 Ma to Present}},
  author={DeMets, C and Merkouriev, S},
  journal={Geophysical Journal International},
  volume={207},
  number={2},
  pages={741--773},
  year={2016},
  publisher={Oxford University Press}
}

@article{mccaffrey2013active,
  title={{Active Tectonics of Northwestern US Inferred from GPS-derived Surface Velocities}},
  author={McCaffrey, Robert and King, Robert W and Payne, Suzette J and Lancaster, Matthew},
  journal={Journal of Geophysical Research: Solid Earth},
  volume={118},
  number={2},
  pages={709--723},
  year={2013},
  publisher={Wiley Online Library}
}

@inproceedings{martin2017real,
  title={{Real-time Uncertainty Quantification Using Correlated Noise Models for GNSS Positioning}},
  author={Martin, AD and Soundy, Andrew WR and Panckhurst, Bradley J and Brown, CP and Schumayer, Daniel and Molteno, Tim CA and Parry, Matthew},
  booktitle={2017 IEEE SENSORS},
  pages={1--3},
  year={2017},
  organization={IEEE}
}

@article{yang2013tectonic,
  title={{The Tectonic Crustal Stress Field and Style of Faulting Along the Pacific North America Plate Boundary in Southern California}},
  author={Yang, Wenzheng and Hauksson, Egill},
  journal={Geophysical Journal International},
  volume={194},
  number={1},
  pages={100--117},
  year={2013},
  publisher={Oxford University Press}
}

@article{duchnowski2024robust,
  title={{Robust Procedures in Processing Measurements in Geodesy and Surveying: A Review}},
  author={Duchnowski, Robert and Wyszkowska, Patrycja},
  journal={Measurement Science and Technology},
  volume={35},
  number={5},
  pages={052002},
  year={2024},
  publisher={IOP Publishing}
}

@article{hammond2016gps,
  title={{GPS Imaging of Vertical Land Motion in California and Nevada: Implications for Sierra Nevada Uplift}},
  author={Hammond, William C and Blewitt, Geoffrey and Kreemer, Corn{\'e}},
  journal={Journal of Geophysical Research: Solid Earth},
  volume={121},
  number={10},
  pages={7681--7703},
  year={2016},
  publisher={Wiley Online Library}
}

@article{verard2019plate,
  title={{Plate Tectonic Modelling: Review and Perspectives}},
  author={Verard, Christian},
  journal={Geological Magazine},
  volume={156},
  number={2},
  pages={208--241},
  year={2019}
}

@article{rollins2018interseismic,
  title={{Interseismic Strain Accumulation on Faults Beneath Los Angeles, California}},
  author={Rollins, Chris and Avouac, Jean-Philippe and Landry, Walter and Argus, Donald F and Barbot, Sylvain},
  journal={Journal of Geophysical Research: Solid Earth},
  volume={123},
  number={8},
  pages={7126--7150},
  year={2018},
  publisher={Wiley Online Library}
}

@article{rebischung2023reference,
  title={{Reference Frame Committee Technical Report 2024}},
  author={Rebischung, P},
  journal={IGS Central Bureau},
  pages={237},
  year={2023}
}
    }
\endgroup

\phantomsection\label{supplementary-material}
\bigskip

\newpage

\begin{center}
{\large\bf SUPPLEMENTARY MATERIAL}

\end{center}

\begin{description}

\item[Github repository] A public GitHub repository available at:\\ \href{https://github.com/SMAC-Group/paper_wamore_repro}
{\texttt{https://github.com/SMAC-Group/paper\_wamore\_repro}} with all codes to reproduce the simulations and the case study presented in the paper.


\item[\texttt{R}-package] An open-source \texttt{R} package available at:\\
\url{https://github.com/SMAC-Group/gmwmx2} that implements the proposed estimation framework and that provides functions to download and plot GNSS position time series from the Nevada Geodetic Laboratory.

\item[Web application] A web application available at:\\
\url{https://data-analytics-lab.shinyapps.io/wamore-crustal-deformation-explorer} that allows to visualize the estimated tectonic velocities and estimated crustal uplift velocity of the network of GNSS stations discussed in the case study.
\end{description}

\newpage
\section{Appendices}
\label{sec.appendix_structure}

The appendices are organized as follows:
\begin{enumerate}
    \item \textbf{WAMORE Properties} (Appendix \ref{app:gmwmx_properties}): Here we deliver statements and proofs for all the statistical properties of the WAMORE:
    \begin{enumerate}
        \item We firstly define some extra notation that is used to deliver the theoretical results of the WAMORE (Appendix \ref{app:notation});
        \item For completeness and a self-contained work, we include existing definitions and results that will be used for the new results presented in our work (Appendix \ref{app:auxiliary_results});
        \item We derive the covariance matrix $\bm{\Phi}_n$, i.e., the covariance matrix of $\hat{\bm{\beta}}$ (Appendix \ref{app:cov_beta});
        \item We provide the theoretical forms of the WV and of its covariance (Appendix \ref{app:wv_theo_forms});
        \item Based on the covariance matrix $\bm{\Phi}_n$ we deliver the asymptotic distribution of $\hat{\bm{\beta}}$ under a short-memory and long-memory regime for the error process $\bm{\varepsilon}$ (Appendices \ref{app:short_memory} and \ref{app:long_memory} respectively);
        \item We prove the consistency of the empirical WV $\hat{\boldsymbol{\nu}}$ computed on $\hat{\boldsymbol{\varepsilon}}$  (Appendix~\ref{app:consistency_nu_hat}) and the consistency of the estimated stochastic parameters $\hat{\boldsymbol{\gamma}}$ (\ref{app:consistency_of_gamma_hat});
        \item We prove the consistency of the estimator of the proposed model selection criterion for the WAMORE towards the true criterion (Appendix \ref{app:model_selection}).
    \end{enumerate}
    \item \textbf{Computational Results} (Appendix \ref{app:comp_results}): In this appendix we provide the definitions and proofs for the computational results that ensure scalability of the WAMORE:
    \begin{enumerate}

\item For any covariance matrix of the error process $\bm{\Sigma}$ (estimated or true), we derive an $\mathcal{O}(n^{-1})$-accurate approximation of the theoretical WV $\bm{\nu}(\bm{\gamma}, \bm{\vartheta})$ based solely on the averages of the diagonal and superdiagonals of $\bm{\Sigma}$ (Appendix \ref{app:fast_theo_wv});

        \item We discuss a computationally efficient approximation for the covariance matrix of the estimated residuals (not of the true error process) and define the corresponding theoretical WV based on Appendix \ref{app:fast_theo_wv} (Appendix \ref{app:IminusH});
        \item We derive the exact form of the covariance matrix of the estimated WV $\hat{\bm{\nu}}$, i.e., $\mathbf{V}$, and provide a computationally efficient approach to compute it (Appendix \ref{app:compute_v}).
    \end{enumerate}
    \item \textbf{Simulation Studies} (Appendix \ref{app:simulations}): In this appendix, we provide extensive simulation studies that showcase the performance of the WAMORE and of the proposed criterion in different settings.
    \begin{enumerate}
        \item We provide the complete results of simulation studies (Appendix \ref{app:main_simu_results});
        \item We show that the Markov chain used to generate $\mathbf{Z}$ in the simulations is a special case of \eqref{eq:markov_missinigness} (Appendix \ref{app:beta_mixing});

        \item We provide empirical results from simulation studies investigating the performance of the model selection criterion presented in Section~\ref{sec:model_selection}   (Appendix \ref{app:model_selection_simulation}).
        
    \end{enumerate}
\item \textbf{Case Study: Additional Results} (Appendix \ref{app:case_study}): 
In this appendix, we present supplementary material related to the case study introduced in Section~\ref{sec:case_study}.
\begin{enumerate}
    \item We provide additional graphs supporting the case study (Appendix \ref{app:case_study_additional_graphs});
    \item We provide results and discussion of the emulation study (Appendix \ref{app:emulation_study}).
\end{enumerate}

\end{enumerate}

\newpage
\subsection{WAMORE Statistical Properties}
\label{app:gmwmx_properties}

\subsubsection{Notation}
\label{app:notation}

\noindent \textbf{Mixing Processes}

\noindent To start, let us define $\alpha$- and $\beta$-mixing processes. Firstly, a process $\{X_i\}_{i\in{\mathbb{Z}}}$ is said to be $\alpha$-mixing if
$$\alpha_k=\sup_{t\in{\mathbb{Z}}}\alpha(\sigma(X_s,s\le t),\sigma(X_s,s\ge t+k))\to 0,$$ as $k \to \infty$, where
$
\alpha(\mathcal{A}, \mathcal{B})=\sup _{A \in \mathcal{A}, B \in \mathcal{B}}|\mathbb{P}(A \cap B)-\mathbb{P}(A) \mathbb{P}(B)|
$ for any two $\sigma$-fields $\mathcal{A}$ and $\mathcal{B}$.
Similarly, a process $\{X_t\}_{t\in{\mathbb{Z}}}$ is said to be $\beta$-mixing if
$$\beta_k=\sup_{t\in{\mathbb{Z}}}\beta(\sigma(X_s,s\le t),\sigma(X_s,s\ge t+k))\to 0,$$ as $k \to \infty$, where
$
\beta(\mathcal{A}, \mathcal{B})=2^{-1}\sup \sum_{i = 1}^I\sum_{j = 1}^J|\mathbb{P}(A_i \cap B_j)-\mathbb{P}(A_i) \mathbb{P}(B_j)|
$ for any two $\sigma$-fields $\mathcal{A}$ and $\mathcal{B}$, and  the supremum is taken over all pairs of (finite) partitions $\{A_1, \dots, A_I\}$ and $\{B_1, \dots, B_J\}$ of the sample space such that $A_i \in \mathcal{A}$ for each $i$ and $B_j \in \mathcal{B}$ for each $j$. Note that $\alpha_k$ and $\beta_k$ are referred to as $\alpha$- and $\beta$-mixing coefficients respectively and that the conditions based on the $\alpha$-mixing are weaker than that of $\beta$-mixing. More precisely, it follows that $2\alpha_k  \leq \beta_k$ for any $k \in \mathbb{N}$ \cite[e.g., ][]{bradley2005basic}.
\\
\noindent \textbf{Matrix Operations}

\noindent For a matrix $\M$, its operator norm is denoted by $\|\M\|_{\op}$. If $\M$ is a square matrix, the maximum and minimum eigenvalues of $\M$ are denoted by $\lambda_{\max}(\M)$ and $\lambda_{\min}(\M)$, respectively. The trace of $\M$ is denoted by $\operatorname{tr}(\M)$. The transpose of $\M$ is denoted by $\M^{\trans}$. For a matrix $\M \in \real^{n \times p}$, we denote its $b^{\text{th}}$ column as $\M_{\cdot b} \in \mathbb R^n$ and its $r^{\text{th}}$ row as $\M_{r \cdot} \in \real^p$.
\\
For two matrices $\A$ and $\B$ of the same dimension $m \times n$, we denote the Hadamard product as $\A \odot \B$, where $\A \odot \B$ is a matrix of dimension $m \times n$ with elements given by $(\A \odot \B)_{i, j}=(\A)_{i, j}(\B)_{i, j}$.
\\
For a $m \times n$ matrix $\A$ and a $p \times q $ matrix $\B$, we denote the  Kronecker product $\mathbf{A} \otimes \mathbf{B}$ as the $p m \times q n$  block matrix: 

\begin{equation*}
\mathbf{A} \otimes \mathbf{B} = 
\begin{pmatrix}
a_{11} \mathbf{B} & \cdots & a_{1n} \mathbf{B} \\
\vdots & \ddots & \vdots \\
a_{m1} \mathbf{B} & \cdots & a_{mn} \mathbf{B}
\end{pmatrix}
\end{equation*}

For two vectors $\x, \y \in \real^k$ and a matrix $\mathbf{A} \in \real^{k \times k}$, we write $\left\|\x-\y\right\|_{\mathbf{A}}^2 $ to denote $(\x -\y)^{\trans} \A (\x -\y)$.
\\
\noindent \textbf{Sequences}

\noindent For a sequence of random variables $\{X_n\}$ and positive numbers $\{a_n\}$, we write $X_n = O_p(a_n)$ if $\lim_{K \to \infty}\limsup_{n \to \infty}\mathbb{P}(|X_n| \geq Ka_n) = 0$. For two sequences of positive numbers $\{a_n\}$ and $\{b_n\}$, we write $a_n = O(b_n)$ if there exists some constant $C > 0$ such that $a_n/b_n \leq C$ for all large $n$.
\\
\noindent \textbf{Design-Based Matrices}
    
\noindent Let $\X \in \real^{n \times p}$ be the design matrix, $\Y \in \real^n$ be the responses, and $\Z \in \{0, 1\}^n$ be the vector indicating the missing observations. Denote $\tilde{\X} = \Z \otimes \bm{1}_p^{\trans} \odot \X \in \real^{n \times p}$ the masked design matrix and $\tilde{\Y} = \Z \odot \Y$ the masked responses.

We define a collection of matrices starting from $\D_n \in \mathbb R^{p \times p}$ which is a diagonal matrix given by
\begin{align}\label{eq:D_n}
    \D_n = \diag(\X^{\trans}\X) = \diag(\|\X_{\cdot 1}\|_2^2, \dots, \|\X_{\cdot p}\|_2^2),
\end{align}
where $\X_{\cdot b} \in \mathbb R^n$ is the $b^{\text{th}}$ column of $\X$. Moreover, we define the following matrices and vectors that will also be used throughout the corresponding proofs:
\begin{itemize}
    \item $\bm{\Pi}_n = \X^{\trans}\X \in \real^{p \times p}$;
    \item $\tilde{\bm{\Pi}}_n = \tilde{\X}^{\trans}\tilde{\X}\in \real^{p \times p}$;
    \item $\C_n = \D_n^{-1/2}\X^{\trans}\X\D_n^{-1/2} \in \real^{p \times p}$;
    \item $\Q_{i \cdot} = \bm{\Pi}_n^{-1/2}\X_{i \cdot} \in \real^{p}$ and $\tilde{\Q}_{i \cdot} = Z_i \cdot \Q_{i \cdot} \in \real^{p}$.
\end{itemize}

\subsubsection{Auxiliary Results}
\label{app:auxiliary_results}

In this appendix we list and state some auxiliary results that will be used for the proofs given in the following appendices.

\begin{Lemma}[Theorem 5.2 in \citealt{bradley2005basic}]\label{lemma:mixing_indep_comp}
Denote $X^{(k)} = \{X^{(k)}_t\}_{t \in \mathbb{Z}}$, for $k \in \mathbb N^*$. Suppose these sequences, $X^{(k)}$, are mutually independent across $k$. Suppose that for each $t \in \mathbb{Z}$, $h_t: \real\times \real \times \real \times \dots \mapsto \real$ is a Borel function. Define the sequence $X = \{X_t\}_{t \in \mathbb{Z}}$ of random variables with $X_t = h_t\big(X^{(1)}_t,X^{(2)}_t,X^{(3)}_t,\dots\big)$, for $t \in \mathbb{Z}$. Then for any $\ell \geq 1$, it follows that $\alpha_{\ell} \leq \sum_{k = 1}^{\infty}\alpha^{(k)}_{\ell}$ and $\beta_{\ell} \leq \sum_{k = 1}^{\infty}\beta^{(k)}_{\ell}$.
\end{Lemma}

\begin{Lemma}[Theorem~2.20 in \citealp{fan2008nonlinear}]\label{lemma:mixing_mement_ineq}
    Let $\{X_t\}_{t \in \mathbb{Z}}$ be a strictly stationary and $\alpha$-mixing process, with $\alpha$-mixing coefficients $\{\alpha_k\}_{k \geq 1}$ and the autocovariances $\{\gamma(k)\}_{k \in \mathbb{Z}}$.  Suppose that one of the following conditions holds:
    \begin{itemize}
        \item[(i)]
        $\mathbb{E}|X_t|^{\delta} < \infty \quad \text{and} \quad \sum_{k \geq 1}\alpha_k^{1 - 2/\delta} < \infty \quad \text{for some constant} \quad \delta > 2$.
        \item[(ii)]
        $\mathbb{P}(|X_t| < C) = 1 \quad \text{for some constant} \quad C > 0, \quad \text{and} \quad \sum_{k \geq 1}\alpha_k < \infty$.
    \end{itemize}
    Then $\sum_{k \geq 1}|\gamma(k)| < \infty$, and as $n \to \infty$
    \[
    \frac{1}{n}\var\left(\sum_{t = 1}^nX_t\right) \to \gamma(0) + 2\sum_{k = 1}^{\infty}\gamma(k).
    \]
\end{Lemma}

\begin{Lemma}[Theorem~2.21 in \citealp{fan2008nonlinear}]\label{lemma:mixing_clt}
    Assume that $\mathbb{E}X_t = 0$, and $\sigma^2 = \gamma(0) + 2\sum_{j = 1}^{\infty}\gamma(j)$ is positive. Then
    \[
    \frac{1}{\sqrt{n}}\sum_{t = 1}^nX_t \overset{d}{\to} \mathcal{N}(0, \sigma^2),
    \]
    if one of the following conditions holds:
    \begin{itemize}
        \item[(i)]
        $\mathbb{E}|X_t|^{\delta} < \infty \quad \text{and} \quad \sum_{k \geq 1}\alpha_k^{1 - 2/\delta} < \infty \quad \text{for some constant} \quad \delta > 2$.
        \item[(ii)]
        $\mathbb{P}(|X_t| < C) = 1 \quad \text{for some constant} \quad C > 0, \quad \text{and} \quad \sum_{k \geq 1}\alpha_k < \infty$.
    \end{itemize}
\end{Lemma}

\begin{Lemma}[Theorem 1 in \citealp{banna2016bernstein}]\label{lemma:matrix_bernstein}
Let $\{\M_t\}_{t \in \mathbb{N}^{+}}$ is a family of self-adjoint random matrices of size $d$. Assume that there exists a constant $c > 0$ such that for any $\ell \geq 1$, $\beta_{\M}(\ell) \leq \exp(1-c\,\ell)$, and there exists a positive constant $D$ such that for any $t \in \mathbb{N}^{+}$,
\begin{align*}
    \mathbb{E}[\M_t] = \bm{0} \quad \text{and} \quad \|\M_t\|_{\mathrm{op}} \leq D \quad \text{almost surely}.
\end{align*}
Then there exists an absolute constant $C$ such that for any $x > 0$ and any integer $n \geq 2$,
\[
\mathbb{P}\left(\left\|\sum_{t=1}^n \M_t\right\|_{\mathrm{op}} \geq x\right) \leq d\exp\left(-\frac{C\,x^2}{\iota^2\,n + c^{-1}\,D^2 + x\,D\,\gamma(c,n)}\right),
\]
where
\[
\iota^2 = \sup_{\mathcal{K} \subseteq\{1, \dots, n\}}\frac{1}{|\mathcal{K}|}\left\|\mathbb{E}\left(\sum_{i \in \mathcal{K}}\M_i\right)^2\right\|_{\mathrm{op}}
\]
and
\[
\gamma(c,n) = \frac{\log n}{\log 2}\max\left\{2, \frac{32\log n}{c\log 2}\right\}.
\]
\end{Lemma}

\begin{Lemma}[Weyl’s inequality]\label{lemma:weyl_ineq}
For symmetric matrices $\A, \B \in \mathbb R^{r \times r}$, we have that
    $$\max_{i = 1}^r|\lambda_{i}(\A) - \lambda_i(\B)| \leq \|\A - \B \|_{\op}.$$
\end{Lemma}

The next theorem is from \cite{pipiras2000convergence}, which is an invariance principle for  weighted sums of a long-memory dependent process. To state the theorem, we need the following definitions.
Define the function space
\[
|\Lambda|^{d} = \left\{ f : \int_{\real} \int_{\real} | f(u) | | f(v) | |u - v|^{2d - 1} \, \mathrm{d}u \, \mathrm{d}v < \infty \right\},
\]
for $d \in (0,1/2)$ and the associated norm
\[
\|f\|^2_{|\Lambda|^{d}} = \int_{\real} \int_{\real} | f(u) | | f(v) | |u - v|^{2d - 1} \, \mathrm{d}u \, \mathrm{d}v.
\]
For $k \in \mathbb N \cup \{\infty\}$, define the approximations of a function $f$:
\[
f^{+}_{n,k} = \sum_{i=0}^{k} f\left( \frac{i}{n} \right) \mathbbm{1}_{[i/n, (i+1)/n)},
\quad
f^{-}_{n,k} = \sum_{i=-k}^{-1} f\left( \frac{i}{n} \right) \mathbbm{1}_{[i/n, (i+1)/n)},
\]
\[
f^{+}_{n} = f^{+}_{n,\infty},
\quad
f^{-}_{n} = f^{-}_{n,\infty},
\quad
f_{n} = f^{+}_{n} + f^{-}_{n}.
\]
Moreover, define the following two-sided (normalized) partial sum process of $\{e_i\}_{i \in \mathbb Z}$ as
\[
B^{d}_n(u) =
\begin{cases}
\frac{1}{n^{d + 1/2}} \sum_{i=1}^{\lfloor nu \rfloor} e_i, & u \geq 0, \\[10pt]
-\frac{1}{n^{d + 1/2}} \sum_{i=\lfloor nu \rfloor + 1}^{0} e_i, & u < 0,
\end{cases}
\]

\begin{Theorem}[Theorem 3.1 in \citealp{pipiras2000convergence}]\label{thm:clt_long-memory}
    Let $d \in (0, 1/2)$. Let $f, f^{\pm}_n, f^{\pm}_{n,k}, f_n$ be deterministic functions defined above. Suppose that the following conditions are satisfied:
\begin{enumerate}
    \item[(i)] $f, f^{\pm}_n \in |\Lambda|^{d}$, $\|f^{\pm}_n - f^{\pm}_{n,k}\|_{|\Lambda|^{d}} \to 0$, as $k \to \infty$, $\|f - f_n\|_{|\Lambda|^{d}} \to 0$, as $n \to \infty$,
    \item[(ii)] $\{e_i\}_{i \in \mathbb{Z}}$ is a stationary process with $\mathbb E(e_i) = 0$ and $\var(e_i)<\infty$, such that $|E(e_0 e_\ell)| \leq c |\ell|^{2d - 1}$, $\ell \in \mathbb{N}$, and is such that the sequence of processes $B_n^{d}$ converges to $B^{d}$ in the sense of the finite-dimensional distributions.
\end{enumerate}
Then
\[
\frac{1}{n^{d + 1/2}} \sum_{i = -\infty}^{\infty} f\left(\frac{i}{n}\right) e_i \xrightarrow{\mathcal{D}} \int_{\real} f(u) \, \mathrm{d}B^{d}(u),
\]
where $\{B^d(u)\}_{u \in \mathbb R}$ is a two-sided standard fractional Brownian motion.
\end{Theorem}

The next Lemma is given in \cite{pollard1991asymptotics}. It provides a result for a sequence of random convex functions, ensuring that pointwise convergence in probability to a deterministic function implies uniform convergence over compact sets. We will use it to obtain consistency of the estimated stochastic parameters $\hat{\boldsymbol{\gamma}}$ in Appendix~\ref{app:consistency_of_gamma_hat}.

\begin{Lemma}[Convexity Lemma in \citealp{pollard1991asymptotics}]
\label{lemma::uniform_convergence}
    Let $\{\lambda_n(\bm\theta): \bm\theta \in \bm\Theta\}$ be a sequence of random convex functions defined on a convex, open subset $\bm\Theta \subseteq \mathbb R^d$. Suppose $\lambda(\cdot)$ is an $\mathbb R$-valued function on $\bm\Theta$ for which $\lambda_n(\bm\theta) \xrightarrow{p} \lambda(\bm\theta)$ pointwise for each $\bm\theta  \in \bm\Theta$. Then for each compact subset $K \subset \bm\Theta$,
    $$\sup_{\bm\theta \in K}|\lambda_n(\bm\theta) - \lambda(\bm\theta)| \xrightarrow{p} 0.$$
    The function $\lambda(\cdot)$ is necessarily convex on $\bm\Theta$.
\end{Lemma}

\subsubsection{Covariance of $\hat{\bm{\beta}}$ Under Missingness}
\label{app:cov_beta}

In this appendix, we deliver a collection of results that will be then used to determine the asymptotic distribution of the estimator $\hat{\bm{\beta}}$ under two different regimes for the rate of decay of the dependence structure of the error process $\{\varepsilon_i\}_{i \in \mathbb{Z}}$, more specifically the \textit{short-memory} and \textit{long-memory} regimes (see Appendices \ref{app:short_memory} and \ref{app:long_memory} respectively). For this purpose, Lemma \ref{lemma:beta_cov} determines the form of the covariance matrix of $\hat{\boldsymbol{\beta}}$ defined in \eqref{eq.beta_cov_miss}, i.e., ~$\bm{\Phi}_n$, \textit{under missingness}. For this we also provide Proposition \ref{prop:theo_wv} which provides the theoretical form of the WV under missingness which is a component of the covariance matrix $\bm{\Phi}_n$. Throughout this appendix, we consider the error process $\{\varepsilon_{i}\}_{i \in \mathbb{Z}}$ to be stationary, $\alpha$-mixing and with $\mathbb{E}(\varepsilon_i) = 0$.

To start, let $\hat{\boldsymbol{\varepsilon}}$ be the estimated error process defined as $\hat{\boldsymbol{\varepsilon}} = {\tilde{\Y}} -\tilde{\X} \hat{\boldsymbol{\beta}}$ in Section~\ref{sec:inference}, and denote $\X_{i \cdot} \in \real^p$ the $i^{\text{th}}$ row of the design matrix $\X \in \real^{n \times p}$. We now recall the assumption for which we consider the following Markov chain-based model for missingness patterns:
\begin{Assumption}[Missingness]\label{ass:missing}
    Let the missing observation process $\Z = \{Z_i\}_{i = 1}^n \subseteq \{0, 1\}$ be a stationary Markov chain in the form of 
    \begin{equation}\label{eq:markov_missinigness}
    Z_i =
\begin{cases}
    Z_{i-1}, & \text{with probability } \rho, \\
    W_i \sim \mathrm{Bernoulli}(\mu(\bm{\vartheta})), & \text{with probability } 1-\rho,
\end{cases}
\end{equation}
    with $\rho \in [0, 1)$ and $\mu(\bm{\bm{\vartheta}}) \in (0,1]$. Then the expected value $\mu(\bm{\bm{\vartheta}}) = \mathbb E[Z_i]$ is such that
    $$\mu(\bm{\bm{\vartheta}})\,\lambda_{\min}(\X^{\trans} \X) \geq C_{x}\, \alpha_n \, \max_i (\|X_{i \cdot}\|_{2}^2)\,n^{1/2}\,\log (n),$$
    where $C_{x} > 0$ is some constant and $\alpha_n > 0$ is any slowly diverging sequence.
    In addition, assume that $\Z$ is independent of the error process $\{\varepsilon_i\}_{i \in \mathbb{Z}}$.
\end{Assumption}

The implication of this assumption is essentially that the average proportion of missing observations does not affect the invertibility of the matrix $\tilde{\X}^{\trans} \tilde{\X}$,
which is needed to compute $\hat{\bm{\beta}}$ and all quantities derived from it (including its covariance matrix). 

This assumption guarantees (asymptotically) that this matrix is indeed invertibile.

Based on this, we deliver the following lemma, which provides the form of the asymptotic covariance matrix of $\hat{\bm{\beta}}$ defined in \eqref{eq.beta_cov_miss} in the presence of the missing observation process.

\begin{Lemma}\label{lemma:beta_cov}
    Suppose \Cref{ass:missing} holds. Then, as $n \rightarrow \infty$, we have that
\begin{align}
    \hat{\bm{\beta}} - \bm{\beta} = (1+o_p(1)) \mu(\bm{\vartheta})^{-1}( \X^{\trans} \X)^{-1} \X^{\trans} \tilde{\bm{\varepsilon}},
\end{align}
and hence we have that
    $$\bm{\Phi}_n = \var\left(\mu(\bm{\vartheta})^{-1}(\X^{\trans}\X )^{-1} \X^{\trans} \tilde{\bm{\varepsilon}}\right) = \mu(\bm{\vartheta})^{-2}(\X^{\trans}\X)^{-1} \X^{\trans} \bm{\Sigma}(\bm{\gamma}, \bm{\vartheta}) \X (\X^{\trans}\X)^{-1} , 
    $$
    where $\bm{\Sigma}(\bm{\gamma}, \bm{\vartheta}) = [\bm{\Lambda}(\bm{\vartheta}) + \mu(\bm{\vartheta})^2 \bm{1}\bm{1}^{\trans}] \odot \bm{\Sigma}(\bm{\gamma})$.
\end{Lemma}
\begin{Remark}
    \Cref{lemma:beta_cov} provides the explicit form of the covariance matrix for the approximation of $\hat{\bm{\beta}} - \bm{\beta}$. Later in this section, we demonstrate that, under mild conditions on both the design matrix and the error process, and with appropriate rescaling, $\hat{\bm{\beta}}$ is uniformly tight, i.e., ~$O_p(1)$. Consequently, after suitable rescaling, $\bm{\Phi}_n$'s limit represents the asymptotic variance (the precise forms under two dependence regimes are provided later). %
\end{Remark}

Before we prove Lemma \ref{lemma:beta_cov}, we firstly state and prove additional lemmas that will be needed. The first additional lemma concerns the properties of the missing observation process $\Z$. For example, in \cite{xu2021online} it was shown that the missingness pattern of \eqref{eq:markov_missinigness} is geometrically ergodic, i.e., ~the $\beta$-mixing coefficients satisfying $\beta_Z(\ell) \leq \exp(1-c\ell)$ for some constant $c > 0$ only depending on $\rho$. In the context of this work, we note that $\{Z_i\}_{i = 1}^n$ defined in \eqref{eq:markov_missinigness} is stationary, such that
\[
\mathbb{E}[Z_i] = \rho\mathbb{E}[Z_{i-1}] + (1 - \rho)\mu(\bm{\bm{\vartheta}}),
\]
which leads to
\[
\mu(\bm{\bm{\vartheta}}) = \mathbb E[Z_i].
\]
Note also, when $\rho = 0$, \eqref{eq:markov_missinigness} reduces to a model of i.i.d.~Bernoulli random variables. With these definitions we have the following auxiliary lemma.

\begin{Lemma}\label{lemma:ergodic}
    Let the missingness sequence $\Z = (Z_1, \dots, Z_n)^{\trans}$ satisfy \Cref{ass:missing}. Then $\Z$ is a geometric ergodic Markov chain, and its $\beta$-mixing coefficients satisfy $\beta_{Z}(\ell) \leq \rho^{\ell} = \exp(1-c\ell)$, for some constant $c > 0$ only depending on $\rho$. 
\end{Lemma}
\begin{proof}
    Note that $\{Z_i\}$ is a stationary Markov chain with state space $\mathcal{X} = \{0, 1\}$. By Theorem 3.7 in \cite{bradley2005basic}, it suffices to verify Doeblin’s condition \citep[e.g., ][]{rosenthal1995convergence}, which holds on the transition probability $\mathbb P(\cdot, \cdot)$. Indeed, if there exists an $\epsilon > 0$ and a probability measure $\eta(\cdot)$, such that for all $x \in \mathcal{X}$ and measurable subsets $A \subseteq \mathcal{X}$, then we have that
    $\mathbb P(x, A) \geq \epsilon\,\eta(A)$. With this, we notice that we can write \eqref{eq:markov_missinigness} equivalently as
    \[
        Z_i = (1-U)\cdot Z_{i-1} + U\cdot B,
    \]
    where $U \sim \mathrm{Bernoulli}(1-\rho)$, $B \sim \mathrm{Bernoulli}(\mu(\bm{\bm{\vartheta}}))$ and $\{Z_i\}$ are mutually independent. We have that for any $A \subseteq \mathcal{X}$
    \begin{align*}
        \mathbb P(x, A) =& \mathbb P(Z_i \in A| Z_{i-1} = x)\\
        =& \mathbb P(Z_i \in A, U = 1| Z_{i-1} = x)\\
        &+ \mathbb P(Z_i \in A, U = 0| Z_{i-1} = x)\\
        \geq& \mathbb P(Z_i \in A, U = 1| Z_{i-1} = x)\\
        =& \mathbb P(B \in A, U = 1)\\
        =& (1-\rho)\mathbb P(B \in A).
    \end{align*}
    Thus, we have $\epsilon = 1-\rho > 0$, which concludes the proof.
\end{proof}

The following auxiliary lemma on the other hand provides a probabilistic guarantee on the structure of the design matrix under missingness (ensuring that the estimator $\hat{\bm{\beta}}$ can be computed with high probability in these settings).

\begin{Lemma}\label{lemma:restricted_eigen}
    Suppose Assumption~\ref{ass:missing} holds. Then with probability $1 - c\,n^{-3}$
    \begin{align*}
    \lambda_{\min}(\tilde{\X}^{\trans} \tilde{\X}) \geq C_{\tilde{x}}\,\mu(\bm{\vartheta})\,\lambda_{\min}(\X^{\trans} \X),
\end{align*}
where $c > 0$ and $C_{\tilde{x}} \in (0,1)$ are absolute constants.
\end{Lemma}
\begin{proof}[Proof of \Cref{lemma:restricted_eigen}]
    It follows from Assumption~\ref{ass:missing} that $\Z$ is a strictly stationary, finite state Markov chain. Then, Theorem 3.7 in \cite{bradley2005basic} implies that $\{Z_i\}_{i = 1}^n$ is a geometric ergodic sequence, i.e., ~its $\beta$-mixing coefficients with lag $k$ satisfy the condition $\beta_Z(k) \leq \exp(1-c\,k)$. We apply Lemma~\ref{lemma:matrix_bernstein} with $\M_i = (Z_i - \mu(\bm{\vartheta}))\X_{i \cdot}\X_{i \cdot}^{\trans}$. Note that
$$\max_i\|(Z_i - \mu(\bm{\vartheta}))\X_{i \cdot}\X_{i \cdot}^{\trans}\|_{\op} \leq \max_i \{\|\X_{i \cdot}\X_{i \cdot}^{\trans}\|_{\op}\} = \max_i\{\|\X_{i \cdot}\|_2^2\},$$
and
$$\iota^2 \leq C_{z}\max_i\{\|\X_{i \cdot}\X_{i \cdot}^{\trans}\|_{\op}^2\} = C_{z}\max_i\{\|\X_{i \cdot}\|_2^4\},$$
where $C_{z} > 0$ is a constant only depending on the Markov model of $\Z$. By Lemma~\ref{lemma:matrix_bernstein}, we have that with probability at least $1 - p\,n^{-3}$
\begin{align}\label{eq:event1}
    \|\tilde{\X}^{\trans} \tilde{\X} -\mu(\bm{\vartheta}) \X^{\trans} \X\|_{\op} \leq  C_{\op}\max_i\{\|\X_{i \cdot}\|_{2}^2\}\,n^{1/2}\,\log n,
\end{align}
for some constant $C_{\op} > 0$.
By Lemma~\ref{lemma:weyl_ineq}, we have that
\begin{align*}
    \lambda_{\min}(\tilde{\X}^{\trans} \tilde{\X}) \geq & \,\mu(\bm{\vartheta})\,\lambda_{\min}(\X^{\trans} \X) - \|\tilde{\X}^{\trans} \tilde{\X} - \mu(\bm{\vartheta}) \X^{\trans} \X\|_{\op}\\
    \geq & \,\mu(\bm{\vartheta})\,\lambda_{\min}(\X^{\trans} \X) - C_{\op}\,\max_i\|\X_{i \cdot}\|_{2}^2\,n^{1/2}\,\log n\\
    \geq & \, (1 - C_{\op}\,C^{-1}_{x}\,\alpha_n^{-1})\,\mu(\bm{\vartheta})\,\lambda_{\min}(\X^{\trans} \X),
\end{align*}
where the last inequality follows from Assumption \ref{ass:missing}.
\end{proof}

With this additional lemma we can now proceed to the proof of \Cref{lemma:beta_cov}. 

\begin{proof}[Proof of \Cref{lemma:beta_cov}]
Assume that $\tilde{\X}^{\trans} \tilde{\X}$ is non-singular, which holds with high probability due to \Cref{lemma:restricted_eigen}. Based on this, the least-squares estimator satisfies the following:
\begin{align}\label{eq:beta}
    \hat{\bm{\beta}} - \bm{\beta} =& (\tilde{\X}^{\trans} \tilde{\X})^{-1}\tilde{\X}^{\trans}\tilde{\Y} - \bm{\beta} = (\tilde{\X}^{\trans} \tilde{\X})^{-1}\tilde{\X}^{\trans}(\tilde{\X}\bm{\beta}+ \Z \odot \bm{\varepsilon}) - \bm{\beta}\nonumber\\
    =& (\tilde{\X}^{\trans} \tilde{\X})^{-1}\tilde{\X}^{\trans}(\Z \odot \bm{\varepsilon})\nonumber\\
    =& (\tilde{\X}^{\trans} \tilde{\X})^{-1}\tilde{\X}^{\trans}\tilde{\bm{\varepsilon}}\nonumber\\
    =& \left\{(\tilde{\X}^{\trans} \tilde{\X})^{-1} - \mu(\bm{\vartheta})^{-1}( \X^{\trans}\X)^{-1}\right\} \X^{\trans} \tilde{\bm{\varepsilon}} + \mu(\bm{\vartheta})^{-1}(\X^{\trans}\X)^{-1}\X^{\trans} \tilde{\bm{\varepsilon}} \nonumber\\ 
    =& (1+o_p(1)) \mu(\bm{\vartheta})^{-1}( \X^{\trans} \X)^{-1} \X^{\trans} \tilde{\bm{\varepsilon}},
\end{align}
where the last equality follows from the following fact. Since $\|(\A^{-1} - \B^{-1}) \B\|_{\op} \leq \lambda_{\min}^{-1}(\A)\|\A - \B\|_{\op}$, we have that
\begin{align}\label{eq:rev_ratio}
    &\left\|\left\{(\tilde{\X}^{\trans} \tilde{\X})^{-1} - \mu(\bm{\vartheta})^{-1}(\X^{\trans} \X)^{-1}\right\}\mu(\bm{\vartheta}) \X^{\trans}\X\right\|_{\op}\nonumber\\
    \leq& \lambda^{-1}_{\min}(\tilde{\X}^{\trans} \tilde{\X})\left\|\tilde{\X}^{\trans} \tilde{\X} - \mu(\bm{\vartheta}) \X^{\trans} \X\right\|_{\op}\nonumber\\
    =& O_p\left((1 - C_{\op}C^{-1}_{\mathrm{gap}}\alpha_n^{-1})^{-1}\mu(\bm{\vartheta})^{-1}\lambda^{-1}_{\min}(\X^{\trans} \X)\,C_{\op}\,\max_i\{\|\X_{i \cdot}\|_{2}\}\, n^{1/2}\, \log n\right)\nonumber\\
    =& O_p\left(C_{\op}C^{-1}_{\mathrm{gap}}\alpha^{-1}_n(1 - C_{\op}C^{-1}_{\mathrm{gap}}\alpha_n^{-1})^{-1}\right)\nonumber\\
    =& o_p(1),
\end{align}
where the first equality follows from \eqref{eq:event1} in the proof of \Cref{lemma:restricted_eigen}, and the second equality follows from \Cref{ass:missing}. 
Therefore, we have $$\hat{\bm{\beta}} - \bm{\beta} - \mu(\bm{\vartheta})^{-1}(\X^{\trans}\X)^{-1} \X^{\trans} \tilde{\bm{\varepsilon}} = o_p(1)\,\mu(\bm{\vartheta})^{-1}(\X^{\trans} \X)^{-1} \X^{\trans} \tilde{\bm{\varepsilon}} \overset{p}{\to} 0.$$ 
In addition, we have
\begin{align*}
    \bm{\Phi}_n =&  
 \var\left(\mu(\bm{\vartheta})^{-1}(\X^{\trans} \X )^{-1}\X^{\trans} \tilde{\bm{\varepsilon}} \right)\\
    =&
    \mu(\bm{\vartheta})^{-2}\mathbb E\left[(\X^{\trans} \X)^{-1}\tilde{\X}^{\trans} \bm{\Sigma}(\bm{\gamma})  \tilde{\X} ( \X^{\trans} \X)^{-1} \right]\\
    =&  \mu(\bm{\vartheta})^{-2}\mathbb E\left[(\X^{\trans}\X)^{-1}(\Z \otimes \bm{1}_p^{\trans} \odot \X)^{\trans} \bm{\Sigma}(\bm{\gamma})  (\Z \otimes \bm{1}_p^{\trans} \odot \X) ( \X^{\trans} \X)^{-1} \right]\\
    =& \mu(\bm{\vartheta})^{-2}(\X^{\trans} \X)^{-1} \X^{\trans} \left\{\mathbb{E}[\Z\Z^{\trans}] \odot \bm{\Sigma}(\bm{\gamma})\right\}  \X ( \X^{\trans} \X)^{-1}\\
    =& \mu(\bm{\vartheta})^{-2}(\X^{\trans} \X)^{-1} \X^{\trans} \left\{(\bm{\Lambda}(\bm{\vartheta}) + \mu(\bm{\vartheta})^2\bm{1}\bm{1}^{\trans}) \odot \bm{\Sigma}(\bm{\gamma})\right\} \X ( \X^{\trans} \X)^{-1}.
\end{align*}
\end{proof}
To obtain the finite sample covariance matrix $\bm{\Phi}_n$, we need the corresponding form of the theoretical WV in order to obtain $\hat{\boldsymbol{\gamma}}$ and then use $\boldsymbol{\Sigma}(\hat{\boldsymbol{\gamma}})$ in the expression of the covariance of $\hat{\boldsymbol{\beta}}$. For this reason, in the appendix section below we put forward a general form of the theoretical WV for a zero-mean stochastic process (as is the case of the error process $\bm{\varepsilon}$) as well as of its covariance matrix. 

\subsubsection{Theoretical WV Forms}
\label{app:wv_theo_forms}

Let $\h_{j} = \left(h_{j,L_j - 1}, \hdots, h_{j, 0}\right)^{\trans} \in \real^{L_j}$ be a generic wavelet filter vector at scale $j \in \mathbb{N}^+$ of length $L_j$ and let $M_j$ represent the number of wavelet coefficients issued from the wavelet decomposition at scale $j$. This allows us to define 
\begin{equation}
    \mathbf{A}_j = \sum_{l=0}^{L_j-1}\frac{h_{j,l}}{M_j} \H_{j,l},
    \label{eq:filt_mat}
\end{equation}
where $\H_{j,l}$ is a $n \times n$ matrix such that its $n \times (n-L_j+1)$ partitioned matrix with its $1^{\text{st}}$ to $n^{\text{th}}$ rows and its $(L_j-l)^{\text{th}}$ to $(n-l)^{\text{th}}$ columns is given by $$\begin{pmatrix} h_{j,L_j-1} & 0           & \ldots & 0    \\ h_{j,L_j-2} & h_{j,L_j-1} & \ldots & 0     \\  h_{j,L_j-3} & h_{j,L_j-2} & \ldots & 0       \\ \vdots      & \vdots      & \vdots & \vdots      \\  \ldots      & \ldots      & \ldots & h_{j,L_j-1}  \\ \vdots      & \vdots      & \vdots & \vdots      \\ 0           & 0           & 0      & h_{j,0} \end{pmatrix}.$$
With the above definitions, we can deliver the general form of the theoretical WV for a zero-mean stochastic process in the following proposition.

\begin{Proposition}
\label{prop:general_form_theo_wv}
    Assume that $\bm{\varepsilon} = (\varepsilon_1, \varepsilon_2, \ldots, \varepsilon_n)^{\trans} \in \real^n$ is a zero-mean stochastic process with covariance matrix $\bm{\Sigma}(\bm{\gamma}) \in \real^{n \times n}$ parametrized by $\bm{\gamma} \in \real^q$, then the theoretical WV at scale $j$ is given by
    $$\nu_j(\bm{\gamma}) = \tr[\A_j \bm{\Sigma}(\bm{\gamma})].$$
    \label{prop:theo_wv}
\end{Proposition}

\begin{proof}
    The proof of this proposition is a result of the expectation of a quadratic form of the WV for a zero-mean error process such as $\bm{\varepsilon}$. Indeed, recalling that the $t^{\text{th}}$ wavelet coefficient at the $j^{\text{th}}$ scale is defined as 
\begin{equation*}
    W_{j, t} = \sum_{l=0}^{L_j-1}h_{j,l}\varepsilon_{t-l},
\end{equation*}
with $t \in \{L_j, \ldots, n\}$, we can write the vector of $M_j$ wavelet coefficients at scale $j$ as:
\begin{equation*}
\begin{aligned}
        \W_j &= \begin{pmatrix}
W_{j, L_j}\\
W_{j, L_j+1}\\
W_{j, L_j+2}\\
\vdots \\
W_{j, n}
\end{pmatrix}
=
\underbrace{\begin{pmatrix}
h_{j, L_j-1} & h_{j, L_j-2} & \ldots & h_{j,0} & 0 & 0 &  \ldots & 0 \\
0 & h_{j, L_j-1} & h_{j, L_j-2} & \ldots & h_{j,0} & 0 &  \ldots & 0\\
0 & 0 & h_{j, L_j-1} & h_{j, L_j-2} & \ldots & h_{j,0}  &  \ldots & 0\\
0 & 0 & 0 &\ldots  &  h_{j, L_j-1} & h_{j, L_j-2} & \ldots & h_{j,0}\\
\end{pmatrix}}_{=\B_j}
\underbrace{\begin{pmatrix}
\varepsilon_1 \\
\varepsilon_2 \\
\vdots\\
\varepsilon_n
\end{pmatrix}}_{\boldsymbol{\varepsilon}}
\end{aligned}
\end{equation*}
where $\B_j$ is an $M_j \times n$ matrix that contains the Haar wavelet filters. Using this expression of the wavelet coefficients, we can rewrite the theoretical wavelet variance as:
\begin{equation*}
  \begin{aligned}
        \boldsymbol{\nu}_j &= \mathbb{E}\left[  \frac{1}{M_j}\sum_{i=L_j}^n W_{j, i}^2\right]       =\mathbb{E}\left[  \frac{1}{M_j} \W_j^{\trans} \W_j \right] \\
        &= \frac{1}{M_j} \mathbb{E} \left[(\B_j \boldsymbol{\varepsilon})^{\trans} \B_j \boldsymbol{\varepsilon}\right] = \frac{1}{M_j} \mathbb{E} \left[ \boldsymbol{\varepsilon}^{\trans} \B_j^{\trans}  \B_j \boldsymbol{\varepsilon} \right] \\
        &= \frac{1}{M_j} \mathbb{E} \left[ \tr \left\{\boldsymbol{\varepsilon}^{\trans} \B_j^{\trans}  \B_j \boldsymbol{\varepsilon} \right\}\right] =  \frac{1}{M_j} \mathbb{E} \left[ \tr \left\{
       \B_j^{\trans}   \B_j  \boldsymbol{\varepsilon} \boldsymbol{\varepsilon}^{\trans}  \right\}\right] \\
       &=   \frac{1}{M_j}  \tr \left\{
       \B_j^{\trans}   \B_j  \mathbb{E} \left[  \boldsymbol{\varepsilon} \boldsymbol{\varepsilon}^{\trans}  \right] \right\} = \tr[\A_j \bm{\Sigma}(\bm{\gamma})]
  \end{aligned}
\end{equation*}
where $\A_j = \frac{1}{M_j} \B_j^{\trans}   \B_j $.
\end{proof}
Using this result together with the covariance of $\tilde{\boldsymbol{\varepsilon}}$, we obtain:

\begin{equation}
    \nu_j\left(\boldsymbol{\gamma},{\boldsymbol{\vartheta}} \right) = \text{tr}\left\{\mathbf{A}_j \left[\boldsymbol{\Sigma}\left(\boldsymbol{\gamma}\right) \odot \left\{\boldsymbol{\Lambda}(\boldsymbol{\vartheta}) + \mu(\boldsymbol{\vartheta})^2 \mathbf{1} \mathbf{1}^{\trans} \right\}\right] \right\}.
    \label{eq:wv_missing_in_appendix}
\end{equation}

Building on the quadratic form representation derived above, the empirical
wavelet variance at scale $j$ can be written as
\begin{equation*}
    \hat{\nu}_j
    =\frac{1}{M_j}\sum_{i=L_j}^n W_{j, i}^2 = 
    \frac{1}{M_j}\mathbf{W}_j^{\trans}\mathbf{W}_j
    =
    \boldsymbol{\varepsilon}^{\trans}\mathbf{A}_j\boldsymbol{\varepsilon}.
\end{equation*}
Let $\mathbf{V} = \var(\hat{\boldsymbol{\nu}})$ denote the covariance matrix of
the empirical wavelet variance across scales. The $(j,l)^{\text{th}}$ element of
$\mathbf{V}$ is therefore given by
\begin{equation*}
    u_{j,l}
    =
    \cov\left(
    \boldsymbol{\varepsilon}^{\trans}\mathbf{A}_j\boldsymbol{\varepsilon},
    \boldsymbol{\varepsilon}^{\trans}\mathbf{A}_l\boldsymbol{\varepsilon}
    \right).
\end{equation*}

Assuming that $\boldsymbol{\varepsilon}$ is a zero-mean Gaussian process with
covariance matrix $\boldsymbol{\Sigma}(\boldsymbol{\gamma})$, the covariance of
two quadratic forms yields the closed-form expression
\begin{equation}
\label{eq:cov_wv_appendix}
    u_{j,l}
    =
    2 \tr\left\{
    \mathbf{A}_j \boldsymbol{\Sigma}(\boldsymbol{\gamma})
    \mathbf{A}_l \boldsymbol{\Sigma}(\boldsymbol{\gamma})
    \right\}.
\end{equation}

\subsubsection{Short-Memory Regime}
\label{app:short_memory}

The following assumption specifies the short-memory regime on the error process.

\begin{Assumption}[Short-Memory Error Process]\label{ass:epsilon}
    Let $\{\varepsilon_{i}\}_{i \in \mathbb{Z}}$ be the error process that is strictly stationary and $\alpha$-mixing and satisfying one of the following conditions: \begin{itemize}
        \item[(i)]
        $\mathbb{E}|\varepsilon_i|^{\delta} < \infty \quad \text{and} \quad \sum_{k \geq 1}\alpha_k^{1 - 2/\delta} < \infty \quad \text{for some constant} \quad \delta > 2$.
        \item[(ii)]
        $\mathbb{P}(|\varepsilon_i| < C) = 1 \quad \text{for some constant} \quad C > 0, \quad \text{and} \quad \sum_{k \geq 1}\alpha_k < \infty$.
    \end{itemize}
\end{Assumption}
The above assumption is commonly used to ensure the validity of the classical central limit theorem under temporal dependence (e.g., \Cref{lemma:mixing_clt}) and the existence of the associated asymptotic variance (e.g., \Cref{lemma:mixing_mement_ineq}).

Next, we state the assumptions which place conditions on the behavior of the design matrix $\X$ and its derived quantities.

\begin{Assumption}[Fixed Design]\label{ass:X}
    Let $\X \in \real^{n \times p}$ be a deterministic covariate process. 
    Suppose there exist absolute constants $\{\tau_{jk}\}_{1 \leq j,k \leq p}$ such that as $n \to \infty$
    \begin{itemize}
        \item $\frac{1}{n}\sum_{i = 1}^nX_{ij}^2 \to \tau_{jj} > 0$, for any $j$;
        \item $\frac{1}{n}\sum_{i = 1}^nX_{ij}X_{ik} \to \tau_{jk}$, for any $j,k$;
        \item $\max_i(\X_{i \cdot}^{\trans} \D_{n}^{-1} \X_{i \cdot}) = o(1)$.
    \end{itemize}
    Moreover, suppose there exists a nondegenerate deterministic matrix $\C \in \mathbb R^{p \times p}$ such that
$$\C = \lim_{n \to \infty} \C_n = \lim_{n \to \infty}\D_n^{-1/2}\X^{\trans}\X\D_n^{-1/2},$$
where $\D_n$ is defined in \eqref{eq:D_n}.
\end{Assumption}

\begin{Assumption}\label{ass:lagk}
There exists, for each $k \in \mathbb{Z}$, $\bm{\Psi}_k \in \real^{p \times p}$, such that
    \begin{align*}
        \lim_{n \to \infty}\sum_{i = 1}^{n - |k|}\Q_{i \cdot}\Q_{i+k \cdot}^{\top} = \bm{\Psi}_k.
    \end{align*}
\end{Assumption}

The above two assumptions are relatively weak and essentially ensure that the design matrix is well defined in order to obtain stable quantities for deriving the asymptotic distribution of $\hat{\bm{\beta}}$. In particular, deterministic regressors may be rescaled without loss of generality, since any deterministic rescaling can be absorbed into the corresponding regression coefficient. This includes trend regressors, which may be normalized by the observation span in the asymptotic analysis.

\begin{Theorem}[Asymptotic Normality -- Short-memory]\label{theorem:short_memory}
Under Assumptions \ref{ass:epsilon} and \ref{ass:X} we have that
    \begin{align*}
         \bm{\Pi}_n^{1/2}(\hat{\bm{\beta}} - \bm{\beta}) = (1+o_p(1))\mu(\bm{\vartheta})^{-1}\sum_{i = 1}^n \varepsilon_{i}\tilde{\Q}_{i \cdot} = O_p(1).
    \end{align*}
In addition, if \Cref{ass:lagk} holds, then we have that
\begin{align}\label{eq:apply_CLT}
    \tilde{\bm{\Pi}}_n^{1/2}(\hat{\bm{\beta}} - \bm{\beta}) \overset{d}{\to} \mathcal{N}(\bm{0}, \bar{\bm{\Phi}}),
\end{align}
where $\bar{\bm{\Phi} }= \sum_{k \in \mathbb{Z}}\mathbb{E}[Z_0Z_k]\mathbb{E}[\varepsilon_{0}\varepsilon_{k}]\bm{\Psi}_k$. Consequently,
\begin{align*}
    \D_n^{1/2}(\hat{\bm{\beta}} - \bm{\beta}) 
    \xrightarrow{\mathcal{D}} \mu(\bm{\vartheta})^{-1}\C^{-1/2}\mathcal{N}(\bm{0}, \bar{\bm{\Phi}}).
\end{align*}
\end{Theorem}

We underline that as $n \to \infty$, it can easily be shown that $\D_n^{1/2}\bm{\Phi}_n\D_n^{1/2} \to \mu(\bm{\vartheta})^{-2}\C^{-1/2} \bar{\bm{\Phi}}\C^{-1/2}$, where $\bm{\Phi}_n$ is given in \Cref{lemma:beta_cov}.

\begin{proof}[Proof of \Cref{theorem:short_memory}]
   
We recall that
\[
\bm{\Pi}_n = \X^{\trans}\X
\]
Combining the expansion in \eqref{eq:beta} with the definition $\Q_{i\cdot}=\bm{\Pi}_n^{-1/2}\X_{i\cdot}$ gives
\begin{align*}
    \bm{\Pi}_n^{1/2}(\hat{\bm{\beta}} - \bm{\beta})
    =& (1+o_p(1)) \mu(\bm{\vartheta})^{-1}\sum_{i =1}^n \tilde{\varepsilon}_i\Q_{i\cdot}.
\end{align*}
We now show that the leading sum is $O_p(1)$. Since $\tilde{\varepsilon}_i=Z_i\varepsilon_i$ and $\{Z_i\}$ is independent of $\{\varepsilon_i\}$, Lemma~\ref{lemma:mixing_indep_comp} and Assumption~\ref{ass:epsilon} imply that the autocovariances $\tilde{\gamma}(k)=\cov(\tilde{\varepsilon}_0,\tilde{\varepsilon}_k)$ are absolutely summable. Hence, for each $j=1,\ldots,p$,
\begin{align*}
\var\left(\sum_{i=1}^n\tilde{\varepsilon}_i Q_{ij}\right)
&=\sum_{i=1}^n\sum_{\ell=1}^n Q_{ij}Q_{\ell j}\tilde{\gamma}(\ell-i)\\
&\leq \sum_{k=-(n-1)}^{n-1}|\tilde{\gamma}(k)|\sum_{i=1}^{n-|k|}|Q_{ij}Q_{i+|k|,j}|\\
&\leq \sum_{k=-\infty}^{\infty}|\tilde{\gamma}(k)| < \infty,
\end{align*}
where the last inequality uses Cauchy's inequality and $\sum_{i=1}^nQ_{ij}^2=1$. Since $p$ is fixed, this proves $\sum_{i=1}^n\tilde{\varepsilon}_i\Q_{i\cdot}=O_p(1)$.

The proof of the limiting distribution of $\sum_{i = 1}^n\tilde\varepsilon_{i}\tilde{\Q}_{i \cdot}$ follows that of Theorem~1(i) in \cite{hannan1973asymptotic}, which relies on the requirement that, if conditioning on the random missingness vector $\Z$, we have
\begin{align}\label{eq:rescaling_condition}
    \max_{i \leq n}(\tilde{\X}_{i \cdot}^{\trans} \tilde{\bm{\Pi}}_n^{-1} \tilde{\X}_{i \cdot})^{1/2} = o(1).
\end{align}

To link the above requirement to \Cref{ass:X},
we will apply \Cref{lemma:matrix_bernstein} to obtain a high probability bound on $\|\tilde{\bm{\Pi}}_n - \mu(\bm{\vartheta}) \bm{\Pi}_n\|_{\op}$ and a lower bound of $\lambda_{\min}(\tilde{\bm{\Pi}}_n)$ using \Cref{lemma:weyl_ineq}. It follows from Theorem 3.7 in \cite{bradley2005basic} that the strictly stationary, finite state Markov chain $\Z = \{Z_i\}$ is geometrically ergodic, i.e., ~$\beta_Z(\ell) \leq \exp(1-c\ell)$ (see Lemma \ref{lemma:ergodic}). We apply \Cref{lemma:matrix_bernstein} with $\M_i = (Z_i - \mu(\bm{\vartheta}))\X_{i \cdot}\X_{i \cdot}^{\trans}$. We have that
$$\max_i\|(Z_i - \mu(\bm{\vartheta}))\X_{i \cdot}\X_{i \cdot}^{\trans}\|_{\op} \leq \max_i\|\X_{i \cdot}\X_{i \cdot}^{\trans}\|_{\op},$$
and
$$\iota^2 \leq C\max_i\|\X_{i \cdot}\X_{i \cdot}^{\trans}\|_{\op}^2,$$
where $C > 0$ is a constant only depending on the Markov model of $Z_i$. By \Cref{lemma:matrix_bernstein}, we have that with probability at least $1 - p\,n^{-3}$
\begin{align}\label{eq:event}
    \|\tilde{\bm{\Pi}}_n - \mu(\bm{\vartheta})\bm{\Pi}_n\|_{\op} \leq C_{\op}\max_i\|\X_{i \cdot}\X_{i \cdot}^{\trans}\|_{\op}n^{1/2}\log n,
\end{align}
for some sufficient large constant $C_{\op} > 0$. The following proof proceeds by conditioning on the high probability event in \eqref{eq:event}. By \Cref{lemma:weyl_ineq}, we have that
\begin{align*}
    \lambda_{\min}(\tilde{\bm{\Pi}}_n) & \geq \mu(\bm{\vartheta})\lambda_{\min}(\bm{\Pi}_n) - \|\tilde{\bm{\Pi}}_n - \mu(\bm{\vartheta})\bm{\Pi}_n\|_{\op}\\
    & \geq \mu(\bm{\vartheta})\lambda_{\min}(\bm{\Pi}_n) - C_{\op}\max_i\|\X_{i \cdot}\X_{i \cdot}^{\trans}\|_{\op}n^{1/2}\log n\\
    & = C^{\prime}\mu(\bm{\vartheta})\lambda_{\min}(\bm{\Pi}_n),
\end{align*}
where the equality follows from \Cref{ass:missing}, i.e., $$\mu(\bm{\bm{\vartheta}})\,\lambda_{\min}(\X^{\trans} \X) \geq C_{\mathrm{gap}} \, \alpha_n \, \max_i\{\|\X_{i \cdot}\|_{2}^2\}n^{1/2} \, \log (n),$$ and $C^{\prime} \in (0,1)$ is an absolute constant. Consequently, we have that
$$\max_{i \leq n}(\tilde{\X}_{i \cdot}^{\trans} \tilde{\bm{\Pi}}_n^{-1} \tilde{\X}_{i \cdot})^{1/2} \leq \max_{i \leq n}(\X_{i \cdot}^{\trans} \tilde{\bm{\Pi}}_n^{-1} \X_{i \cdot})^{1/2} \leq \sqrt{\frac{\max_{i \leq n}|\X_{i \cdot}|_2^2}{C^{\prime}\mu(\bm{\vartheta})\lambda_{\min}(\bm{\Pi}_n)}}  = o(1).$$

\noindent The proof of \eqref{eq:apply_CLT} then follows that of Theorem~1(i) in \cite{hannan1973asymptotic}.
\end{proof}

\subsubsection{Long-Memory Regime}
\label{app:long_memory}

For the long-memory regime we will only consider the one-sided linear process model for the error process $\{\varepsilon_i\}_{i \in \mathbb{Z}}$ which is defined as follows:
\begin{align}\label{eq:error_lp}
    \varepsilon_i = \sum_{k = 0}^{\infty}a_k e_{i-k},
\end{align}
where $\{e_h\}_{h \in \mathbb{Z}}$ are i.i.d.~innovations such that $\mathbb E(e_h) = 0$ and $\var(e_h) = \sigma_{e}^2 < \infty$. The coefficient sequence is defined by $a_k=0$ for $k<0$, by a finite normalization $a_0$ (for instance, $a_0=1$), and, for positive lags, by the asymptotic relation
\begin{align*}
    a_k \sim c_a\,k^{d-1}, \qquad k\to\infty,
\end{align*}
where $c_a$ is some constant and $0 < d < 1/2$. Then, for $\ell \in \mathbb Z$ with $|\ell| \to \infty$, the autocovariance at lag-$\ell$ satisfies:
\[
    \cov(\varepsilon_i, \varepsilon_{i+\ell}) = \sigma_{e}^2\sum_{r = 0}^{\infty}a_r\,a_{r+|\ell|} \asymp C_{a}\frac{\sigma_{e}^2}{|\ell|^{1-2d}},
\]
where $C_{a} > 0$ is a constant and $r$ is a dummy summation index. Note that when $0 < d < 1/2$ we have that $\cov(\varepsilon_i, \varepsilon_{i+\ell})$ is not summable and $\{\varepsilon_i\}_{i \in \mathbb{Z}}$ is a long-memory process. In contrast, when $d < 0$, we have that $\cov(\varepsilon_i, \varepsilon_{i+\ell})$ is summable and $\{\varepsilon_i\}_{i \in \mathbb{Z}}$ is a short-memory process. We refer to \cite{beran2013long} for a comprehensive review of long-memory processes. In this regime, we define an additional assumption.

\begin{Assumption}[Long-Memory Error Process]\label{ass:epsilon_long-memory}
    $\{\varepsilon_{i}\}_{i = 1}^n \subset \real$ is a stationary process in the form of \eqref{eq:error_lp} such that $\mathbb{E}(\varepsilon_{i}) = 0$, $\|\varepsilon_{i}\|_2 < \infty$ and with $0 < d < 1/2$.
\end{Assumption}

With this assumption, which basically formalizes the conditions of the long-memory process we use for this section, we can now deliver the following result.

\begin{Theorem}[Asymptotic Distribution -- Long-Memory]\label{theorem:long_memory}
Let the fixed design $\X$ satisfy \Cref{ass:X}. Let the error process $\{\varepsilon_i\}_{i = 1}^n$ satisfy \Cref{ass:epsilon_long-memory}. Then, with $\D_n = \diag(\X^{\trans}\X)$, $\C$ given in \Cref{ass:X} and for $j \in \{1, \dots, p\}$
\[
g_j(u) = \lim_{n \to \infty}X_{\lceil nu \rceil, j}, \;\; G_j(u) = \frac{g_j(u)}{\sqrt{\int_0^1g_j^2(u)\,\mathrm{d}u}}, \;\;  \G(u) = (G_1(u), \dots, G_p(u))^{\trans}, \; \text{for} \; u \in (0,1],
\]
we have that
\begin{align*}
    n^{-d}\D_n^{1/2}(\hat{\bm{\beta}} - \bm{\beta}) \xrightarrow{\mathcal{D}} \C^{-1}\int_{0}^1 \G(u) \, \mathrm{d}B^{d}(u),
\end{align*}
where $\{B^d(u)\}_{u \geq 0} \subset \real$ is a standard fractional Brownian motion with Hurst index $H = d + 1/2$.
\end{Theorem}

\begin{Remark}[Gaussianity of the limit distribution]
    We note that the standard fractional Brownian motion $B^d(\cdot)$ is a centered Gaussian process.  Note also $\G: [0,1] \to \real^p$  is a deterministic vector-valued function and is squared integrable.  We have 
    \[
    \mathcal{L}(\G) = \int_{0}^1 \G(u) \, \mathrm{d}B^{d}(u)
    \]
    is a linear functional of a Gaussian process. Thus, $\mathcal{L}(\G)$ is multivariate normal with mean zero and variance
    \[
    \var\left(\mathcal{L}(\G)\right) = 2d(d+1/2)\int_{0}^1\int_0^1\G(u)(\G(v))^{\trans}|u - v|^{2d-1}\,\mathrm{d}u\,\mathrm{d}v.
    \]
    See \cite{pipiras2001classes} for details.
\end{Remark}

\begin{proof}
       Recall from \eqref{eq:beta} that
    the least squares estimator satisfies the following
\begin{align*}
    \hat{\bm{\beta}} - \bm{\beta} 
    =& (1+o_p(1)) \, \mu(\bm{\vartheta})^{-1}( \X^{\trans} \X )^{-1} \X^{\trans}\tilde{\bm{\varepsilon}}.
\end{align*}
Then
\begin{align*}
    \D_n^{1/2}(\hat{\bm{\beta}} - \bm{\beta}) 
    =& o_p(1) \mu(\bm{\vartheta})^{-1}(\D_n^{-1/2} \X^{\trans} \X \D_n^{-1/2})^{-1} \D_n^{-1/2}\X^{\trans} \tilde{\bm{\varepsilon}}\\
    &+ \mu(\bm{\vartheta})^{-1}(\D_n^{-1/2} \X^{\trans}  \X\D_n^{-1/2})^{-1}\D_n^{-1/2}\X^{\trans} \tilde{\bm{\varepsilon}}\\
    =& o_p(1) \mu(\bm{\vartheta})^{-1}\C_n^{-1} \D_n^{-1/2}\X^{\trans} \tilde{\bm{\varepsilon}} + \mu(\bm{\vartheta})^{-1}\C_n^{-1}\D_n^{-1/2}\X^{\trans} \tilde{\bm{\varepsilon}}\\
    =& \underbrace{o_p(1) \mu(\bm{\vartheta})^{-1} \C_n^{-1} \G_n \tilde{\bm{\varepsilon}}}_{I} + \underbrace{\mu(\bm{\vartheta})^{-1} \C_n^{-1}\G_n \left\{\tilde{\bm{\varepsilon}} - \mu(\bm{\vartheta})\bm{\varepsilon}\right\}}_{II} + \underbrace{\C_n^{-1}\G_n  \bm{\varepsilon}}_{III},
\end{align*}
where we let $\C_n = \D_n^{-1/2}\X^{\trans}\X \D_n^{-1/2} \in \mathbb R^{p \times p}$ and $\G_n = \D_n^{-1/2}\X^{\trans} \in \mathbb R^{p \times n}$. 
For $b = 1, \dots, p$, define the function $G_{n_b}: \mathbb R \mapsto \mathbb R$, such that, for $i = 1, \dots, n$
$$G_{n_b}\left(\frac{i}{n}\right) = \frac{X_{i,b}}{\sqrt{\frac{1}{n}\sum_{l = 1}^nX_{l,b}^2}},$$
which can be represented by the piece-wise constant function
$$G_{n_b}(u) = \begin{cases}
    G_{n_b}\left(\frac{i}{n}\right), & \frac{i-1}{n} < u \le \frac{i}{n}, \quad i=1,\ldots,n,\\
    0 \quad &\text{otherwise}
\end{cases}.$$
\Cref{ass:X} implies that as $n \to \infty$
\[
\|G_{n_i} - G_i\|_{L_2[0,1]} \to 0,
\]
where $G_i(u) = g_i(u)/\sqrt{\int_0^1g_i^2(u)\,\mathrm{d}u}$.

It is straightforward to see that term $I$ is dominated by term $III$. We start by showing $II = O_p(1)$. Note that $II$ is a $p$-dimensional vector with fixed $p$, so it suffices to show that for each $b \in \{1, \dots, p\}$, the $b^{\text{th}}$ entry $[\G_n (\tilde{\bm{\varepsilon}} - \bm{\varepsilon})]_b = O_p(1)$.
Since $\tilde{\varepsilon}_i=Z_i\varepsilon_i$, we have
\begin{align*}
    [\G_n \{\tilde{\bm{\varepsilon}}-\mu(\bm{\vartheta})\bm{\varepsilon}\}]_b = \frac{1}{\sqrt{n}}\sum_{i = 1}^n G_{n_b}\left(\frac{i}{n}\right) (Z_i - \mu(\bm{\vartheta}))\varepsilon_i,
\end{align*}
and the lag-$k$ autocovariance of the process $\{(Z_i - \mu(\bm{\vartheta}))\varepsilon_i\}$ is
\begin{align*}
    &\left|\cov\left[(Z_i - \mu(\bm{\vartheta}))\varepsilon_i, (Z_{i+k} - \mu(\bm{\vartheta}))\varepsilon_{i+k}\right]\right|\\
    &\quad = \left|\mathbb E\left[ (Z_i - \mu(\bm{\vartheta}))(Z_{i+k} - \mu(\bm{\vartheta}))\right] \mathbb E\left(\varepsilon_i \varepsilon_{i+k} \right)\right|\\
    &\quad \leq \sigma_\varepsilon^2 \left|\cov\left[Z_i - \mu(\bm{\vartheta}), Z_{i+k} - \mu(\bm{\vartheta})\right]\right|,
\end{align*}
where $\sigma_{\varepsilon}^2 = \var\left(\varepsilon_i\right) < \infty$ is the marginal variance of $\bm{\varepsilon}$. Thus, we have for each $b~\in~\{1, \dots, p\}$,
\begin{align*}
    &\var\left([\G_n \{\tilde{\bm{\varepsilon}}-\mu(\bm{\vartheta})\bm{\varepsilon}\}]_b\right)\\
    &\quad = \frac{1}{n}\sum_{i=1}^n\sum_{\ell=1}^n
    G_{n_b}\left(\frac{i}{n}\right)G_{n_b}\left(\frac{\ell}{n}\right)
    \cov\left[(Z_i-\mu(\bm{\vartheta}))\varepsilon_i,(Z_\ell-\mu(\bm{\vartheta}))\varepsilon_\ell\right]\\
    &\quad \leq \frac{C}{n}\sum_{i=1}^n\sum_{\ell=1}^n
    \left|G_{n_b}\left(\frac{i}{n}\right)G_{n_b}\left(\frac{\ell}{n}\right)\right|\rho_Z^{|i-\ell|}\\
    &\quad \leq C\, \sup_n\left\{\frac{1}{n}\sum_{i=1}^nG_{n_b}\left(\frac{i}{n}\right)^2\right\}
    \sum_{k=-\infty}^{\infty}\rho_Z^{|k|}=O(1),
\end{align*}
for some $\rho_Z\in(0,1)$ and finite constant $C>0$, where the last line follows from Cauchy's inequality, the geometric covariance decay of $\{Z_i\}$ and \Cref{ass:X}. Hence $II=O_p(1)$ and therefore $n^{-d}II=o_p(1)$.

We then consider term $III$, which dominates the other two terms and thus provides the limiting distribution. For this purpose, we consider the long-memory regime of $\{\varepsilon_i\}$, i.e., ~$0 < d < 1/2$, for which we verify the condition of \Cref{thm:clt_long-memory}. Hereinafter we define $G_{n_i}(u) = f_n(u)$ and $G_i(u) = f(u)$, for $j = 1, \dots, p$. Since $G_{n_j}(u) < \infty$ and the support of $G_{n_j}(u)$ is $(0, 1]$ by definition, Condition (i) of  \Cref{thm:clt_long-memory} is satisfied. For Condition (ii) of  \Cref{thm:clt_long-memory}, note that $B_n^d(u)$ is a partial sum process of a stationary process with long-range dependence. The convergence to a standard fractional Brownian motion is satisfied by many examples, such as Gaussian processes \cite[see Theorem~4.2 in][]{beran2013long} and the linear processes \cite[see Theorem~4.6 in][]{beran2013long}. Hence, after proper rescaling, by \Cref{thm:clt_long-memory} we have
\[
n^{-d}\C_n^{-1}\G_n  \bm{\varepsilon} \xrightarrow{\mathcal{D}}\C^{-1}\int_{0}^1 \G(u) \, \mathrm{d}B^{d}(u),
\]
where $\{B^d(u)\}_{u \geq 0}$ is a standard fractional Brownian motion.  Therefore we have that
\begin{align*}
    n^{-d}\D_n^{1/2}(\hat{\bm{\beta}} - \bm{\beta}) \xrightarrow{\mathcal{D}} \C^{-1}\int_{0}^1 \G(u) \, \mathrm{d}B^{d}(u),
\end{align*}
which concludes the proof.
\end{proof}

\subsubsection{Consistency of $\hat{\boldsymbol{\nu}}$}
\label{app:consistency_nu_hat}

This appendix discusses the consistency of the WV estimator $\hat{\boldsymbol{\nu}} = (\hat{{\nu}}_1, \dots, \hat{{\nu}}_J)^{\trans}$ for some finite $J \in \mathbb{N}^+$, which is computed on the \textit{observable} residuals $\hat{\boldsymbol{\varepsilon}}= \tilde\Y-\tilde{\X}\hat{\bm{\beta}}$. More specifically, we denote by $\tilde{\bm{\nu}} = (\tilde{{\nu}}_1, \dots, \tilde{{\nu}}_J)^{\trans}$ the empirical WV computed on the \textit{unobservable} vector of residuals $\tilde{\boldsymbol{\varepsilon}}=\tilde{\Y}-\tilde{\X}{\bm{\beta}} = {\boldsymbol{\varepsilon}} \, \odot  \,\Z$, and we denote by ${\bm{\nu}} = ({{\nu}}_1, \dots, {{\nu}}_J)^{\trans}$ the theoretical WV of $\tilde{\boldsymbol{\varepsilon}}$, whose $j^{\text{th}}$ entry is given in \eqref{eq:wv_missing_in_appendix}. We firstly show that  $\tilde{\nu}_j \xrightarrow{p} \nu_j$ and we then show that $\hat{\nu}_j \xrightarrow{p} \tilde{\nu}_j$. By combining these two results we finally prove the desired elementwise result that $\hat{\nu}_j \xrightarrow{p} \nu_j$.

For $j = 1, \hdots, J$ and with $M_j=n-2^j+1$, we denote the wavelet coefficients at scale $j$ computed on the \textit{observable} residuals $\hat{\bm{\varepsilon}}$ and the \textit{unobservable} $\tilde{\bm{\varepsilon}}$ respectively as
\[
     \left\{\hat{W}_{k,j}\right\}_{k \in [1,M_j]} \quad \text{and} \quad \left\{\tilde{W}_{k,j}\right\}_{k \in [1,M_j]} 
\]
where $\hat{W}_{k,j} = \mathbf{h}_{k,j}^{\top}\hat{\bm{\varepsilon}}$ and $\tilde{W}_{k,j} = \mathbf{h}_{k,j}^{\top}\tilde{\bm{\varepsilon}}$ with $\mathbf{h}_{k,j} \in \real^n$ being the $k^{\text{th}}$ Haar wavelet filter at the level $j$.
The empirical wavelet variances are given by
$$
\hat{\nu}_j=\frac{1}{M_j} \sum_{k=1}^{M_j} \hat W_{k, j}^2 \quad \text{and} \quad \tilde{\nu}_j=\frac{1}{M_j} \sum_{k=1}^{M_j} \tilde W_{k, j}^2.
$$
In addition, let us define the first order differences of the process $\boldsymbol{\varepsilon}$ as
\[
    \Delta_i = \varepsilon_i - \varepsilon_{i-1}, \text{ with } i=2,\ldots,n,
\]
whereas the first order differences of the process $\tilde{\boldsymbol{\varepsilon}}$ is defined as
\[
    \tilde{\Delta}_i  = Z_i\varepsilon_i - Z_{i-1}\varepsilon_{i-1}, \text{ with } i=2,\ldots,n.
\]
The reason
for considering the first order differences lies in the fact that the wavelet coefficients can be represented as particular linear combinations of the process, i.e., 
\begin{equation}\label{eq:relationship_diff}
     \hat{W}_{k,j} = \mathbf{c}_j^{\trans}\bm{\Delta}_{k,j} \quad \text{and} \quad \tilde{W}_{k,j} = \mathbf{c}_j^{\trans}\tilde{\bm{\Delta}}_{k,j},
\end{equation}
where $\bm{\Delta}_{k,j} = (\Delta_{k}, \dots, \Delta_{k + 2^j-1})^{\trans}$, $\tilde{\bm{\Delta}}_{k,j} = (\tilde\Delta_{k}, \dots, \tilde\Delta_{k + 2^j-1})^{\trans}$ and $\mathbf{c}_j \in \real^{2^j}$ represents a specific filter vector.

We now state our assumption based on these definitions.
\begin{Assumption}
\label{ass:error_process_WV_consistency}
Let $\{\varepsilon_{i}\}_{i \in \mathbb{Z}}$ be a strictly stationary and $\alpha$-mixing error process satisfying one of the following conditions:
\begin{itemize}
        \item[(i)]
        $\mathbb{E}|\varepsilon_i|^{2\delta} < \infty \quad \text{and} \quad \sum_{k \geq 1}\alpha_k^{1 - 2/\delta} < \infty \quad \text{for some constant} \quad \delta > 2$.
        \item[(ii)]
        $\mathbb{P}(|\varepsilon_i| < C) = 1 \quad \text{for some constant} \quad C > 0, \quad \text{and} \quad \sum_{k \geq 1}\alpha_k < \infty$.
    \end{itemize}
\end{Assumption}
This is a stationary short-memory condition for the error process itself and should be read as an additional condition for the WV consistency result. It is not implied by the moment condition in \Cref{ass:epsilon}(i), which only assumes a lower-order moment and does not by itself give the higher moment required here. The implication is immediate under the bounded alternative in \Cref{ass:epsilon}(ii); under the moment alternative, the stronger moment and corresponding mixing summability are imposed directly in this assumption.
Using the above assumption, we can now present the following auxiliary Lemma~\ref{lemma:consistency_empirical_wv_on_noise_process_times_missingness_process}.

\begin{Lemma}
\label{lemma:consistency_empirical_wv_on_noise_process_times_missingness_process}
    Consider the empirical wavelet variance computed on the true vector of residuals $\tilde{\boldsymbol{\varepsilon}} = \tilde{\Y} - \tilde{\X}{\bm{\beta}} = \boldsymbol{\varepsilon} \odot \Z$ and denoted as $\tilde{\boldsymbol{\nu}}$, as well as the theoretical wavelet variance of the process $\tilde{\boldsymbol{\varepsilon}}$ denoted as $\boldsymbol{\nu}$. Under Assumptions~\ref{ass:missing} and \ref{ass:error_process_WV_consistency}, we have for any $j \in \{1, \dots, J\}$,
\[
   \left| \tilde{\nu}_j - \nu_j \right| \xrightarrow{p} 0.
\]
\label{lemma:nu_tilde_consistent}
\end{Lemma}
\begin{proof}[Proof of \Cref{lemma:consistency_empirical_wv_on_noise_process_times_missingness_process}]
    Note that $\tilde{\nu}_j$ is an unbiased estimator of $\nu_j$.
We have that
\[
    \tilde{\nu}_j-\nu_j = \tilde{\nu}_j-\mathbb{E}\left[\tilde{\nu}_j\right]=\frac{1}{M_j} \sum_{k=1}^{M_j} \tilde{W}_{{k, j}}^2-\mathbb{E}\left[\tilde{W}_{k, j}^2\right],
\]
where $M_j \to \infty$ as $n \to \infty$. Under \Cref{ass:missing}, \Cref{lemma:ergodic} shows that $\{Z_i\}_i$ has $\alpha$-mixing coefficients that decay exponentially. Since $\{\varepsilon_i\}_i$ and $\{Z_i\}_i$ are independent, \Cref{ass:error_process_WV_consistency} and \Cref{lemma:mixing_indep_comp} imply that the masked process $\tilde{\varepsilon}_i=Z_i\varepsilon_i$ is strictly stationary and $\alpha$-mixing with summable coefficients of the required order. For fixed scale $j$, the sequence $\{\tilde W_{k,j}\}_{k\in\mathbb Z}$ is a finite-block measurable transformation of $\{\tilde{\varepsilon}_i\}$, and therefore $\{\tilde W_{k,j}^2-\mathbb E(\tilde W_{k,j}^2)\}_{k\in\mathbb Z}$ is also strictly stationary and $\alpha$-mixing, with the same summability after a finite lag shift. The moment condition in \Cref{ass:error_process_WV_consistency} gives $\mathbb E|\tilde W_{k,j}^2-\mathbb E(\tilde W_{k,j}^2)|^\delta<\infty$ in case (i), while case (ii) is bounded. Applying \Cref{lemma:mixing_mement_ineq} to the centered sequence $\tilde W_{k,j}^2-\mathbb E(\tilde W_{k,j}^2)$ gives
\[
    \var\left(\frac{1}{\sqrt{M_j}}\sum_{k=1}^{M_j}\left\{\tilde W_{k,j}^2-\mathbb E(\tilde W_{k,j}^2)\right\}\right)=O(1).
\]
Consequently, $\var(\tilde{\nu}_j-\nu_j)=O(M_j^{-1})$, and Chebyshev's inequality yields the final result.
\end{proof}

We now consider the wavelet variance  $\hat{\bm{\nu}}$ and show that $|\hat{\nu}_j - \tilde{\nu}_j | \xrightarrow{p} 0$. To obtain this result, we need to impose an assumption on the rows of the design matrix.

\begin{Assumption}[Bounded row vector of the design matrix]
\label{ass:bounded_design_matrix}
Denoting the $i^{\text{th}}$ row of the design matrix $\X$ as $\X_{i \cdot}$, we assume that
\begin{equation*}
    \max _{i\in \{1,\ldots, n\}} \left\|\X_{i \cdot} \right\|_2 < \infty.
\end{equation*}
    
\end{Assumption}

Using this additional Assumption~\ref{ass:bounded_design_matrix} on the rows of the design matrix, we can now state an intermediary result in Lemma~\ref{lemma:difference_empirical_wv_on_true_residuals_minus_empirical_wv_on_estimated_residuals_goes_to_0}
which establishes the convergence of the empirical wavelet variance computed on the \textit{estimated} vector of residuals $\hat{\boldsymbol{\varepsilon}}=\Y-\tilde{\X}\hat{\bm{\beta}}$ denoted as $\hat{\boldsymbol{\nu}}$  to the empirical wavelet variance computed on the true vector of residuals $\tilde{\boldsymbol{\varepsilon}} = \tilde{\Y} - \tilde{\X}{\bm{\beta}}$ and denoted as $\tilde{\boldsymbol{\nu}}$.

\begin{Lemma}
Denoting the empirical wavelet variance computed on the \textit{estimated} vector of residuals $\hat{\boldsymbol{\varepsilon}}= \tilde\Y-\tilde{\X}\hat{\bm{\beta}}$ as $\hat{\boldsymbol{\nu}}$  and the empirical wavelet variance computed on the true vector of residuals $\tilde{\boldsymbol{\varepsilon}} = \tilde{\Y} - \tilde{\X}{\bm{\beta}}$ as $\tilde{\boldsymbol{\nu}}$, under Assumption~\ref{ass:bounded_design_matrix}, we have
    
\begin{equation*}
    \left|\hat{\nu}_j - \tilde{\nu}_j \right|    \xrightarrow{p} 0.
\end{equation*}
\label{lemma:difference_empirical_wv_on_true_residuals_minus_empirical_wv_on_estimated_residuals_goes_to_0}

\end{Lemma}

\begin{proof}[Proof of \Cref{lemma:difference_empirical_wv_on_true_residuals_minus_empirical_wv_on_estimated_residuals_goes_to_0}]
It follows that
\begin{equation*}
    \begin{aligned}
        \left|\hat{\nu}_j - \tilde{\nu}_j \right| &= \left|\frac{1}{M_j} \sum_{k=1}^{M_j} \left( \hat{W}_{k,j}^2 - \tilde{W}_{k,j}^2 \right)\right| \\
        &= \left|\frac{1}{M_j} \sum_{k=1}^{M_j} \left( \hat{W}_{k,j} - \tilde{W}_{k,j}\right)^2 +  \frac{2}{M_j} \sum_{k=1}^{M_j} \tilde{W}_{k,j} \left( \hat{W}_{k,j} - \tilde{W}_{k,j} \right)  \right| \\
        &\leq \left|\frac{1}{M_j} \sum_{k=1}^{M_j} \left( \hat{W}_{k,j} - \tilde{W}_{k,j}\right)^2 \right| + 2 \left| \sum_{k=1}^{M_j}\sqrt{\frac{1}{M_j}} \tilde{W}_{k,j} \sqrt{\frac{1}{M_j}}\left( \hat{W}_{k,j} - \tilde{W}_{k,j} \right)  \right| \\
        &\leq \left|\frac{1}{M_j} \sum_{k=1}^{M_j} \left( \hat{W}_{k,j} - \tilde{W}_{k,j}\right)^2 \right| + 2 \left|\sqrt{\sum_{k=1}^{M_j} \frac{1}{M_j} \tilde{W}_{k,j} ^2} \sqrt{ \sum_{k=1}^{M_j} \frac{1}{M_j} \left(\hat{W}_{k,j}-\tilde{W}_{k,j}\right)^2 } \right| \\
        &=\left|\frac{1}{M_j} \sum_{k=1}^{M_j}\left(\hat{W}_{k_j}-\tilde{W}_{k_j j}\right)^2\right|+2 \left| \sqrt{\frac{1}{M_j} \sum_{k=1}^{M_j} \tilde{W}^2_{k, j}} \sqrt{\frac{1}{M_j} \sum_{k=1}^{M_j}\left(\hat{W}_{k,j}-\tilde{W}_{k,j}\right)^2}\right|.
    \end{aligned}
\end{equation*}

Noting that the difference between the wavelet coefficients can be written as
\begin{equation*}
    \begin{aligned}
        \hat{W}_{k,j} - \tilde{W}_{k,j}&= \mathbf{h}_{k,j}^{\trans}\left(\hat{\boldsymbol{\varepsilon}} - \tilde{\boldsymbol{\varepsilon}} \right) \\
        &= \mathbf{h}_{k,j}^{\trans}\left(\tilde{\Y}- \tilde{\X}\hat{\boldsymbol{\beta}} - \left(\tilde{\Y}- \tilde{\X} \boldsymbol{\beta} \right)\right) \\
        &= \mathbf{h}_{k,j}^{\trans} \left(\tilde{\X} \left( \boldsymbol{\beta} - \hat{\boldsymbol{\beta}}\right) \right),
    \end{aligned}
\end{equation*}
we obtain
\begin{equation*}
    \left| \frac{1}{M_j} \sum_{k=1}^{M_j} \left(\mathbf{h}_{k,j}^{\trans} \tilde{\X} \left( \boldsymbol{\beta} - \hat{\boldsymbol{\beta}}\right) \right)^2\right|+2 \left| \sqrt{\frac{1}{M_j} \sum_{k=1}^{M_j} W^2_{k, j}} \sqrt{\frac{1}{M_j} \sum_{k=1}^{M_j}\left(
    \mathbf{h}_{k,j}^{\trans} \tilde{\X} \left( \boldsymbol{\beta} - \hat{\boldsymbol{\beta}}\right) 
    \right)^2}\right|.
\end{equation*}

For $ \left| \sum_{k=1}^{M_j} \left(  \mathbf{h}_{k,j}^{\trans} \tilde{\X} \left( \boldsymbol{\beta} - \hat{\boldsymbol{\beta}}\right)  \right)^2 \right|$, we have
\begin{equation*}
    \begin{aligned}
         \left| \sum_{k=1}^{M_j} \left(  \mathbf{h}_{k,j}^{\trans} \tilde{\X} \left( \boldsymbol{\beta} - \hat{\boldsymbol{\beta}}\right)  \right)^2 \right| 
       \leq M_j \max _{k=\left[1, M_j\right]}\left| \mathbf{h}_{k,j}^{\trans} \tilde{\X}\left( \boldsymbol{\beta} - \hat{\boldsymbol{\beta} } \right)  \right|^2.
    \end{aligned}
\end{equation*}
Taking the square root of the term yields
\begin{equation*}
    \begin{aligned}
  \sqrt{M_j} \max _{k=\left[1, M_j\right]}\left| \mathbf{h}_{k,j}^{\trans} \tilde{\X}\left( \boldsymbol{\beta} - \hat{\boldsymbol{\beta} } \right)  \right| &=
    \sqrt{M_j} \max _{k=\left[1, M_j\right]}\left| 
    \sum_{i=1}^n h_{k,j,i} \tilde{\X}_{i \cdot}
    \left( \boldsymbol{\beta} - \hat{\boldsymbol{\beta} } \right)  \right| \\
    &= \sqrt{M_j} \max _{k=\left[1, M_j\right]}\left| 
    \sum_{r=0}^{2^j-1} h_{k+r,j} \tilde{\X}_{k+r \cdot}
    \left( \boldsymbol{\beta} - \hat{\boldsymbol{\beta} } \right)  \right| \\
    & \leq \sqrt{M_j} \max _{k=\left[1, M_j\right]} \left\{ \sum_{r=0}^{2^j-1} \left| \frac{1}{2^j} \right| \left|\tilde{\X}_{k+r \cdot} \left( \boldsymbol{\beta} - \hat{\boldsymbol{\beta} } \right)  \right|\right\} \\
    &= \sqrt{M_j}\max _{k=\left[1, M_j\right]} \left\{  \max _{r=\left[0, 2^j-1\right]}  \left| \tilde{\X}_{k+r \cdot} \left( \boldsymbol{\beta} - \hat{\boldsymbol{\beta} } \right) \right| \right\} \\
    &=  \sqrt{M_j}\max _{i=\left[1, n\right]} \left| \tilde{\X}_{i \cdot} \left( \boldsymbol{\beta} - \hat{\boldsymbol{\beta} } \right) \right| \\
    & \leq \sqrt{M_j} \max _{i=\left[1, n\right]}  \left\|\tilde{\X}_{i \cdot} \right\|_2 \left\| \left( \boldsymbol{\beta} - \hat{\boldsymbol{\beta} } \right) \right\|_2 \\
    &= \left\| \sqrt{M_j} \left( \boldsymbol{\beta} - \hat{\boldsymbol{\beta} } \right)  \right\|_2 \max _{i=\left[1, n\right]} \left\|\tilde{\X}_{i \cdot} \right\|_2. 
    \end{aligned}
\end{equation*}

Using Assumption~\ref{ass:bounded_design_matrix}, we obtain
\begin{equation*}
    \begin{aligned}
         \sqrt{M_j} \max _{k=\left[1, M_j\right]}\left| \mathbf{h}_{k,j}^{\trans} \tilde{\X}\left( \boldsymbol{\beta} - \hat{\boldsymbol{\beta} } \right)  \right| &= \left\| \sqrt{M_j} \left( \boldsymbol{\beta} - \hat{\boldsymbol{\beta} } \right)  \right\|_2 \max _{i=\left[1, n\right]} \left\|\tilde{\X}_{i.} \right\|_2
    = \mathcal{O}_p(1).
    \end{aligned}
\end{equation*}

Hence we have
\begin{equation*}
    \begin{aligned}
     \left|\hat{\nu}_j - \tilde{\nu}_j \right| &= \left|\frac{1}{M_j} \sum_{k=1}^{M_j}\left(\hat{W}_{k_j}-\tilde W_{k_j j}\right)^2\right|+2 \left| \sqrt{\frac{1}{M_j} \sum_{k=1}^{M_j} \tilde W^2_{k, j}} \sqrt{\frac{1}{M_j} \sum_{k=1}^{M_j}\left(\hat{W}_{k,j}- \tilde W_{k,j}\right)^2}\right| \\
           &= \left|\frac{1}{M_j} \mathcal{O}_p(1)\right| +  2 \left| \mathcal{O}_p(1) \sqrt{\frac{1}{M_j}} \mathcal{O}_p(1) \right| 
           \stackrel{p}{\longrightarrow} 0
    \end{aligned}
\end{equation*}
\end{proof}

Combining these steps, we obtain the consistency of the estimator of the WV computed on the estimated residuals and can state Theorem~\ref{theorem:consistency_nu_hat}.

\begin{Theorem}[Consistency of the wavelet variance on the estimated residuals for stationary short-memory processes]
\label{theorem:consistency_nu_hat}

  Consider the empirical wavelet variance computed on the \textit{estimated} vector of residuals $\hat{\boldsymbol{\varepsilon}}=\Y-\tilde{\X}\hat{\bm{\beta}}$ denoted as $\hat{\boldsymbol{\nu}}$ and the theoretical wavelet variance of the process $\tilde{\boldsymbol{\varepsilon}} = \tilde{\Y} - \tilde{\X}{\bm{\beta}} = \boldsymbol{\varepsilon} \odot \Z$ denoted as $\boldsymbol{\nu}$. Under Assumption~\ref{ass:bounded_design_matrix}, we have

\begin{equation*}
    \hat{\bm{\nu}}_j \overset{p}{\rightarrow} \boldsymbol{\nu}_j
\end{equation*}
\end{Theorem}

\begin{proof}
The result can be directly obtained from combining results of Lemma~\ref{lemma:nu_tilde_consistent} and Lemma~\ref{lemma:difference_empirical_wv_on_true_residuals_minus_empirical_wv_on_estimated_residuals_goes_to_0}. Indeed, recalling that $\tilde{\boldsymbol{\nu}}$ is an unbiased estimator of ${\boldsymbol{\nu}}$, we obtain:

\begin{equation*}
  \begin{aligned}
        \left| \hat{\boldsymbol{\nu}}_j  - \boldsymbol{\nu}_j \right| &=       \left| \hat{\boldsymbol{\nu}}_j  - \mathbb{E}[\tilde{\boldsymbol{\nu}}_j] \right| \\
        &=\left| \hat{\boldsymbol{\nu}}_j - \tilde{\boldsymbol{\nu}}_j + \tilde{\boldsymbol{\nu}}_j - \mathbb{E}[\tilde{\boldsymbol{\nu}}_j] \right| \\
        & \leq \left| \hat{\boldsymbol{\nu}}_j - \tilde{\boldsymbol{\nu}}_j  \right| + \left| \tilde{\boldsymbol{\nu}}_j - \boldsymbol{\nu}_j \right|  \overset{p}{\rightarrow} 0
  \end{aligned}
\end{equation*}

\end{proof}

\subsubsection{Consistency of $\hat{\boldsymbol{\gamma}}$}
\label{app:consistency_of_gamma_hat}

Recall the stochastic parameters $\boldsymbol{\gamma} \in \Gamma \subset \real^q$ and its estimator $\hat{\boldsymbol{\gamma}}$ defined in \eqref{eq:gmwm_new}.

Given the consistency of $\hat{\bm{\nu}}$ obtained in Appendix~\ref{app:consistency_nu_hat}, the consistency of the estimated stochastic parameters $\hat{\boldsymbol{\gamma}}$ is relatively straightforward based on existing results \citep[see e.g.,][]{newey1994large, guerrier2022robust}. For this purpose we list the additional required assumptions below which consist in standard regularity conditions for these settings.

\begin{Assumption}
\label{ass:gamma_parameter_space_compact}
    The parameter space $\boldsymbol{\Gamma}$ of the stochastic error parameter $\bm{\gamma}$ is compact.
\end{Assumption}
This compactness assumption on the parameter space $\boldsymbol{\Gamma}$ is standard in asymptotic theory and allows us to show uniform convergence of the later defined function $Q_n(\boldsymbol{\gamma}^\star)$ to $Q_0(\boldsymbol{\gamma}^\star)$ in order to then obtain consistency of $\hat{\boldsymbol{\gamma}}$.
\begin{Assumption}
\label{ass::nu_gamma_continuous}
The theoretical WV $\boldsymbol{\nu}(\boldsymbol{\gamma}) $ is continuous $\forall \, \boldsymbol{\gamma} \in \boldsymbol{\Gamma}$.
\end{Assumption}
This assumption on the continuity of the theoretical WV function is a standard regularity condition commonly made in method-of-moments estimation frameworks and allow us to make use of Theorem 2.1 in \cite{newey1994large} to obtain consistency of $\hat{\boldsymbol{\gamma}}$.

\begin{Assumption}
\label{ass::nu_gamma_injective}
The theoretical WV is injective with respect to $\bm{\gamma}$, i.e., $\boldsymbol{\nu}(\boldsymbol
\gamma_1) = \boldsymbol{\nu}(\boldsymbol{\gamma}_2)$ if and only if $\boldsymbol{\gamma}_1 = \boldsymbol{\gamma}_2$.
\end{Assumption}

This assumption on the injectivity of the theoretical WV with respect to $\boldsymbol{\gamma}$ is a crucial  condition. It ensures that the mapping from the parameter space to the moment space is one-to-one, so that the true parameter value $\boldsymbol{\gamma}$ can be uniquely recovered from the moment condition. 

\begin{Assumption}
\label{ass::consistent_omega_matrix}
    When using an estimated matrix $\hat{\bm{\Omega}}$, then it is such that $\left\|\hat{\boldsymbol{\Omega}}-\boldsymbol{\Omega}\right\|_{\op}~\xrightarrow{p}~0
    $.
\end{Assumption}
This assumption ensures that the estimated weighting matrix $\hat{\boldsymbol{\Omega}}$ converges in probability to the true matrix ${\boldsymbol{\Omega}}$, in operator norm. It is a standard condition in GMM-type estimation procedures where the optimal weighting matrix is estimated. 

As stated earlier, Assumptions~\ref{ass:gamma_parameter_space_compact}, \ref{ass::nu_gamma_continuous}, \ref{ass::nu_gamma_injective} and \ref{ass::consistent_omega_matrix} are standard regularity conditions that ensure that derivatives and expansions can be made on the GMWM objective function so that its convergence to the true objective function implies the convergence of $\hat{\bm{\gamma}}$ to $\bm{\gamma}$ by the continuous mapping theorem. Theorem~\ref{thm:consistency_stochastic_parameters} establishes the consistency of the estimated stochastic parameters $\hat{\boldsymbol{\gamma}}$ and briefly discusses these points in the following proof.

\begin{Theorem}
\label{thm:consistency_stochastic_parameters}
Under Assumptions \ref{ass:gamma_parameter_space_compact}, \ref{ass::nu_gamma_continuous}, \ref{ass::nu_gamma_injective} and \ref{ass::consistent_omega_matrix}, we have:
    \begin{equation}
        \hat{\boldsymbol{\gamma}} \stackrel{p}{\longrightarrow} \boldsymbol{\gamma}
    \end{equation}

\end{Theorem}

\begin{proof}
Let $\boldsymbol{\gamma}^\star \in \boldsymbol{\Gamma}$, we define
$$
Q_n(\boldsymbol{\gamma}^\star) = \left\{\hat{\boldsymbol{\nu}}-\boldsymbol{\nu}(\boldsymbol{\gamma}^\star)\right\}^{\trans} \hat{\boldsymbol{\Omega}}\left\{\hat{\boldsymbol{\nu}}-\boldsymbol{\nu}(\boldsymbol{\gamma}^\star)\right\}
$$
and
$$
Q_0(\boldsymbol{\gamma}^\star) = \left\{\boldsymbol{\nu}\left(\boldsymbol{\gamma}\right)-\boldsymbol{\nu}(\boldsymbol{\gamma}^\star)\right\}^{\trans}\boldsymbol{\Omega}\left\{\boldsymbol{\nu}\left(\boldsymbol{\gamma}\right)-\boldsymbol{\nu}(\boldsymbol{\gamma}^\star)\right\}
$$

Based on Assumption~\ref{ass:gamma_parameter_space_compact} implying compactness of the parameter space $\boldsymbol{\Gamma}$ and Assumption \ref{ass::consistent_omega_matrix} implying consistency of $\hat{\boldsymbol{\Omega}}$, as well as Theorem~\ref{theorem:consistency_nu_hat} that state the consistency of the estimated WV computed on the estimated residuals and using the continuous mapping theorem, we can apply Lemma~\ref{lemma::uniform_convergence} and obtain uniform convergence of $Q_n(\boldsymbol{\gamma}^{\star})$ to $Q(\boldsymbol{\gamma}^{\star}$):
\begin{equation}
\label{uniform_confergence_loss_fct_gmm}
    \sup _{\boldsymbol{\gamma} \in \boldsymbol{\Gamma}}\left|Q_n(\boldsymbol{\gamma}^{\star})-Q_0(\boldsymbol{\gamma}^{\star})\right| \overset{p}{\longrightarrow} 0 .
\end{equation}

Based on Assumption \ref{ass::nu_gamma_injective} implying injectivity of the theoretical WV and the non-singularity of $\boldsymbol{\Omega}$, $Q_0(\boldsymbol{\gamma}^\star)$ has a unique minimum at $\boldsymbol{\gamma}$. Therefore, by Assumption \ref{ass:gamma_parameter_space_compact} implying compactness of the parameter space $\boldsymbol{\Gamma}$, Assumption \ref{ass::nu_gamma_continuous} implying that the theoretical WV is continuous, Assumption \ref{ass::nu_gamma_injective} implying that the theoretical WV is injective and Assumption \ref{ass::consistent_omega_matrix} implying consistency of $\hat{\boldsymbol{\Omega}}$, and using  \eqref{uniform_confergence_loss_fct_gmm} implying uniform convergence of $Q_n(\boldsymbol{\gamma}^{\star})$ to $Q_0(\boldsymbol{\gamma}^{\star})$, Theorem 2.1 in \cite{newey1994large} can be applied implying the consistency of $\hat{\boldsymbol{\gamma}}$.

\end{proof}

\subsubsection{Model Selection Properties}
\label{app:model_selection}

For ease of notation, we present the following derivations assuming that the process is parameterized only by $\bgamma$. Since the missing-data mechanism, parameterized by $\boldsymbol{\vartheta}$, is assumed to be common to all candidate models, its associated complexity penalty is constant and therefore omitted. The general formulation is recovered by replacing $\bnu(\bgamma)$ with $\bnu(\bgamma,\boldsymbol{\vartheta})$.

Recall the proposed model selection criterion (introduced in Section \ref{sec:model_selection}) defined as:
\begin{equation}
\label{def:model_selection_criterion1}
    C = \mathbb{E}_0\left[ \mathbb{E}\left\{\left\|\hat{\bnu}_0 - \bnu(\hat{\bgamma}) \right\|_{\bOmega}^2 \right\}\right] = \mathbb{E}_0\left( \mathbb{E}\left[\left\{\hat{\bnu}_0 - \bnu(\hat{\bgamma})\right\}\trans \bOmega \left\{\hat{\bnu}_0 - \bnu(\hat{\bgamma})\right\}\right]\right),
\end{equation}
where $\hat{\bgamma}$ is computed on the estimated residuals $\hat{\bm{\varepsilon}}$, $\hat{\bnu}_0$ is computed on $\hat{\bm{\varepsilon}}_0$ (an independent copy of $\hat{\bm{\varepsilon}}$ satisfying $\hat{\bm{\varepsilon}} \overset{d}{=} \hat{\bm{\varepsilon}}_0$) and $\boldsymbol{\Omega} = \mathbf{V}^{-1}$. To deliver a \textit{feasible} estimator for this criterion, we first state the following lemma.

\begin{Lemma}
\label{lem:alternative_expression_C}
The model selection criterion defined in \eqref{def:model_selection_criterion1} can also be expressed as 
\begin{equation*}
    C = \mathbb{E}\left\{\left\|\hat{\bnu} - \bnu(\hat{\bgamma}) \right\|_{\bOmega}^2 \right\} + 2 \tr\left[ \cov\left\{\hat{\bnu}, \bnu(\hat{\bgamma})\right\} \bOmega\right].
\end{equation*}
\end{Lemma}

\begin{proof}
Recall that we aim to show
\begin{equation*}
    C = \mathbb{E}_0\left[ \mathbb{E}\left\{\left\|\hat{\bnu}_0 - \bnu(\hat{\bgamma}) \right\|_{\bOmega}^2 \right\}\right] =  \mathbb{E}\left\{\left\|\hat{\bnu} - \bnu(\hat{\bgamma}) \right\|_{\bOmega}^2 \right\} + 2 \tr\left[ \cov\left\{\hat{\bnu}, \bnu(\hat{\bgamma})\right\} \bOmega\right].
\end{equation*}
Notice that, for any $\x, \y \in \real^k$ and symmetric matrix $\mathbf{A} \in \real^{k \times k}$, we have
\begin{equation*}
    \left\|\x-\y\right\|_{\mathbf{A}}^2 = (\x-\y)\trans \mathbf{A} (\x-\y) = \x\trans \mathbf{A} \x - \y\trans \mathbf{A} \y - 2(\x-\y)\trans\mathbf{A}\y. 
\end{equation*}

Thus, we have
\begin{equation*}
    C = \mathbb{E}_0\left(\hat{\bnu}_0\trans \bOmega \hat{\bnu}_0\right) - \mathbb{E}\left\{\bnu(\hat{\bgamma})\trans \bOmega \bnu(\hat{\bgamma})\right\} - 2 \mathbb{E}\left[\left\{\mathbb{E}_0\left(\hat{\bnu}_0\right) - \bnu(\hat{\bgamma})\right\}\trans \bOmega \bnu(\hat{\bgamma})\right].
\end{equation*}
Then, we consider the following term
\begin{equation*}
    D = \mathbb{E}\left\{\left\|\hat{\bnu} - \bnu(\hat{\bgamma}) \right\|_{\bOmega}^2
    \right\} = \mathbb{E}\left(\hat{\bnu}\trans \bOmega \hat{\bnu}\right) - \mathbb{E}\left\{\bnu(\hat{\bgamma})\trans \bOmega \bnu(\hat{\bgamma})\right\} - 2 \mathbb{E}\left[\left\{\hat{\bnu}- \bnu(\hat{\bgamma})\right\}\trans \bOmega \bnu(\hat{\bgamma})\right],
\end{equation*}
and thus
\begin{equation*}
    C - D = 2\mathbb{E}\left[\left\{\hat{\bnu} -\mathbb{E}_0\left(\hat{\bnu}_0\right) \right\}\trans \bOmega \bnu(\hat{\bgamma})\right].
\end{equation*}
Notice that
\begin{equation*}
    \mathbb{E}\left(\x\trans \y \right) = \mathbb{E}\left\{\tr\left(\x\trans \y \right)\right\} = \tr \left\{\mathbb{E}\left(\y\x\trans \right)\right\} = \tr\left\{\cov\left(\x, \y\right)\right\} + \tr\left\{ \mathbb{E}(\x) \mathbb{E}(\y)\trans\right\}.
\end{equation*}
Thus, we have
\begin{equation*}
\begin{aligned}
    C - D &= 2\mathbb{E}\left[\left\{\hat{\bnu} -\mathbb{E}_0\left(\hat{\bnu}_0\right) \right\}\trans \bOmega \bnu(\hat{\bgamma})\right] \\
    &= 2\tr\left[\cov\left\{\hat{\bnu}, \bOmega\bnu(\hat{\bgamma})\right\}\right] + 2\tr \left[\left\{\mathbb{E}(\hat{\bnu})-\mathbb{E}_0(\hat{\bnu}_0)\right\} \mathbb{E}\left\{\bOmega\bnu(\hat{\bgamma})\right\}\trans\right] \\
    &= 2\tr \left[\cov \left\{\hat{\bnu},  \bnu(\hat{\bgamma})\right\}\bOmega\right],
\end{aligned}
\end{equation*}
since $\mathbb{E}(\hat{\bnu})-\mathbb{E}_0(\hat{\bnu}_0)=\0$. Finally, we obtain
\begin{equation*}
    C = D + 2\tr \left[\cov \left\{\hat{\bnu},  \bnu(\hat{\bgamma})\right\}\bOmega\right] = \mathbb{E}\left\{\left\|\hat{\bnu} - \bnu(\hat{\bgamma}) \right\|_{\bOmega}^2 \right\} + 2 \tr\left[ \cov\left\{\hat{\bnu}, \bnu(\hat{\bgamma})\right\} \bOmega\right].
\end{equation*}
\end{proof}
Recalling that $\mathbf{V} = \var(\hat{\boldsymbol{\nu}})$, $\hat{\mathbf{V}}$ is a consistent estimator of $\mathbf{V}$, $\hat{\boldsymbol{\Omega}} = \hat{\mathbf{V}}^{-1}$ and $\boldsymbol{\Omega} = \mathbf{V}^{-1}$, we consider an estimator of $C$ defined as follows:
\begin{equation*}
    \hat{C} = \left\|\hat{\bnu} - \bnu(\hat{\bgamma}) \right\|_{\hat{\bOmega}}^2 + 2 \tr\left[ \hat{\V} \B(\hat{\bgamma}, \hat{\bOmega})\trans \hat{\bOmega}\right].
\end{equation*}
 Then we have $\tr\left[ \hat{\V} \B(\hat{\bgamma}, \hat{\bOmega})\trans \hat{\bOmega}\right] = \dim (\bm{\gamma})$.

Also, $\hat{\bOmega}$ has a deterministic limit
$\hat{\bOmega} \xrightarrow{p} \bOmega$, and we define
\begin{equation*}
    \G(\bgamma)
    =
    \frac{\partial \bnu(\bgamma)}{\partial \bgamma}.
\end{equation*}
For generic $\boldsymbol{\gamma}_1$, $\boldsymbol{\gamma}_2$ and
$\boldsymbol{\Omega}$, define
\begin{equation*}
    \B(\boldsymbol{\gamma}_1,\boldsymbol{\gamma}_2,\boldsymbol{\Omega})
    =
    \G(\boldsymbol{\gamma}_2)
    \left\{
    \G(\boldsymbol{\gamma}_1)\trans
    \boldsymbol{\Omega}
    \G(\boldsymbol{\gamma}_2)
    \right\}^{-1}
    \G(\boldsymbol{\gamma}_1)\trans
    \boldsymbol{\Omega}.
\end{equation*}
When both arguments coincide, we use the shorthand
\[
    \B(\boldsymbol{\gamma},\boldsymbol{\Omega})
   =
    \B(\boldsymbol{\gamma},\boldsymbol{\gamma},\boldsymbol{\Omega}).
\]

\noindent The following theorem states the properties of the proposed estimator $\hat{C}$.

\begin{Theorem}
\label{thm:model_selection_criterion_consistency}
Assume that 
\begin{enumerate}
    \item $\bnu(\bgamma)$ is twice differentiable, and $\partial^2 \nu_j(\bgamma)/(\partial\bgamma\partial\bgamma\trans)$ has bounded singular values.
    \item The matrices $\bOmega, \hat{\bOmega}, \hat{\V}$ have bounded singular values.
    \item The wavelet coefficients $W_{k,j}$ with $k=1,\ldots,M_j$ and $j=1,\ldots,J$ are all bounded. 
\end{enumerate}
Then we have $\hat{C} = C+o_p(1)$ and $\mathbb{E}(\hat{C}) = C + o(1)$.
\end{Theorem}

\begin{proof}
    Since $\bnu(\bgamma)$ is twice differentiable, by Taylor's theorem we have
\begin{equation}
\label{eqn:taylor_expansion_nu_gamma_hat}
    \bnu(\hat{\bgamma}) = \bnu(\bgamma) + \G(\bgamma)(\hat{\bgamma} - \bgamma) + \r(\hat{\bgamma}, \bgamma_1', \ldots, \bgamma_J'),
\end{equation}
where we recall that $\G(\bgamma) =  \partial \bnu(\bgamma)/\partial \bgamma$. Here $\bgamma_j'$ with $j=1,\ldots,J$ lies between $\hat{\bgamma}$ and $\bgamma$, and $\r(\hat{\bgamma}, \bgamma_1', \ldots, \bgamma_J')$ is defined as a vector of length $J$ whose $j^{\text{th}}$ element is given by 
\begin{equation*}
    r_j(\hat{\bgamma}, \bgamma_1', \ldots, \bgamma_J') = \frac{1}{2} (\hat{\bgamma}-\bgamma)\trans \frac{\partial^2 \nu_j (\bgamma_j')}{\partial \bgamma \partial \bgamma\trans} (\hat{\bgamma}-\bgamma).
\end{equation*}

By the definition of $\hat{\bgamma}$ in \eqref{eq:gmwm_new}, we have 
\begin{equation*}
    \0 = \frac{\partial}{\partial \bgamma} \left\{\hat{\bnu} - \bnu(\hat{\bgamma})\right\}\trans \bOmega \left\{\hat{\bnu} - \bnu(\hat{\bgamma})\right\} = -2 \left\{\hat{\bnu} - \bnu(\hat{\bgamma})\right\}\trans \bOmega \G(\hat{\bgamma}),
\end{equation*}
and thus we have
\begin{equation*}
\begin{aligned}
    \0 &= \left\{\hat{\bnu} - \bnu(\hat{\bgamma})\right\}\trans \bOmega \G(\hat{\bgamma}) \\
    &= \left\{\hat{\bnu} - \bnu(\bgamma) - \G(\bgamma) (\hat{\bgamma}-\bgamma) - \r(\hat{\bgamma}, \bgamma_1',\ldots,\bgamma_J')\right\}\trans \bOmega \G(\hat{\bgamma}) \\
    &= \G(\hat{\bgamma})\trans\bOmega \left\{\hat{\bnu}-\bnu(\bgamma)\right\} - \G(\hat{\bgamma})\trans\bOmega \G(\bgamma) \hat{\bgamma} + \G(\hat{\bgamma})\trans\bOmega \G(\bgamma) \bgamma - \G(\hat{\bgamma})\trans\bOmega \r(\hat{\bgamma}, \bgamma_1',\ldots,\bgamma_J').
\end{aligned}
\end{equation*}
By rearranging the terms, we further obtain
\begin{equation*}
\begin{aligned}
    \hat{\bgamma} &= \left\{\G(\hat{\bgamma})\trans\bOmega \G(\bgamma)\right\}^{-1} \G(\hat{\bgamma})\trans\bOmega \left\{\G(\bgamma)\bgamma + \hat{\bnu} - \bnu(\bgamma) - \r(\hat{\bgamma}, \bgamma_1',\ldots,\bgamma_J')\right\} \\
    &= \bgamma + \left\{\G(\hat{\bgamma})\trans\bOmega \G(\bgamma)\right\}^{-1} \G(\hat{\bgamma})\trans\bOmega \left\{\hat{\bnu} - \bnu(\bgamma) - \r(\hat{\bgamma}, \bgamma_1',\ldots,\bgamma_J')\right\}
\end{aligned}
\end{equation*}
By \eqref{eqn:taylor_expansion_nu_gamma_hat}, we can further write
\begin{equation*}
\begin{aligned}
    \bnu(\hat{\bgamma}) &= \bnu(\bgamma) + \G(\bgamma)(\hat{\bgamma} - \bgamma) + \r(\hat{\bgamma}, \bgamma_1', \ldots, \bgamma_J') \\
    &= \bnu(\bgamma) + \r(\hat{\bgamma}, \bgamma_1', \ldots, \bgamma_J') + \G(\bgamma)\left\{\G(\hat{\bgamma})\trans\bOmega \G(\bgamma)\right\}^{-1} \cdot\\ 
    &\hspace{10em}\G(\hat{\bgamma})\trans\bOmega \left\{\hat{\bnu} - \bnu(\bgamma) - \r(\hat{\bgamma}, \bgamma_1',\ldots,\bgamma_J')\right\} \\
    &= \bnu(\bgamma) + \r(\hat{\bgamma}, \bgamma_1', \ldots, \bgamma_J') + \G(\bgamma)\left\{\G(\bgamma)\trans\bOmega \G(\bgamma)\right\}^{-1} \cdot\\
    &\hspace{10em}\G(\bgamma)\trans\bOmega \left\{\hat{\bnu} - \bnu(\bgamma) - \r(\hat{\bgamma}, \bgamma_1',\ldots,\bgamma_J')\right\} \\
    &\quad + \left[\G(\bgamma)\left\{\G(\hat{\bgamma})\trans\bOmega \G(\bgamma)\right\}^{-1} \G(\hat{\bgamma})\trans\bOmega - \G(\bgamma)\left\{\G(\bgamma)\trans\bOmega \G(\bgamma)\right\}^{-1} \G(\bgamma)\trans\bOmega\right] \cdot\\ 
    &\hspace{20em}\left\{\hat{\bnu} - \bnu(\bgamma) - \r(\hat{\bgamma}, \bgamma_1',\ldots,\bgamma_J')\right\} \\
    &= \bnu(\bgamma) + \r(\hat{\bgamma}, \bgamma_1', \ldots, \bgamma_J')
+ \B(\bgamma,\bOmega)
\left\{\hat{\bnu} - \bnu(\bgamma) - \r(\hat{\bgamma}, \bgamma_1',\ldots,\bgamma_J')\right\} \\
&\quad + \left\{
\B(\hat{\bgamma},\bgamma,\bOmega)
-
\B(\bgamma,\bOmega)
\right\}
\left\{\hat{\bnu} - \bnu(\bgamma) - \r(\hat{\bgamma}, \bgamma_1',\ldots,\bgamma_J')\right\}.
\end{aligned}
\end{equation*}

Since $\bnu(\bgamma)$ is twice differentiable and
$\G(\bgamma)=\partial\bnu(\bgamma)/\partial\bgamma$,
the mapping $\G(\cdot)$ is continuous. Therefore,
$\B(\boldsymbol{\gamma}_1,\boldsymbol{\gamma}_2,\boldsymbol{\Omega})$
is continuous in
$(\boldsymbol{\gamma}_1,\boldsymbol{\gamma}_2)$.
In addition, since $\bg$ lies in a compact space and $\bOmega$ has bounded singular values, we can show that $\B(\boldsymbol{\gamma}_1,\boldsymbol{\gamma}_2,\boldsymbol{\Omega})$ has bounded singular values. In addition, since $\hat{\bgamma}\overset{p}{\to} \bgamma$ by Theorem~\ref{thm:consistency_stochastic_parameters}, by the continuous mapping theorem we obtain
\begin{equation*}
    \B(\hat{\bgamma},\bgamma,\bOmega)
    \overset{p}{\to}
    \B(\bgamma,\bOmega).
\end{equation*}
Moreover, we have
\begin{equation*}
    \left\|\r(\hat{\bgamma}, \bgamma_1',\ldots,\bgamma_J')\right\|_2 \leq \sqrt{J} \max_{j=1,\ldots,J} \left|r_j(\hat{\bg}, \bg_1', \ldots, \bg_J')\right| \leq c \|\hat{\bg}-\bg\|_2^2 = o_p(1),
\end{equation*}
where $c$ is some finite positive constant. Here the inequality uses the fact that $\partial^2 \nu_j(\bg) / (\partial\bg\partial\bg\trans)$ has bounded singular values, and the last equality uses $\hat{\bgamma}\overset{p}{\to} \bgamma$ by Theorem~\ref{thm:consistency_stochastic_parameters}. Together with $\hat{\bnu} \overset{p}{\to} \bnu(\bg)$ by Theorem~\ref{theorem:consistency_nu_hat}, by Slutsky's theorem we can obtain 
\begin{equation*}
    \left\{
\B(\hat{\bgamma},\bgamma,\bOmega)
-
\B(\bgamma,\bOmega)
\right\} \left\{\hat{\bnu} - \bnu(\bgamma) - \r(\hat{\bgamma}, \bgamma_1',\ldots,\bgamma_J')\right\} = o_p(1).
\end{equation*}
So we further have
\begin{equation*}
\begin{aligned}
    \bnu(\hat{\bgamma}) = \bnu(\bg) + \B(\bg,\bOmega)\hat{\bnu} - \B(\bg,\bOmega)\bnu(\bg) + o_p(1).
\end{aligned}
\end{equation*}
Therefore, we obtain
\begin{equation*}
    \cov\left\{\hat{\bnu}, \bnu(\hat{\bg})\right\} = \cov\left\{\hat{\bnu}, \B(\bg,\bOmega)\hat{\bnu}\right\} + o(1) = \var(\hat{\bnu}) \B(\bg,\bOmega)\trans + o(1) = \V\B(\bg,\bOmega)\trans + o(1),
\end{equation*}
and 
\begin{equation*}
    C = \mathbb{E}\left\{\left\|\hat{\bnu} - \bnu(\hat{\bgamma}) \right\|_{\bOmega}^2 \right\} + 2 \tr\left[ \cov\left\{\hat{\bnu}, \bnu(\hat{\bgamma})\right\} \bOmega\right] = \mathbb{E}\left\{\left\|\hat{\bnu} - \bnu(\hat{\bgamma}) \right\|_{\bOmega}^2 \right\} + 2 \tr\left[ \bm{\V} \B(\bg,\bOmega)\trans \bOmega\right] + o(1).
\end{equation*}
We recall that 
\begin{equation*}
    \hat{C} =  \left\|\hat{\bnu} - \bnu(\hat{\bgamma}) \right\|_{\hat{\bOmega}}^2 + 2 \tr\left[ \hat{\V} \B(\hat{\bgamma}, \hat{\bOmega})\trans \hat{\bOmega}\right],
\end{equation*}
and thus we have
\begin{equation*}
\begin{aligned}
    \hat{C} - C &= \left\|\hat{\bnu} - \bnu(\hat{\bgamma}) \right\|_{\hat{\bOmega}}^2 - \mathbb{E}\left\{\left\|\hat{\bnu} - \bnu(\hat{\bgamma}) \right\|_{\bOmega}^2 \right\} + 2 \tr\left[ \hat{\V} \B(\hat{\bgamma}, \hat{\bOmega})\trans \hat{\bOmega}\right] - 2 \tr\left[ \bm{\V} \B(\bg,\bOmega)\trans \bOmega\right] + o(1) \\
    &= \left\{\hat{\bnu} - \bnu(\hat{\bg})\right\}\trans \hat{\bOmega}\left\{\hat{\bnu} - \bnu(\hat{\bg})\right\} - \mathbb{E} \left[ \left\{\hat{\bnu} - \bnu(\hat{\bg})\right\}\trans \bOmega \left\{\hat{\bnu} - \bnu(\hat{\bg})\right\} \right] + o_p(1) \\
    &= \left\{\hat{\bnu} - \bnu(\hat{\bg})\right\}\trans (\hat{\bOmega}-\bOmega) \left\{\hat{\bnu} - \bnu(\hat{\bg})\right\} + \left\{\hat{\bnu} - \bnu(\hat{\bg})\right\}\trans \bOmega \left\{\hat{\bnu} - \bnu(\hat{\bg})\right\} - \\ & \quad \quad \mathbb{E} \left[ \left\{\hat{\bnu} - \bnu(\hat{\bg})\right\}\trans \bOmega \left\{\hat{\bnu} - \bnu(\hat{\bg})\right\} \right] + o_p(1) \\
    &= \left\{\hat{\bnu} - \bnu(\hat{\bg})\right\}\trans \bOmega \left\{\hat{\bnu} - \bnu(\hat{\bg})\right\} - \mathbb{E} \left[ \left\{\hat{\bnu} - \bnu(\hat{\bg})\right\}\trans \bOmega \left\{\hat{\bnu} - \bnu(\hat{\bg})\right\} \right] + o_p(1).
\end{aligned}
\end{equation*}

The second equality follows from $\hat{\V}\overset{p}{\to}\V$,
$\hat{\bOmega}\overset{p}{\to}\bOmega$, and
$\B(\hat{\bgamma},\hat{\bOmega})\overset{p}{\to}\B(\bgamma,\bOmega)$
by the continuous mapping theorem. The last equality is because $\left\{\hat{\bnu} - \bnu(\hat{\bg})\right\}\trans (\hat{\bOmega}-\bOmega) \left\{\hat{\bnu} - \bnu(\hat{\bg})\right\} = o_p(1)$. Indeed, $\hat{\bOmega}\overset{p}{\to}\bOmega$ and both matrices have bounded singular values, $\hat{\bnu}$ is bounded as the wavelet coefficients $W_{k,j}$ with $k=1,\ldots,M_j$ and $j=1,\ldots,J$ are all bounded, $\bnu(\hat{\bg})$ is bounded as $\bnu(\cdot)$ is continuous and $\hat{\bg}$ lies in a compact space. Lastly, by Chebyshev's inequality, for an arbitrary constant $M>0$ we have 
\begin{equation*}
\begin{aligned}
    \mathbb{P}\left( \left|\left\{\hat{\bnu} - \bnu(\hat{\bg})\right\}\trans \bOmega \left\{\hat{\bnu} - \bnu(\hat{\bg})\right\} - \mathbb{E} \left[ \left\{\hat{\bnu} - \bnu(\hat{\bg})\right\}\trans \bOmega \left\{\hat{\bnu} - \bnu(\hat{\bg})\right\} \right]\right| \geq M\right) \leq \\
    & \hspace{-15em} M^{-2} \var\left( \left\{\hat{\bnu} - \bnu(\hat{\bg})\right\}\trans \bOmega \left\{\hat{\bnu} - \bnu(\hat{\bg})\right\} \right) = o(1),
\end{aligned}
\end{equation*}
and thus we get
\begin{equation*}
    \left\{\hat{\bnu} - \bnu(\hat{\bg})\right\}\trans \bOmega \left\{\hat{\bnu} - \bnu(\hat{\bg})\right\} - \mathbb{E} \left[ \left\{\hat{\bnu} - \bnu(\hat{\bg})\right\}\trans \bOmega \left\{\hat{\bnu} - \bnu(\hat{\bg})\right\} \right] = o_p(1),
\end{equation*}
and
\begin{equation*}
    \hat{C}-C= \left\{\hat{\bnu} - \bnu(\hat{\bg})\right\}\trans \bOmega \left\{\hat{\bnu} - \bnu(\hat{\bg})\right\} - \mathbb{E} \left[ \left\{\hat{\bnu} - \bnu(\hat{\bg})\right\}\trans \bOmega \left\{\hat{\bnu} - \bnu(\hat{\bg})\right\} \right] + o_p(1) = o_p(1).
\end{equation*}
Moreover, since $\hat{\bnu}$, $\bnu(\hat{\bg})$ are bounded and $\hat{\bOmega}$ has bounded singular values, we have $\|\hat{\bnu}-\bnu(\hat{\bg})\|_{\hat{\bOmega}}^2$ to be bounded. Moreover, $\hat{\V}$ and $\B(\hat{\bg},\bOmega)$ have bounded singular values, so we also have $\tr\left[ \hat{\V} \B(\hat{\bgamma}, \hat{\bOmega})\trans \hat{\bOmega}\right]$ to be bounded, and thus $\hat{C}$ is bounded. Therefore, by the dominated convergence theorem, we also obtain $\mathbb{E}(\hat{C})=C+o(1)$.
\end{proof}

The asymptotic properties of the AIC have been studied for a large class of models (see e.g., \citealp{mcquarrie1998regression, claeskens2008model}). The AIC-like behavior of the proposed criterion can be understood from its GMM structure. The GMWM estimator minimizes a quadratic discrepancy between the empirical WV and the model-implied WV, and is therefore a GMM-type estimator based on WV moments. When $\bnu(\bgamma)$ is linear in $\bgamma$,  the lack-of-fit term is a weighted least-squares problem in the empirical WV. More generally, under regularity conditions, a first-order expansion around the population target gives the same local weighted least-squares form. Based on the weight $\V^{-1}$, the lack-of-fit term is therefore a standardized quadratic form, analogous to the quadratic forms underlying Hansen's $J$-test, likelihood-ratio arguments, and the classical AIC expansion.

This interpretation is consistent with the GMM model and moment selection framework of \citet{andrews2001consistent}. Their AIC-type criterion has the form of a GMM $J$-statistic with a correction depending on the numbers of moments and parameters. In our setting, the number of WV moments is fixed across candidate stochastic models. 
Therefore, for a candidate model with $q_m$ stochastic parameters, their AIC-type GMM criterion reduces, up to a constant common to all candidates, to a quadratic lack-of-fit term plus $2q_m$, matching the structure of the proposed criterion. 
This supports the usual AIC-type interpretation in regular nested comparisons: underfitted models are expected to be selected with probability tending to zero, whereas overfitted models may retain a positive limiting selection probability.

\newpage

\subsection{Computational Results}
\label{app:comp_results}

\subsubsection{Computation of the Theoretical WV}
\label{app:fast_theo_wv}

Let $\{\varepsilon_i\}_{i=1}^n$ be the stochastic error process such that $\mathbb{E}[\boldsymbol{\varepsilon}] = \boldsymbol{0}$ and $\var(\boldsymbol{\varepsilon}) = \bm{\Sigma} \in \real^{n \times n}$ where $\boldsymbol{\varepsilon} = \left(\varepsilon_1 \; \ldots \; \varepsilon_n \right)^{\trans}$. Define the average of the entries of the $h^{\text{th}}$ super diagonal of $\boldsymbol{\Sigma}$ with 
\begin{equation*}
    \gamma(h) = \frac{1}{n-h}\sum_{t=1}^{n-h}\cov(\varepsilon_t, \varepsilon_{t+h}), \text{ where } h = 0, \ldots, n-1.
\end{equation*}
We define our approximation for the theoretical WV of $\mathbf{\varepsilon}$ as:
\begin{equation}
    \nu^{\star}_j = \frac{1}{L_j}\left\{\gamma(0) - \gamma\left(\frac{L_j}{2}\right)\right\} + \frac{2}{L_j^2}\sum_{h=1}^{\frac{L_j}{2}-1}h\left\{2\gamma\left(\frac{L_j}{2}-h\right)-\gamma(h) - \gamma(L_j-h)\right\},
    \label{eq:approx_theo_wv}
\end{equation}
where $L_j = 2^j$. Recall that the theoretical WV is given by:
\begin{equation}
        \nu_j = \mathbb{E}(\hat{\nu}_j) = \frac{1}{M_j}\sum_{t=1}^{M_j}\mathbb{E}\left(W_{j, t}^2\right),
        \label{eq:theo_wv_no_approx}
\end{equation}
where $M_j = n-L_j+1$.

\begin{Theorem}[Order of approximation for the theoretical WV]
\label{theorem:order_approximation_theo_wv}
Given the theoretical WV $\nu_j$ defined in \eqref{eq:theo_wv_no_approx} and the proposed approximation $\nu_j^\star$ defined in \eqref{eq:approx_theo_wv}, we have:
\begin{equation*}
    \nu_j = \nu_j^\star +\mathcal{O}(n^{-1}).
\end{equation*}
\end{Theorem}

We now prove Theorem~\ref{theorem:order_approximation_theo_wv}.
\begin{proof}
    Consider the theoretical WV defined as:
   \begin{equation}
   \label{eq:theo_wv_with_approx}
 \begin{aligned}
        \nu_j = \mathbb{E}(\hat{\nu}_j) &= \frac{1}{M_j}\sum_{t=1}^{M_j}\mathbb{E}\left(W_{j, t}^2\right)\\
        &=\frac{1}{M_j}\sum_{t=1}^{M_j}\var(W_{j,t})\\
        &=\frac{1}{M_j}\sum_{t=1}^{M_j}\sum_{l_1=0}^{L_j-1}\sum_{l_2=0}^{L_j-1}\cov\left(h_{j, l_1}\varepsilon_{t+l_1}, h_{j, l_2}\varepsilon_{t+l_2}\right)\\
        &= \sum_{l_1=0}^{L_j-1}\sum_{l_2=0}^{L_j-1}h_{j, l_1}h_{j, l_2} \frac{1}{M_j}\sum_{t=1}^{M_j}\cov \left(\varepsilon_{t+l_1}, \varepsilon_{t+l_2}\right) \\
        &= \frac{1}{2^{2j}}\sum_{l_1=0}^{L_j-1}\sum_{l_2=0}^{L_j-1}(-1)^{\mathbbm{1}\left\{l_1 < \frac{L_j}{2}\right\}}(-1)^{\mathbbm{1}\left\{l_2 < \frac{L_j}{2}\right\}} \frac{1}{M_j}\sum_{t=1}^{M_j}\cov(\varepsilon_{t+l_1}, \varepsilon_{t+l_2})\\
        &=\frac{1}{2^{2j}}\sum_{l_1=0}^{L_j-1}\sum_{l_2=0}^{L_j-1}(-1)^{\mathbbm{1}\left\{l_1< \frac{L_j}{2}\right\}}(-1)^{\mathbbm{1}\left\{l_2< \frac{L_j}{2}\right\}} \left\{\gamma(|l_2-l_1|) + \Delta(|l_2 - l_1|) \right\}\\
        &= \nu_j^\star +\frac{1}{2^{2j}}\sum_{l_1=0}^{L_j-1}\sum_{l_2=0}^{L_j-1}(-1)^{\mathbbm{1}\left\{l_1< \frac{L_j}{2}\right\}}(-1)^{\mathbbm{1}\left\{l_2< \frac{L_j}{2}\right\}} \Delta(|l_2 - l_1|),
 \end{aligned}
\end{equation}
where $\Delta(h) = \frac{1}{M_j}\sum_{t=1}^{M_j}\cov(\varepsilon_t, \varepsilon_{t+h}) - \gamma(h)$, with $h=0, \ldots, L_j-1$.
\\
Note that:
\begin{equation*}
   \begin{aligned}
        \frac{1}{M_j}\sum_{t=1}^{M_j}\cov(\varepsilon_t, \varepsilon_{t+h}) &= \frac{1}{n-L_j + 1}\sum_{t=1}^{n-L_j+1}\cov(\varepsilon_t, \varepsilon_{t+h})\\
        &=\frac{1}{n-L_j + 1}\sum_{t=1}^{n-h -\{(L_j-1) -h\}}\cov(\varepsilon_t, \varepsilon_{t+h}) \\
        &=\frac{1}{n-L_j + 1}\sum_{t=1}^{n-h}\cov(\varepsilon_t, \varepsilon_{t+h}) - \frac{1}{n-L_j +1}\sum_{t=n-(L_j-1)+1}^{n-h}\cov(\varepsilon_t, \varepsilon_{t+h}) \\
        &= \frac{n-h}{n-L_j+1}\gamma(h) -  \frac{1}{n-L_j +1}\sum_{t=n-L_j+2}^{n-h}\cov(\varepsilon_t, \varepsilon_{t+h}). 
   \end{aligned}
\end{equation*}
Hence, we have:
\begin{equation*}
    \begin{aligned}
        \Delta(h) &= \frac{1}{M_j}\sum_{t=1}^{M_j}\cov(\varepsilon_t, \varepsilon_{t + h}) - \gamma(h)\\
        &=\frac{n-h-n+L_j-1}{n-L_j+1}\gamma(h) - \frac{1}{n-L_j + 1}\sum_{t=n-L_j + 2}^{n-h}\cov(\varepsilon_t, \varepsilon_{t+h})\\
        &=\frac{L_j-h-1}{n-L_j+1}\gamma(h)-\frac{1}{n-L_j + 1}\sum_{t=n-L_j + 2}^{n-h}\cov(\varepsilon_t, \varepsilon_{t+h}).
    \end{aligned}
\end{equation*}
\\
Therefore, using the bounded second moment of $\boldsymbol{\varepsilon}$ implied by the model defined in \eqref{eq.reg_complete}, we have:
\begin{equation*}
    \begin{aligned}
        |\Delta(h)| \leq C_2 \frac{L_j-h-1}{n - L_j + 1} =C_2 \frac{(L_j-1)-h}{n-(L_j-1)},
    \end{aligned}
\end{equation*}
where $C_2$ is some finite positive constant. Then we further have:
\begin{equation*}
    \begin{aligned}
        \max_{h=0,\ldots,L_j-1} |\Delta(h)| \leq C_2 \frac{L_j-1}{n - (L_j-1)} = C_2 \frac{1}{\frac{n}{L_j-1}-1} \leq C_2 \frac{1}{\frac{n}{2^J-1}-1}.
    \end{aligned}
\end{equation*}

We can then obtain the following rate for the second term of the last equality of \eqref{eq:theo_wv_with_approx}:
\begin{equation*}
    \begin{aligned}
        \left| \frac{1}{2^{2j}}\sum_{l_1=0}^{L_j-1}\sum_{l_2=0}^{L_j-1}(-1)^{\mathbbm{1}\left\{l_1< \frac{L_j}{2}\right\}}(-1)^{\mathbbm{1}\left\{l_2< \frac{L_j}{2}\right\}} \Delta(|l_2 - l_1|) \right| =\mathcal{O}(n^{-1}).
    \end{aligned}
\end{equation*}
\\
Finally, we obtain:
\begin{equation*}
    \begin{aligned}
         \nu_j 
         &= \nu_j^\star +\frac{1}{2^{2j}}\sum_{l_1=0}^{L_j-1}\sum_{l_2=0}^{L_j-1}(-1)^{\mathbbm{1}\left\{l_1< \frac{L_j}{2}\right\}}(-1)^{\mathbbm{1}\left\{l_2< \frac{L_j}{2}\right\}} \Delta(|l_2 - l_1|) = \nu_j^\star + \mathcal{O}(n^{-1}).
    \end{aligned}
\end{equation*}

\end{proof}

\newpage

\subsubsection{Approximation of the Theoretical WV using Residuals}
\label{app:IminusH}

Theorem~\ref{theorem:consistency_nu_hat} in Appendix~\ref{app:consistency_nu_hat} state that the empirical WV computed on the estimated vector of residuals $\hat{\boldsymbol{\varepsilon}}$ denoted as $\hat{\boldsymbol{\nu}}$ converges in probability to the theoretical WV of the process $\tilde{\boldsymbol{\varepsilon}}$ denoted as $\boldsymbol{\nu}$. Note that $\hat{\boldsymbol{\nu}}$ and $\boldsymbol{\nu}$ have the same dimension that is finite. More specifically, we show in Theorem~\ref{theorem:consistency_nu_hat} that each entry of $\hat{\boldsymbol{\nu}}$ converges in probability to the corresponding entry of $\boldsymbol{\nu}$. Based on this result, we could compute the theoretical WV $\boldsymbol{\nu}\left( \bm{\gamma}, {\boldsymbol{\vartheta}} \right)$ in \eqref{eq:gmwm_new} based on the variance covariance of $\tilde{\boldsymbol{\varepsilon}}$, i.e., $\var(\tilde{\boldsymbol{\varepsilon}}) =\boldsymbol{\Sigma}\left(\boldsymbol{\gamma}\right) \odot \left\{\boldsymbol{\Lambda}(\boldsymbol{\vartheta}) + \mu(\boldsymbol{\vartheta})^2 \mathbf{1} \mathbf{1}^{\trans} \right\}$. However, in this appendix, we propose a finite sample approximation that we use in practice. More precisely, we propose to obtain the theoretical WV by using the approximation described in Appendix~\ref{app:fast_theo_wv} computed on the following matrix:
\begin{equation}
\label{eq:matrix_on_which_compute_theo_wv}
    \left\{(\mathbf{I}-\P)\boldsymbol{\Sigma}\left(\boldsymbol{\gamma}\right)(\mathbf{I}-\P)\right\} \odot \left\{  \boldsymbol{\Lambda}(\boldsymbol{\vartheta}) + \mu(\boldsymbol{\vartheta})^2 \mathbf{1} \mathbf{1}^{\trans} \right\} \text{ ,}
\end{equation}
where $\I$ is a $n\times n$ identity matrix, $\P$ is the projection matrix $\P={\X}\left({\X}^{\trans}{\X}\right)^{-1}{\X}^{\trans}$ and $\boldsymbol{\Sigma}(\boldsymbol{\gamma})$ is the variance covariance of the error process $\boldsymbol{\varepsilon}$ parametrized by $\boldsymbol{\gamma}$. 

\noindent We motivate this finite sample approximation as follows.
Recall the definition of the estimated residuals $\hat{\boldsymbol{\varepsilon}} = \tilde{\Y}- \tilde{\X}\hat{\boldsymbol{\beta}}$ and $\tilde{\boldsymbol{\varepsilon}} = \boldsymbol{\varepsilon} \odot \Z$. By Theorem~\ref{theorem:short_memory}, we have that $\forall i,j , \bigg|\var(\hat{\boldsymbol{\varepsilon}})_{i,j} - \var(\tilde{\boldsymbol{\varepsilon}})_{i,j}\bigg| \overset{p}{\longrightarrow}0 $. 
Moreover, we have that $\var\left(\tilde{\boldsymbol{\varepsilon}}\right) = \boldsymbol{\Sigma}\left(\boldsymbol{\gamma}\right) \odot \left\{\boldsymbol{\Lambda}(\boldsymbol{\vartheta}) + \mu(\boldsymbol{\vartheta})^2 \mathbf{1} \mathbf{1}^{\trans} \right\}$. This result is easily obtained using the independence of $\Z$ and $\boldsymbol{\varepsilon}$ as well as $\mathbb{E}(\boldsymbol{\varepsilon}) = \mathbf{0}$. Indeed we have:
\begin{equation*}    \begin{aligned}
\var(\tilde{\boldsymbol{\varepsilon}}) &=\var(\boldsymbol{\varepsilon} \odot \Z)\\
&=\mathbb{E}\left(\boldsymbol{\varepsilon} \boldsymbol{\varepsilon}^{\trans}\right)\odot \mathbb{E}\left(\Z \Z^{\trans}\right)-
\mathbb{E}(\boldsymbol{\varepsilon})\mathbb{E}(\boldsymbol{\varepsilon})^{\trans} \odot \mathbb{E}(\Z)\mathbb{E}(\Z)^{\trans} \\
&=  \left\{\boldsymbol{\Sigma}(\boldsymbol{\gamma})+ \mathbb{E}(\boldsymbol{\varepsilon})\mathbb{E}(\boldsymbol{\varepsilon})^{\trans}\right\} \odot \left\{\boldsymbol{\Lambda}(\boldsymbol{\vartheta})+\mathbb{E}(\Z) \mathbb{E}(\Z)^{\trans} \right\}-\left\{\mathbb{E}(\boldsymbol{\varepsilon})\mathbb{E}(\boldsymbol{\varepsilon})^{\trans}\right\} \odot \left\{\mathbb{E}(\Z) \mathbb{E}(\Z)^{\trans}\right\} \\
&= \boldsymbol{\Sigma}(\boldsymbol{\gamma}) \odot \boldsymbol{\Lambda}(\boldsymbol{\vartheta}) + \boldsymbol{\Sigma}(\boldsymbol{\gamma}) \odot \left\{\mathbb{E}(\Z) \mathbb{E}(\Z)^{\trans}\right\} \\
&= \boldsymbol{\Sigma}\left(\boldsymbol{\gamma}\right) \odot \left\{\boldsymbol{\Lambda}(\boldsymbol{\vartheta}) + \mu(\boldsymbol{\vartheta})^2 \mathbf{1} \mathbf{1}^{\trans} \right\}.
\end{aligned}
\end{equation*}
Also, consider the estimated residuals when there are no missing observations defined as $\hat{\boldsymbol{\varepsilon}}^{\text{comp}}=\mathbf{Y}-\mathbf{X}\hat{\boldsymbol{\beta}}$ where $\hat{\boldsymbol{\beta}}$ is the estimator with $\tilde{\mathbf{Y}} = \mathbf{Y}$ and $\tilde{\mathbf{X}} = \mathbf{X}$ (i.e., with complete observations). We have that $\var(\hat{\boldsymbol{\varepsilon}}^{\text{comp}} \odot \Z) =  \left\{(\mathbf{I}-\P)\boldsymbol{\Sigma}\left(\boldsymbol{\gamma}\right)(\mathbf{I}-\P)\right\} \odot \left\{  \boldsymbol{\Lambda}(\boldsymbol{\vartheta}) + \mu(\boldsymbol{\vartheta})^2 \mathbf{1} \mathbf{1}^{\trans} \right\} $ and since $\hat{\boldsymbol{\varepsilon}}^{\text{comp}} \xrightarrow{p} \boldsymbol{\varepsilon}$ by Theorem~\ref{theorem:short_memory}, we have that 
\begin{equation*}
    \Big[\boldsymbol{\Sigma}\left(\boldsymbol{\gamma}\right) \odot \left\{\boldsymbol{\Lambda}(\boldsymbol{\vartheta}) + \mu(\boldsymbol{\vartheta})^2 \mathbf{1} \mathbf{1}^{\trans} \right\}\Big]_{i,j} - \Big[\left\{(\mathbf{I}-\P)\boldsymbol{\Sigma}\left(\boldsymbol{\gamma}\right)(\mathbf{I}-\P)\right\} \odot \left\{  \boldsymbol{\Lambda}(\boldsymbol{\vartheta}) + \mu(\boldsymbol{\vartheta})^2 \mathbf{1} \mathbf{1}^{\trans} \right\}\Big]_{i,j} \xrightarrow{p}0.
\end{equation*}
Therefore, we have $\forall (i,j)$:
\begin{equation*}
\begin{aligned}
    \Big|\var(\hat{\boldsymbol{\varepsilon}})_{i,j} - \var(\hat{\boldsymbol{\varepsilon}}^{\text{comp}} \odot \Z)_{i,j}\Big| 
    &= \Big| \var(\hat{\boldsymbol{\varepsilon}})_{i,j} -\var(\boldsymbol{\varepsilon} \odot \Z)_{i,j} + \var(\boldsymbol{\varepsilon} \odot \Z)_{i,j}   - \var(\hat{\boldsymbol{\varepsilon}}^{\text{comp}} \odot \Z)_{i,j}\Big| \\
    &\leq \Big| \var(\hat{\boldsymbol{\varepsilon}})_{i,j} -\var(\boldsymbol{\varepsilon} \odot \Z)_{i,j}\Big| + \Big| \var(\boldsymbol{\varepsilon} \odot \Z)_{i,j}   - \var(\hat{\boldsymbol{\varepsilon}}^{\text{comp}} \odot \Z)_{i,j}\Big| \\
    &= o_p(1) + o_p(1) \xrightarrow{p}0.
\end{aligned}
\end{equation*}

A problem with this approach though is that, within the GMWM, the matrix in \eqref{eq:matrix_on_which_compute_theo_wv} should be computed at each evaluation of the GMWM objective function which requires two matrix multiplications of $n \times n$ matrices (entailing a computational complexity of order $\mathcal{O}(n^3)$) as well as the Hadamard product of two $n \times n $ matrices (which has a computational complexity of order $\mathcal{O}(n^2)$).  
\\ 
To overcome this computational bottleneck, we rely on the approximation of the theoretical WV proposed in \ref{app:fast_theo_wv}  based on Lemma 1 of \cite{xu2017study}, which states that the theoretical WV of a (zero-mean) non-stationary process is a linear function in $L_J$ of the averages taken over the diagonal and superdiagonals elements of the process covariance matrix. Considering that $\left\{\boldsymbol{\Lambda}({\boldsymbol{\vartheta}}) + \mu({\boldsymbol{\vartheta}})^2 \mathbf{1} \mathbf{1}^{\trans} \right\} $ is a Toeplitz matrix due to the stationarity of $\Z$, the strategy that we propose to efficiently compute $\boldsymbol{\nu}\left( \bm{\gamma}, \boldsymbol{\vartheta} \right)$ is built on different considerations. We start by considering the case where the error process $\bm{\varepsilon}$ is stationary. Based on this, using the notation $\R = \I- \P$, the first consideration is that $\R\boldsymbol{\Sigma}(\boldsymbol{\gamma})$ is a close approximation to $\R\boldsymbol{\Sigma}(\boldsymbol{\gamma})\R$. An intuitive argument as to why this approximation can be considered reasonable is given by analyzing the trace of the matrix $\R\boldsymbol{\Sigma}(\boldsymbol{\gamma})\R$. Indeed we have:

\begin{equation*}
  \begin{aligned}
        \tr\left[ (\I-\P)\boldsymbol{\Sigma}(\boldsymbol{\gamma})(\I-\P) \right] &=  \tr\left[\boldsymbol{\Sigma}(\boldsymbol{\gamma}) - \boldsymbol{\Sigma}(\boldsymbol{\gamma})\P - \P \boldsymbol{\Sigma}(\boldsymbol{\gamma}) + \P\boldsymbol{\Sigma}(\boldsymbol{\gamma})\P\right] \\
        &=\tr\left[ \boldsymbol{\Sigma}(\boldsymbol{\gamma}) \left(\I -\P \right) - \P \boldsymbol{\Sigma}(\boldsymbol{\gamma}) \left(\I - \P\right) \right] \\
        &= \tr\left[ \boldsymbol{\Sigma}(\boldsymbol{\gamma}) \left(\I -\P \right) - \P \boldsymbol{\Sigma}(\boldsymbol{\gamma}) + \P \boldsymbol{\Sigma}(\boldsymbol{\gamma})\P \right] \\
        &=  \tr\left[\boldsymbol{\Sigma}(\boldsymbol{\gamma}) \left(\I -\P \right) \right] -  \tr\left[ \P \boldsymbol{\Sigma}(\boldsymbol{\gamma})\right] +  \tr\left[ \P \boldsymbol{\Sigma}(\boldsymbol{\gamma})\P  \right] \\
        &= \tr\left[\boldsymbol{\Sigma}(\boldsymbol{\gamma}) \left(\I -\P \right) \right] -  \tr\left[ \P \boldsymbol{\Sigma}(\boldsymbol{\gamma})\right] +  \tr\left[ \P\P \boldsymbol{\Sigma}(\boldsymbol{\gamma})  \right] \\
        &= \tr\left[\boldsymbol{\Sigma}(\boldsymbol{\gamma})\left(\I -\P \right) \right] -  \tr\left[ \P \boldsymbol{\Sigma}(\boldsymbol{\gamma})\right] +  \tr\left[ \P \boldsymbol{\Sigma}(\boldsymbol{\gamma})  \right] \\
        &= \tr\left[ \left(\I -\P \right) \boldsymbol{\Sigma}(\boldsymbol{\gamma})\right] = \tr\left[ \R \boldsymbol{\Sigma}(\boldsymbol{\gamma})\right]
  \end{aligned}
\end{equation*}
\\
 Hence, we will first approximate $\R\boldsymbol{\Sigma}(\boldsymbol{\gamma})\R$ by $\R\boldsymbol{\Sigma}(\boldsymbol{\gamma})$. Now, given the approximation discussed in Appendix \ref{app:fast_theo_wv}, we are interested in efficiently computing the average over the diagonal and superdiagonals of the matrix $\R\boldsymbol{\Sigma}(\boldsymbol{\gamma})$ where $\R$ is idempotent and symmetric and $\boldsymbol{\Sigma}(\boldsymbol{\gamma})$ is symmetric and Toeplitz (assuming a stationarity of the error process $\bm{\varepsilon}$). Using $\rho(k) = \cov(\varepsilon_t , \varepsilon_{t+k})$, with $k=0, \ldots , n-1$, to denote the autocovariance function of $\bm{\varepsilon}$, without lack of generality, we define the sum of elements over the diagonal and over the $l^{\text{th}}$ superdiagonal of $\R \boldsymbol{\Sigma}(\boldsymbol{\gamma})$ as
\begin{equation*}
    S_l=\sum_{i=1}^{n-l}\left[\R \boldsymbol{\Sigma}(\boldsymbol{\gamma})\right]_{i, i+l} = \sum_{k=1}^n \sum_{i=1}^{n-l} R_{i k} \rho(k-i-l),\; \text{ for } l=0, \ldots , n-1 \text{.}
\end{equation*}

 It can be shown that
\begin{equation}
\label{eq::decomposition_s_l}
    S_l = \sum_{j=1}^{m-1} \left[\rho(-l-j) \sum_{i=1}^{m-j} R_{i+j, i}\right]
+
\sum_{j=0}^l \left[ \rho(-j) \sum_{i=1}^m R_{i, i+l-j}\right]
+
\sum_{j=1}^{m-1} \left[ \rho(j) \sum_{i=1}^{m-j} R_{i, i+l+j}\right], 
\end{equation}
where $m = n-l$, which allows to pre-compute the quantities that depend on $\R$ outside of the optimization function and to simply index them in the optimization operation.

 Nevertheless, evaluating all the quantities involving $\R$ in \eqref{eq::decomposition_s_l} remains computationally expensive, especially as $n$ grows. To address this issue, we propose the following computational strategy. Let
\[
\S = \left(S_0, S_1, \ldots, S_{n-1}\right)^{\trans}
\quad \text{and} \quad
\boldsymbol{\rho} = \left(\rho(0), \rho(1), \ldots, \rho(n-1)\right)^{\trans}
\]
denote, respectively, the vector of diagonal and superdiagonal sums and the autocovariance vector of the error process $\boldsymbol{\varepsilon}$. Rather than computing all entries of $\S$, we evaluate $S_l$ only for a selected subset of lags and consider the corresponding differences between $\S$ and $\boldsymbol{\rho}$. The full vector $\S$ is then approximated by linearly interpolating these differences across lags. This approach substantially reduces the computational burden. While a formal approximation error bound for this interpolation strategy is currently unavailable, our simulation results suggest that it provides an accurate approximation of $\S$ and of the resulting theoretical WV across the range of sample sizes, missing-data settings, and stochastic models considered in this work. Empirically, we found that the number of points on which to evaluate $\S$ can be chosen of the order $n^\star=a+b\log_2(n)$, with $a=1$ and $b=3$.
We can then rely on the approximation proposed in Appendix~\ref{app:fast_theo_wv} together with the interpolation strategy described above to obtain a computationally efficient approximation of the theoretical WV computed on the estimated residuals process $\hat{\boldsymbol{\varepsilon}}$.

It is important to note that, when the stochastic process $\boldsymbol{\varepsilon}$ is stationary, the approximation is due to (i) considering $\R\boldsymbol{\Sigma}(\boldsymbol{\gamma})$ instead of $\R\boldsymbol{\Sigma}(\boldsymbol{\gamma})\R$ and (ii) because we approximate $\S$ using a linear interpolation. When the process  $\boldsymbol{\varepsilon}$ is not stationary, we introduce an additional approximation step. Indeed the result obtained in \eqref{eq::decomposition_s_l} is only valid in case $\boldsymbol{\Sigma}(\boldsymbol{\gamma})$ is a Toeplitz matrix which allows us to work directly with the autocovariance function. However, in the case of a non-stationary process, the autocovariance function is not clearly defined. Therefore, when the process $\boldsymbol{\varepsilon}$ is not stationary, we consider the same approximation strategy but consider the average over the diagonal and superdiagonals of the covariance matrix of $\boldsymbol{\varepsilon}$ instead of $\boldsymbol{\rho}$. For the power-law noises considered in this work, the average of the diagonal and superdiagonals can be efficiently computed using the strategy described above.
\\

\subsubsection{Computation of the WV Covariance Matrix}
\label{app:compute_v}

Consider the error process $\bm{\varepsilon}$ to be a stationary process with  autocovariance function $\rho(k) = \cov(\varepsilon_i, \varepsilon_{i+k})$. We also define the autocovariance of the wavelet coefficients at scale $j$ as 
$$
f_j(h) = \cov\left(W_{j, t}, W_{j, t+h}\right) .
$$
We now construct the autocovariance of the wavelet coefficients issued from the Haar wavelet filter at scale $j=1$, $f_1(k), k = 1, \ldots n$ using the following relationship:

\begin{equation*}
f_1(h)  =\frac{1}{2} \rho(k)-\frac{1}{4} \rho(k-1)-\frac{1}{4} \rho(k+1).
\end{equation*}

From this expression, we can now recursively compute the autocovariance of the wavelet coefficients for scales $2,\ldots, J$ using the following formula:

  \begin{equation*}
        f_{j+1}(k)= \frac{3}{2} f_j(k)+f_j\left(k+\frac{2^j}{2}\right)+f_j\left(k-\frac{2^j}{2}\right)+\frac{1}{4} f_j\left(k+2^j\right)+\frac{1}{4} f_j\left(k-2^j\right) .
    \end{equation*}

Using these expression, the variance of the estimated WV is given by:

\begin{equation*}
     \var\left(\hat{\nu}_j\right) =  \frac{2}{M_j^2} \sum_{i=-M_{j}+1}^{M_j-1}\left(M_j-|i|\right) f_j(i)^2,
\end{equation*}

while the covariance is given by

\begin{equation*}
  \begin{aligned}
        \cov\left(\hat{\nu}_j, \hat{\nu}_{j+l}\right) & =\cov\left(\frac{1}{M_j} W_j^{\trans} W_j, \frac{1}{M_{j+l}} W_{j+l}^{\trans} W_{j0+l}\right)\\
        & = \frac{2}{M_j M_{j+l}} \sum_{t=1}^{M_{j+l}} \sum_{m=1}^{M_j} C_{t, m}^2,
  \end{aligned}
\end{equation*}

where the summation above is as follows:

\begin{equation*}
\begin{aligned}
\sum_{t=1}^{M_{j+k}} \sum_{m=1}^{M_j} C_{t, m}^2
&= \sum_{i=-M_{j}+1}^{M_{j +l}-1} \Bigl[
    M_{j + l} \mathbbm{1}\left\{2^j-2^{j+l} \leqslant t \leqslant 0\right\} \\
&\quad + \left(M_{j+l}-|i|\right) \mathbbm{1}\{i>0\}
    + \left(M_{j+l}-\left|i-2^j+2^{j+l}\right|\right)
      \mathbbm{1}\left\{i<2^j-2^{j+l}\right\}
\Bigr] \delta_i ,
\end{aligned}
\end{equation*}

where
\begin{equation*}
    \delta_i  =\left[\sum_{p=0}^{2^l-2}\left\{\frac{p+1}{2^l} f_j\left(i+p  2^{j-1}\right)+\frac{2^l-1-p}{2^l} f_j\left(i+p 2^{j-1}+2^{j+l-1}\right)\right\}+f_j\left(i+2^{j+l-1}-2^{j-1}\right)\right]^2.
\end{equation*}

When the wavelet coefficients are not stationary, using $\boldsymbol{\Sigma}$ and $\boldsymbol{\Sigma}_{W}$ to represent the covariance matrix of the error process and of the wavelet coefficients respectively, then we obtain the variance of the WV using the following procedure:

\begin{enumerate}
    \item Compute $\var\left(\W_1\right)$, where $\W_1 \in \real^{M_1}$ is the vector containing the wavelet coefficients at scale $j=1$, as follows:
$$
\begin{aligned}
&\var \left(\W_1\right)_{k, l}=\sum_{\substack{i \in\{k, k+1\} \\ j \in\{l, l+1\}}} A_{k i} A_{l j}\left(\boldsymbol{\Sigma}\right)_{i j}\\
&\text { where } \quad A_{k i}=\left\{\begin{array}{cc}
-\frac{1}{2} & \text { if } i=k \\
\frac{1}{2} & \text { if } i=k+1
\end{array} \quad A_{l j}= \begin{cases}-\frac{1}{2} & \text { if } j=l \\
\frac{1}{2} & \text { if }j=l+1.\end{cases}\right.
\end{aligned}
$$
    
    \item  Compute all the variances the wavelet coefficients at scale $j=2, \ldots, J$ as follows:

$$
\begin{aligned}
\var\left(\W_{j+1}\right)_{k l}= & \sum_{i=1}^{M_j} \sum_{j=1}^{M_j} B_{k i} \{\var\left(\W_j\right)\}_{i j}  B_{l j} \\
= 
& \sum_{\substack{
i \in\left\{k, k+\frac{L_j}{2}, k+L_j\right\}
\\
j \in\left\{l, l+\frac{L_j}{2}, 1 +L_j\right\}
}}B_{k i} B_{l j}\left[\var\left(\W_{j}\right)\right]_{i j},
\end{aligned}
$$
where
$$
B_{k i}=\left\{\begin{array}{ll}
\frac{1}{2} & \text { if } i=k \\
1 & \text { if }i=k+\frac{L_j}{2} \\
\frac{1}{2} &\text { if } i=k+L_j
\end{array} \; \text{ and } \;B_{l j}=\left\{\begin{array}{lll}
\frac{1}{2} & \text { if } j=l \\
1 & \text { if } j=l+\frac{L_j}{2} \\
\frac{1}{2} & \text { if } j=l+L_j.
\end{array}\right.\right.
$$

    \item Compute $C_{j, j+k} =\cov\left(\W_j, \W_{j+k}\right)$ as follows:

$$
C_{l m}=\sum_{p=0}^{2^{k+1}-2} C_p \left(\Sigma_{W_j}\right)_{l+p 2^{j-1}, m } \quad \begin{array}{r}
\text { with } l=1 \cdots M_{j+k} \text{ and } \\
m=1 \cdots M_j,
\end{array}
$$

where

$$
C_p= \begin{cases}\frac{p+1}{2^k} & \text { if } 0 \leqslant p \leqslant 2^k-1 \\ \frac{2^{k+1}-1-p}{2^k} & \text { if } 2^k \leqslant p \leqslant 2^{k+1}-2.\end{cases}
$$

    \item Compute
$$
\cov\left(\hat{\nu}_j, \hat{\nu}_{j+k}\right)=\frac{2}{M_j M_{j+k}}  \sum_{l=1}^{M_{j+k}} \sum_{m=1}^{M_j}\left[\cov\left(\W_j, \W_{j+k}\right)\right]_{l m}^2,
$$

where we rewrite the above equation as

$$
\cov\left(\hat{\nu}_j, \hat{\nu}_{j+k}\right)=\frac{2}{M_j M_{j+k}} \sum_{h=-M_{j+1}}^{M_{j+k}-1} C_h \tilde{f}_{j, j+k}(h)^2,
$$

where

$$
C_h= \begin{cases}M_{j+k} & \text { if } h=M_{j+k}-M_j, \cdots, 0 \\ M_{j+k}-h & \text { if } h=1, \ldots,  M_{j+k}-1 \\ M_{j+k}-\left|h-M_{j+k}+M_j\right| & \text { if } h=1-M_j, \ldots, M_{j+k}-M_j-1\end{cases}
$$

and $\tilde{f}_{j, j+k}(h)$ is defined as:

$$
\tilde{f}_{j, j+k}(h) = \begin{cases}
    \lambda_1 & \text{ if } h= M_{j+k} - M_j , \ldots, 0 \\
  \lambda_2 & \text{ if } h= 1, \ldots, M_{j+k}-1\\
   \lambda_3 & \text{ if } h= 1 - M_j, \ldots, M_{j+k} -M_j -1,
\end{cases}
$$

\noindent where 
\begin{itemize}
    \item $\lambda_1 = \sqrt{\frac{1}{M_{j+k}} \sum_{l=1}^{M_{j+k}} \cov\left(W_{j+k, l}, W_{j, l-h}\right)^2}$,
    \item $\lambda_2 = \sqrt{\frac{1}{M_{j+k}-h} \sum_{l=h+1}^{M_{j+k}} \cov\left(W_{j+k, l}, W_{j, l-h}\right)^2}$,
    \item $\lambda_3 = \sqrt{\frac{1}{M_{j+k}-|h - M_{j+k}+M_j|} \sum_{l=1}^{M_{j+k} - |h - M_{j+k}+M_j|} \cov\left(W_{j+k, l}, W_{j, l-h}\right)^2}$.
\end{itemize}
\end{enumerate}

\noindent Rewriting $\cov\left(\hat{\nu}_j, \hat{\nu}_{j+k}\right)$ with the function $\tilde{f}_{j, j+k}(h) $ allows us to construct a fast approximation of $\var\left(\hat{\bnu}\right)$ in the non-stationary setting.

\subsection{Simulation Studies}
\label{app:simulations}

\subsubsection{Complete Simulation Results}
\label{app:main_simu_results}

In this section, we provide the complete results of the simulation studies described in Section~\ref{sec:simulations}. More precisely, we provide the ratio of the length of the CI for the simulation setting that considers the stochastic model WN + MAT + RW discussed in Section~\ref{sec:simulations} as well as the results for the two other simulation settings that consider the stochastic model WN + FL and WN + PL. For both simulation settings, we present the average running time, the average type I error, and the median ratio of the length of the CI. The performance of the proposed WAMORE framework are compared to the MLE implemented in the software \texttt{Hector} v2.0 of \cite{Bos2008}. Moreover, for the missingness process $\Z$, we consider a Markov model with transition probabilities defined as follows:

\begin{equation}
\label{eq:markov_model_def}
\begin{aligned}
P(Z_2=1 \mid Z_1=1) &= 1-p_1, &\quad P(Z_2=0 \mid Z_1=1) &= p_1, \\
P(Z_2=1 \mid Z_1=0) &= p_2,   &\quad P(Z_2=0 \mid Z_1=0) &= 1-p_2 .
\end{aligned}
\end{equation}

It can be shown that the general Markov chain represented in \eqref{eq:markov_model_def} is a special example of \eqref{eq:markov_missinigness} and hence respects this requirement for the WAMORE properties to hold. Proof of this is given in Appendix \ref{app:beta_mixing}. With this definition, we can interpret $p_1 \in (0,1)$ as the probability of an observation being missing while $p_2 \in (0,1)$ is the opposite (hence, for example, a higher the value of $p_1$ implies a higher probability of missing observations, all other things remaining equal). Based on this, we have $ \mathbb{E}[Z]= \nicefrac{p_2}{p_1 + p_2}$ which can be interpreted as the general probability of observing a value. Modeling the missingness in GNSS time series using such a process seems appropriate since missing observations can occur due to environmental or human factors, like receiver antenna replacement, poor observation conditions, signal interruptions, or intense crustal movement \citep{bao2021filling}. These factors generally entail a missingness process where there is a low probability of not observing the next data point, but when a point is not observed, it is likely that several consecutive observations will also be missing. Based on this we consider nine settings of missingness which are summarized in Table~\ref{tab:parameters:missingness}.
\begin{table}[h!]
\centering
\begin{tabular}{c c c c}
\toprule
Setting & $p_1$ & $p_2$ & $\mathbb{E}[Z]$\\ 
\midrule
1 & 0.00 & - & 1 \\ 
2 & 0.05 & 0.95 & 0.95\\
3 & 0.05 & 0.45 & 0.90\\
4 & 0.06 & 0.34 & 0.85\\
5 & 0.05 & 0.20 & 0.80\\
6 & 0.06 & 0.18 & 0.75 \\
7 & 0.06 & 0.14 & 0.70\\
8 & 0.07 & 0.13 & 0.65 \\
9 & 0.10 & 0.15 & 0.60\\ 
\bottomrule
\end{tabular}
\caption{Parameter values for the missingness process $\Z$ for each setting.}
\label{tab:parameters:missingness}
\end{table}
\\
In Figure~\ref{fig:wn_matern_rw_ratio_ci_length}, we present the ratio of the length of CI for the setting considering the stochastic model WN + MAT + RW for which the average running time and empirical type I error are presented in Section~\ref{sec:simulations}. In Figure~\ref{fig:wn_fl_running_time} and Figure~\ref{fig:wn_fl_type_1_error}, we present the average running time and the empirical type I error for the setting that considers a stochastic model composed of WN and FL while Figure~\ref{fig:wn_fl_ratio_length_ci} presents the median ratio of the length of CI between the WAMORE approach and the MLE implemented in \texttt{Hector} for that setting. Finally, in Figure~\ref{fig:wn_pl_running_time}, Figure~\ref{fig:wn_pl_type_1_error} and Figure~\ref{fig:wn_pl_ratio_length_ci}, we present respectively the average running time, the empirical type I error and the median ratio of the length of CI for the setting WN + PL.
\begin{figure}
    \centering
    \includegraphics[width=0.6\linewidth]{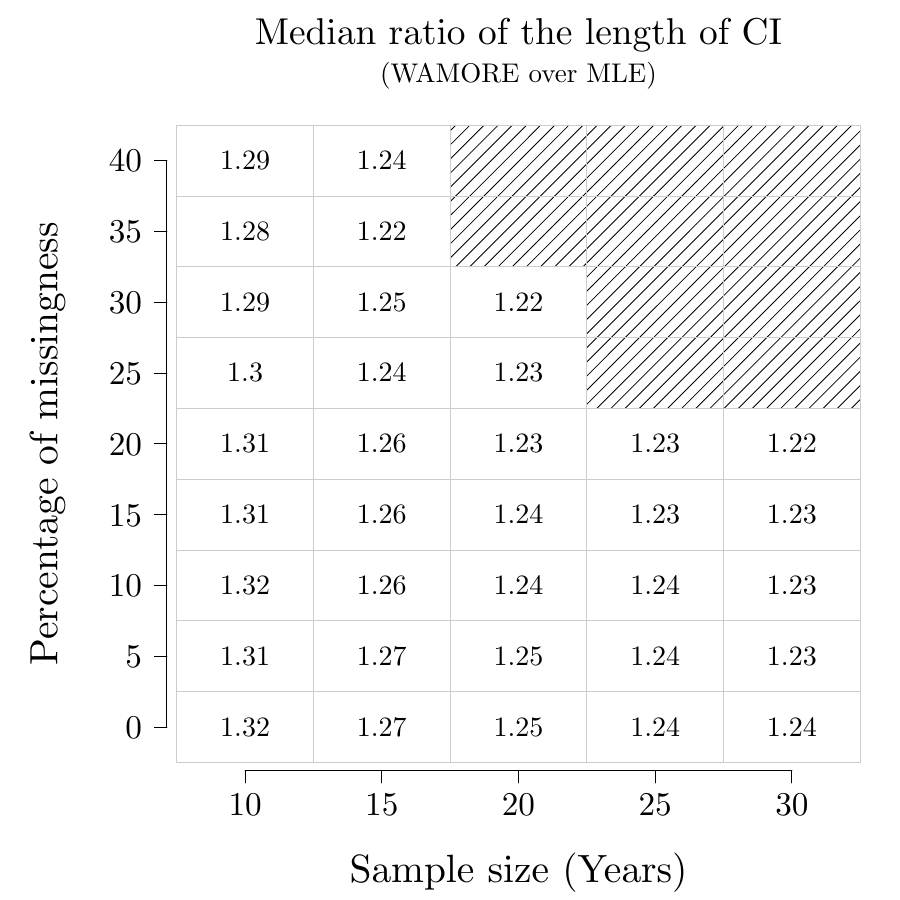}
    \caption{Median ratio of the length of confidence intervals for the setting with the WN + MAT + RW model.}
    \label{fig:wn_matern_rw_ratio_ci_length}
\end{figure}

\begin{figure}
    \centering
    \includegraphics[width=0.9\linewidth]{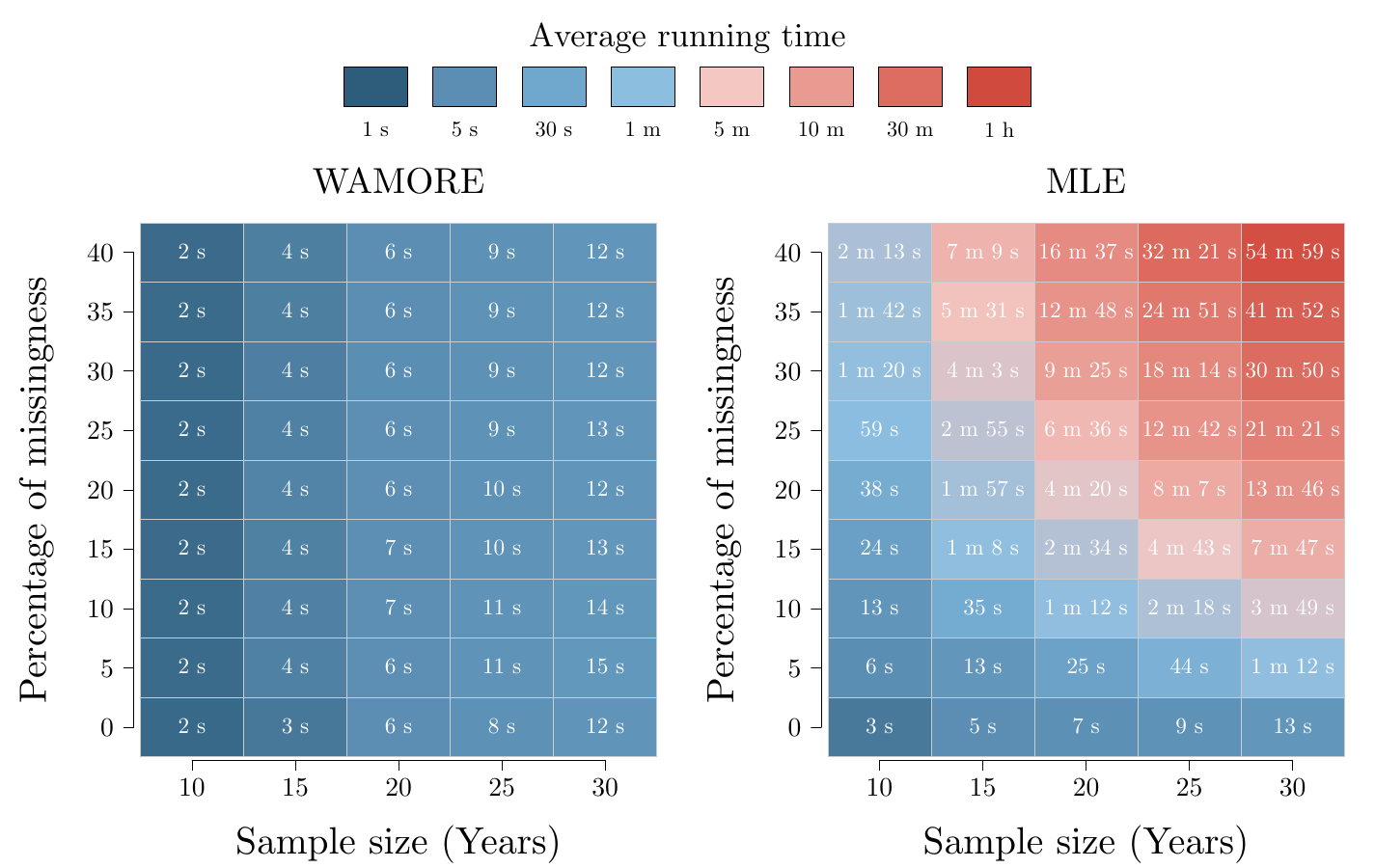}
    \caption{Average running time of the WAMORE (left panel) vs the MLE (right panel) for the setting with the WN + FL model.}
    \label{fig:wn_fl_running_time}
\end{figure}
\begin{figure}
    \centering
    \includegraphics[width=0.9\linewidth]{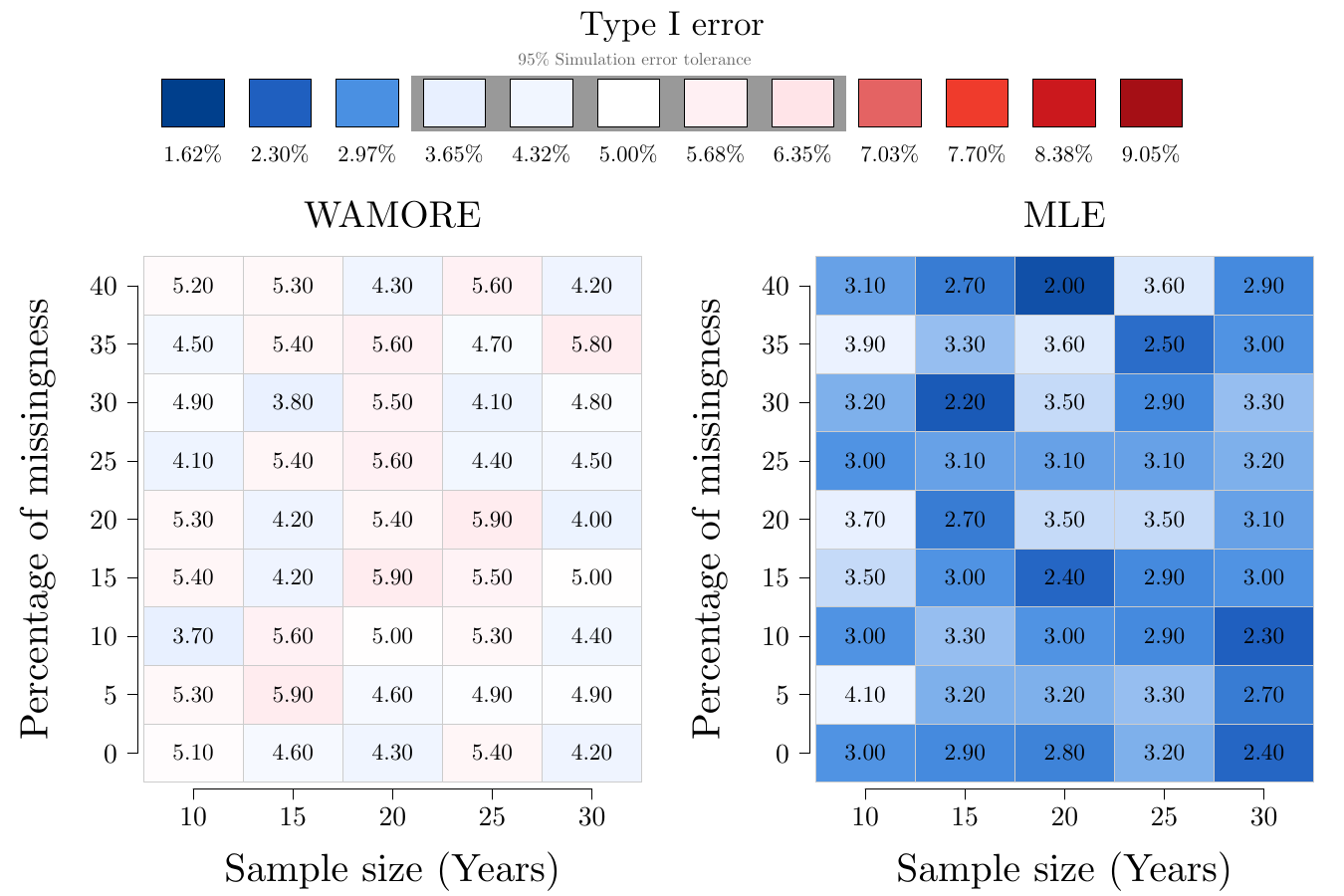}
    \caption{Empirical type I error (in \%) of the WAMORE (left panel) vs the MLE (right panel) for the setting with the WN + FL model.}
    \label{fig:wn_fl_type_1_error}
\end{figure}
\begin{figure}
    \centering
    \includegraphics[width=0.6\linewidth]{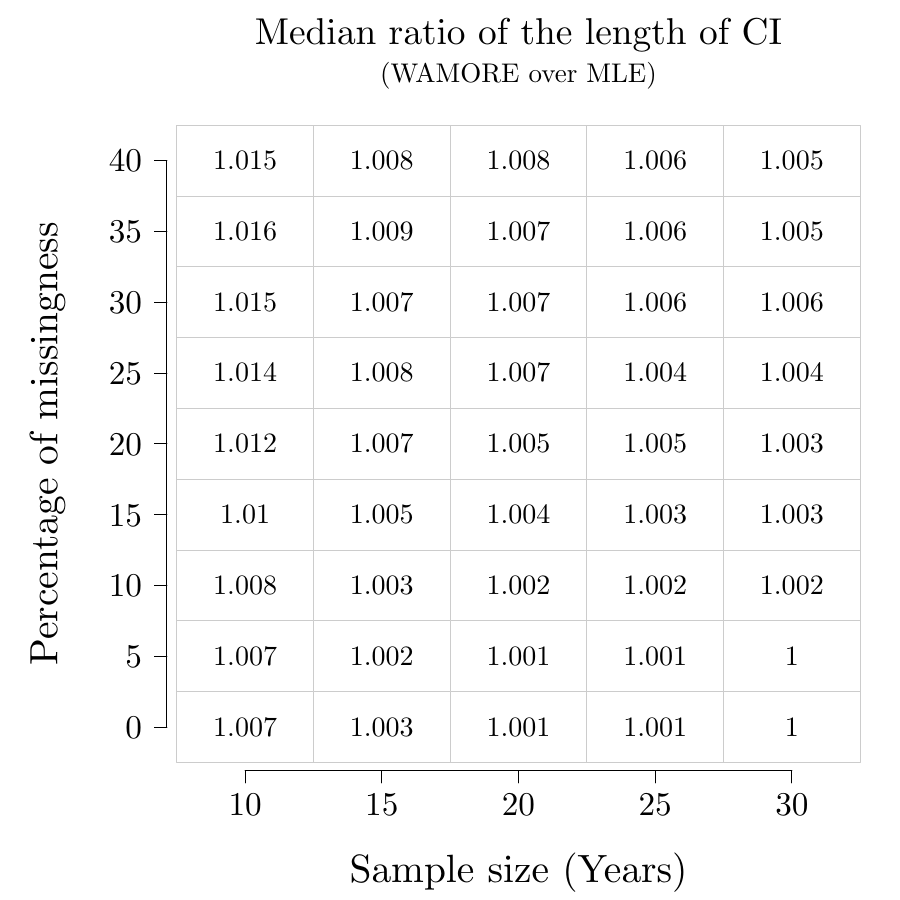}
    \caption{Median ratio of the length of confidence intervals for the setting with the WN + FL model.}
    \label{fig:wn_fl_ratio_length_ci}
\end{figure}
\begin{figure}
    \centering
    \includegraphics[width=0.8\linewidth]{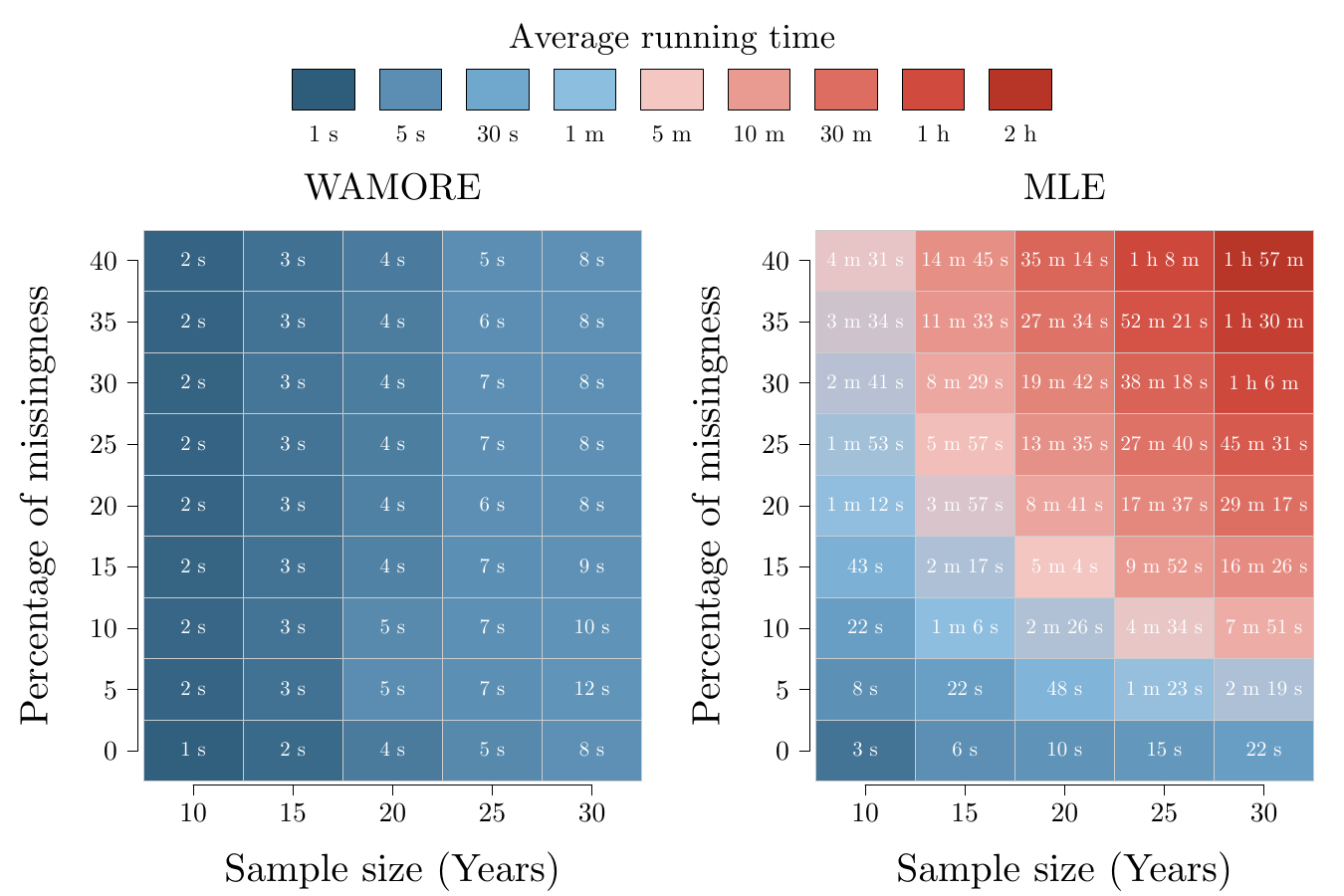}
    \caption{Average running time of the WAMORE (left panel) vs the MLE (right panel) for the setting with the WN + PL model.}
    \label{fig:wn_pl_running_time}
\end{figure}
\begin{figure}
    \centering
    \includegraphics[width=0.8\linewidth]{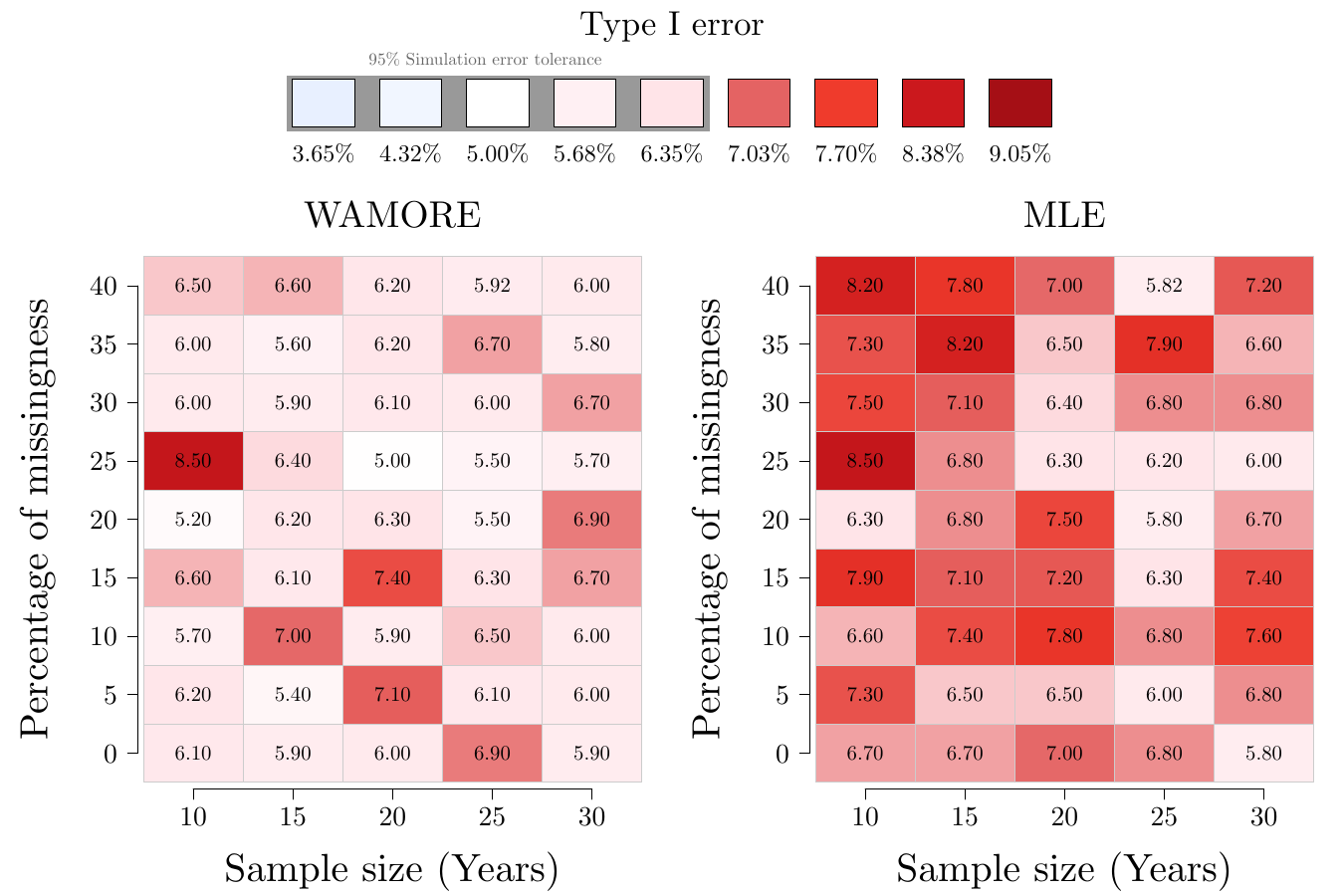}
    \caption{Empirical type I error of the WAMORE (left panel) vs the MLE (right panel) for the setting with the WN + PL model.}
    \label{fig:wn_pl_type_1_error}
\end{figure}
\begin{figure}
    \centering
    \includegraphics[width=0.6\linewidth]{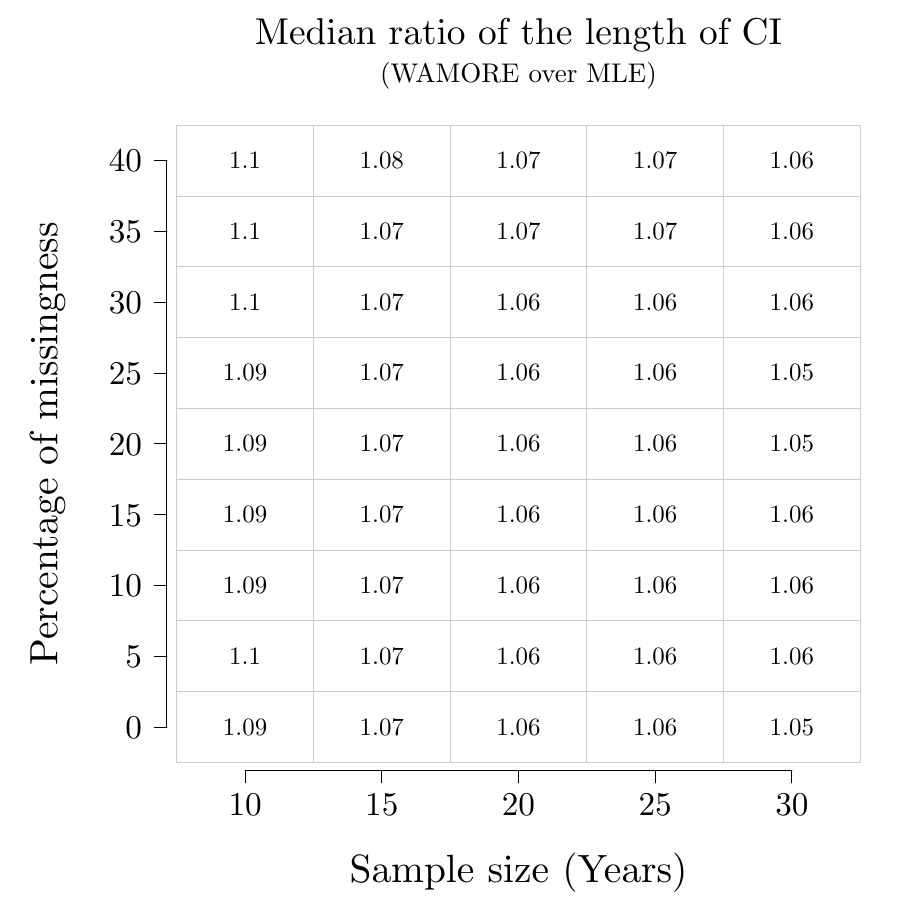}
    \caption{Median ratio of the length of confidence intervals for the setting with the WN + PL model.}
    \label{fig:wn_pl_ratio_length_ci}
\end{figure}

\begin{Remark}
    Based on our experimental findings, ignoring that the estimated WV $\hat{\bm{\nu}}$ is computed on the residuals (and not on the true error process) can introduce a bias in the last scales of the WV. This, in turn, can lead to biased estimates of the stochastic parameters, particularly if a stochastic process is primarily distinguishable in the last scales of the WV. Nevertheless, since we know that the last scales of the WV are also those that generally suffer from a higher variability (due to the lower number of wavelet coefficients on which the WV is estimated at these scales), it could be considered reasonable to ignore this approximation depending on the setting. Nevertheless, with the procedure proposed in Appendix \ref{app:IminusH}, we can still perform this approximation in a computationally efficient manner.
\end{Remark}

\FloatBarrier
\subsubsection{$\beta$-Mixing Missingness Process for Simulations}
\label{app:beta_mixing}

In this appendix we show that the Markov chain used for the simulations in Section~\ref{sec:simulations} is a special case of \eqref{eq:markov_missinigness}. Consider the two-state Markov chain with transition probabilities
\[
\begin{aligned}
P(Z_i=0 \mid Z_{i-1}=1)&=p_1,
\qquad
P(Z_i=1 \mid Z_{i-1}=1)=1-p_1,\\
P(Z_i=1 \mid Z_{i-1}=0)&=p_2,
\qquad
P(Z_i=0 \mid Z_{i-1}=0)=1-p_2.
\end{aligned}
\]

We now match this process with the representation in
\eqref{eq:markov_missinigness},

    \begin{equation*}
    Z_i =
\begin{cases}
    Z_{i-1}, & \text{with probability } \rho, \\
    W_i \sim \mathrm{Bernoulli}(\mu(\bm{\vartheta})), & \text{with probability } 1-\rho,
\end{cases}
\end{equation*}
with $\rho \in [0, 1)$ and $\mu(\bm{\bm{\vartheta}}) \in (0,1]$.

Under this representation,
\[
P(Z_i=0\mid Z_{i-1}=1)
=
(1-\rho)\left\{1-\mu(\bm{\vartheta})\right\},
\]
and
\[
P(Z_i=1\mid Z_{i-1}=0)
=
(1-\rho)\mu(\bm{\vartheta}).
\]

Matching these probabilities with the transition probabilities of the Markov chain yields
\begin{equation}
\label{eq:equality_transition_prob_markov}
    p_1=(1-\rho)\left\{1-\mu(\bm{\vartheta})\right\},
\qquad
p_2=(1-\rho)\mu(\bm{\vartheta}).
\end{equation}

Adding the two equations gives
\[
p_1+p_2=1-\rho,
\]
which implies
\[
\rho=1-p_1-p_2.
\]

Substituting this relation into the second equation of \eqref{eq:equality_transition_prob_markov} gives
\[
\mu(\bm{\vartheta})
=
\frac{p_2}{p_1+p_2}.
\]

Hence, the Markov chain used in the simulations is a special case of \eqref{eq:markov_missinigness} with
\[
\mu(\bm{\vartheta})
=
\frac{p_2}{p_1+p_2},
\qquad
\rho
=
1-p_1-p_2.
\]

\FloatBarrier
\subsubsection{Model Selection Simulation}
\label{app:model_selection_simulation}

In this appendix, we present a simulation study that investigate the performance of the model selection criterion presented in Section~\ref{sec:model_selection}. More precisely, we consider a subset of the model presented in \eqref{eq:func_model} with an intercept, a trend parameter, two sinusoidal periodic components (annual and semi-annual) with parameters in the range of the ones estimated on real GNSS times series analyzed in Section~\ref{sec:case_study}. For the stochastic model, we consider a combination of white noise and  flicker noise, which is one of the most commonly considered model in GNSS times series (see e.g., \citealp{he2019investigation}). In this simulation we consider the task of identifying an additional random walk component to this stochastic model which is a standard problem when modeling GNSS times series (see e.g., \citealp{kaczmarek2018identification, he2021analysis, he2019investigation}).

For the missingness process $\Z$, we generate missing observations using the Markov model described in \eqref{eq:markov_model_def}, with parameters from Setting~3 in Table~\ref{tab:parameters:missingness}. This setting corresponds to $10$\% of missing observations, which approximately matches the estimated median proportion of missing data in publicly available datasets (\citealp{bos2013fast}).

To study the performance the model selection criteria proposed in Section~\ref{sec:model_selection} with respect to the AIC (\citealp{akaike1974new}) implemented in the \texttt{Hector} software, we gradually increase the variance of the random walk and perform $1,000$ Monte Carlo simulations for each unique value of the random walk parameter.  For each simulation, we compute the proposed criterion as well as the AIC considering a stochastic model composed of white noise and flicker noise and a stochastic model composed of white noise, flicker noise and random walk. We record the selected stochastic model for both criterion. For each unique value of the random walk parameter, we count how many time the \textit{full} model (composed of white noise, flicker noise and random walk) is selected for both model selection criteria. We perform this simulation for two lengths of GNSS time series, respectively $20$ and $40$ years of daily data corresponding to respectively $7,300$ and $14,600$ observations. \\
The results of this simulation study are shown in Figure~\ref{fig:result_simu_model_selection}. It can be observed  that while the proposed criterion exhibits a marginally higher selection rate for the full model when the true model consists solely of white noise and flicker noise ($3.8\%$ compared to $0.1\%$ for the AIC in the setting with $20$ years of daily data, and $4.3\%$ compared to $0.2\%$ for the AIC in the setting with $40$ years of daily data), the proposed criterion generally demonstrate a higher selection rate for the full model when the true model includes white noise, flicker noise, and random walk. This difference can be as high as $15.3\%$ for the 20-year daily data setting and up to $18.1\%$ for the 40-year daily data setting. Overall, both criteria show very similar performance.

Additionally, it is worth highlighting that our proposed criterion also demonstrates a significant computational advantage over the likelihood-based criteria. For a dataset covering $20$ years of daily observations, our criterion's median runtime is $1$ minute and $19$ seconds, compared to $4$ minutes and $4$ seconds for the AIC. Similarly, for a dataset spanning $40$ years of daily observations, our criterion's median runtime is $4$ minutes and $6$ seconds, while the AIC approach takes a median of $29$ minutes and $8$ seconds.

\begin{figure}

    \centering
    \includegraphics[width=.9\linewidth]{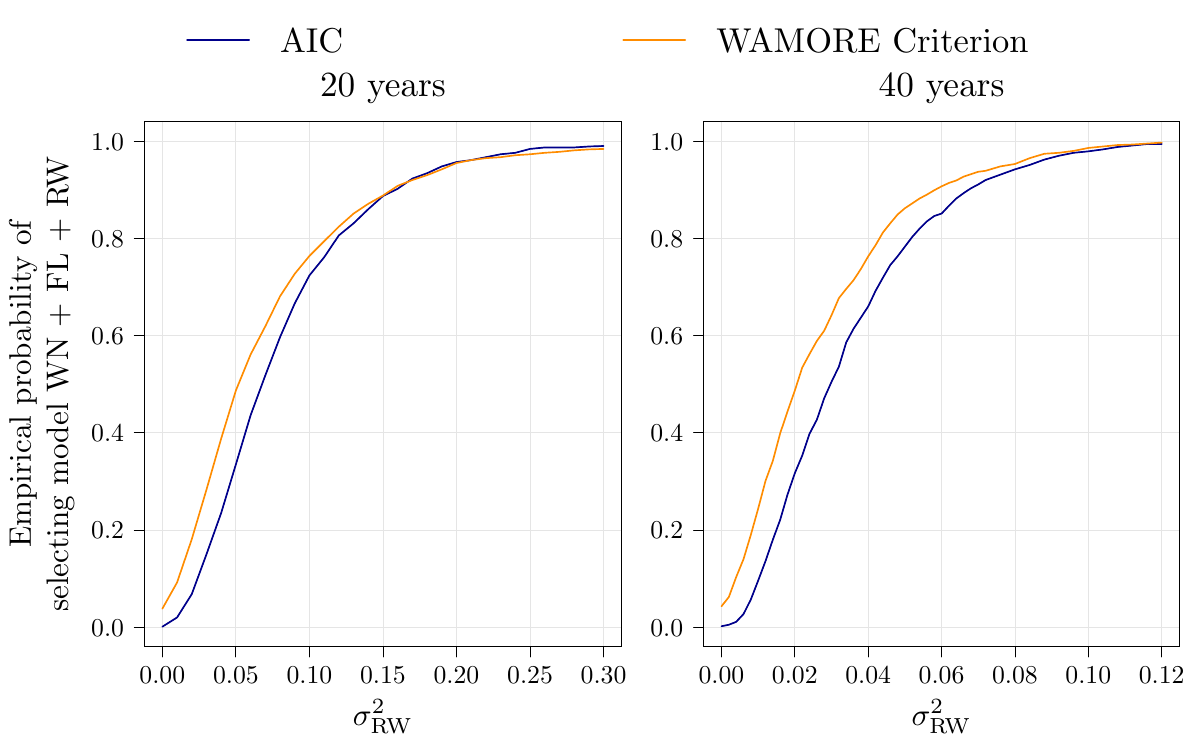}
    \caption{Selection rate of the AIC and the proposed WAMORE criterion for setting considering $20$ years of daily data (left panel) and $40$ years of daily data (right panel).}
    \label{fig:result_simu_model_selection}

\end{figure}

\FloatBarrier

\subsection{Case Study: Additional Results}
\label{app:case_study}
\subsubsection{Comparison with MLE}
\label{app:case_study_additional_graphs}

\begin{center}

\begin{figure}[H]
    \centering
    \includegraphics[width=0.9\linewidth]{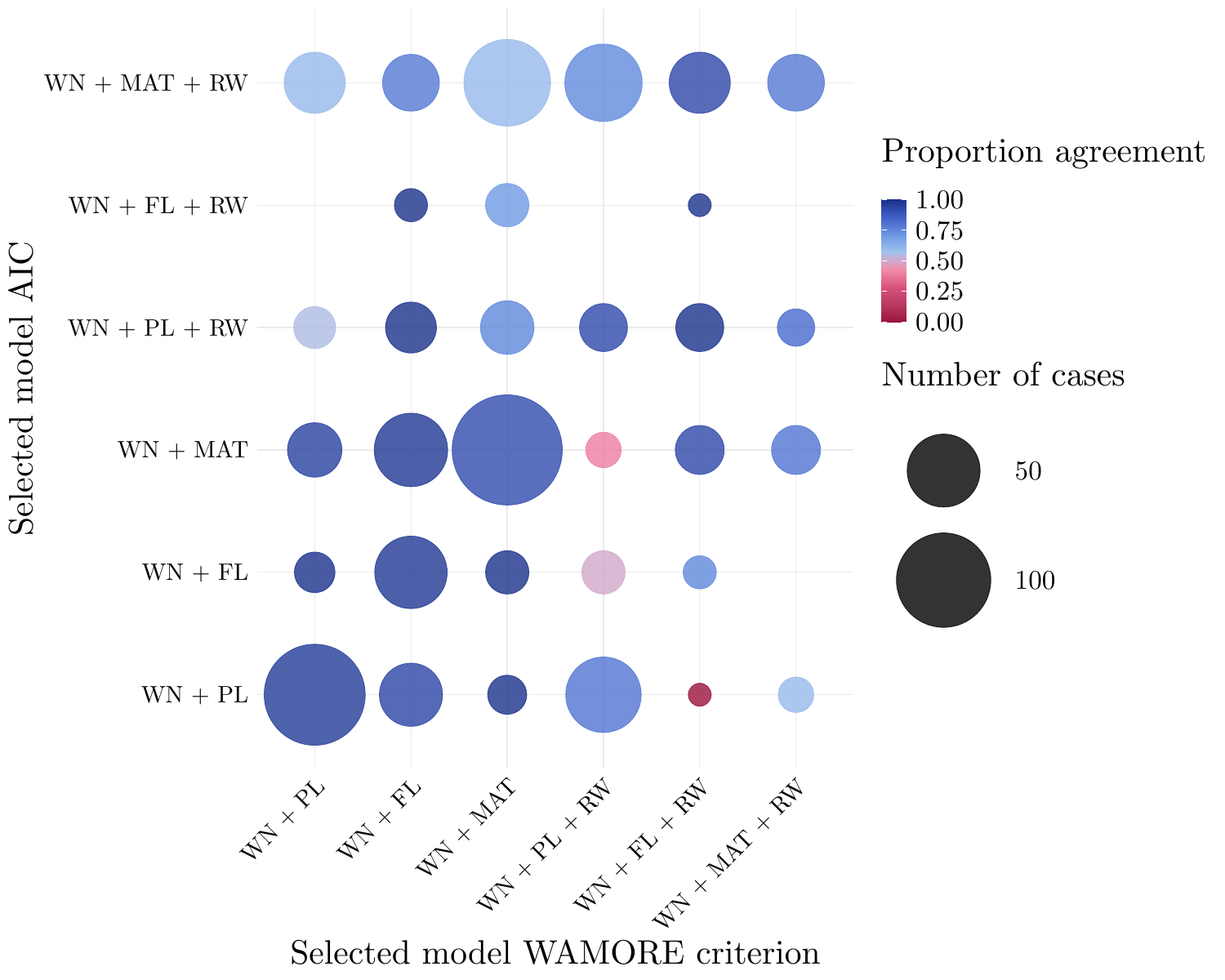}
    \caption{Agreement proportion and case count by selected model combination for the AIC and the WAMORE criterion.}
    \label{fig_agreement_test}
\end{figure}

\begin{figure}[H]
    \centering
    \includegraphics[width=0.9\linewidth]{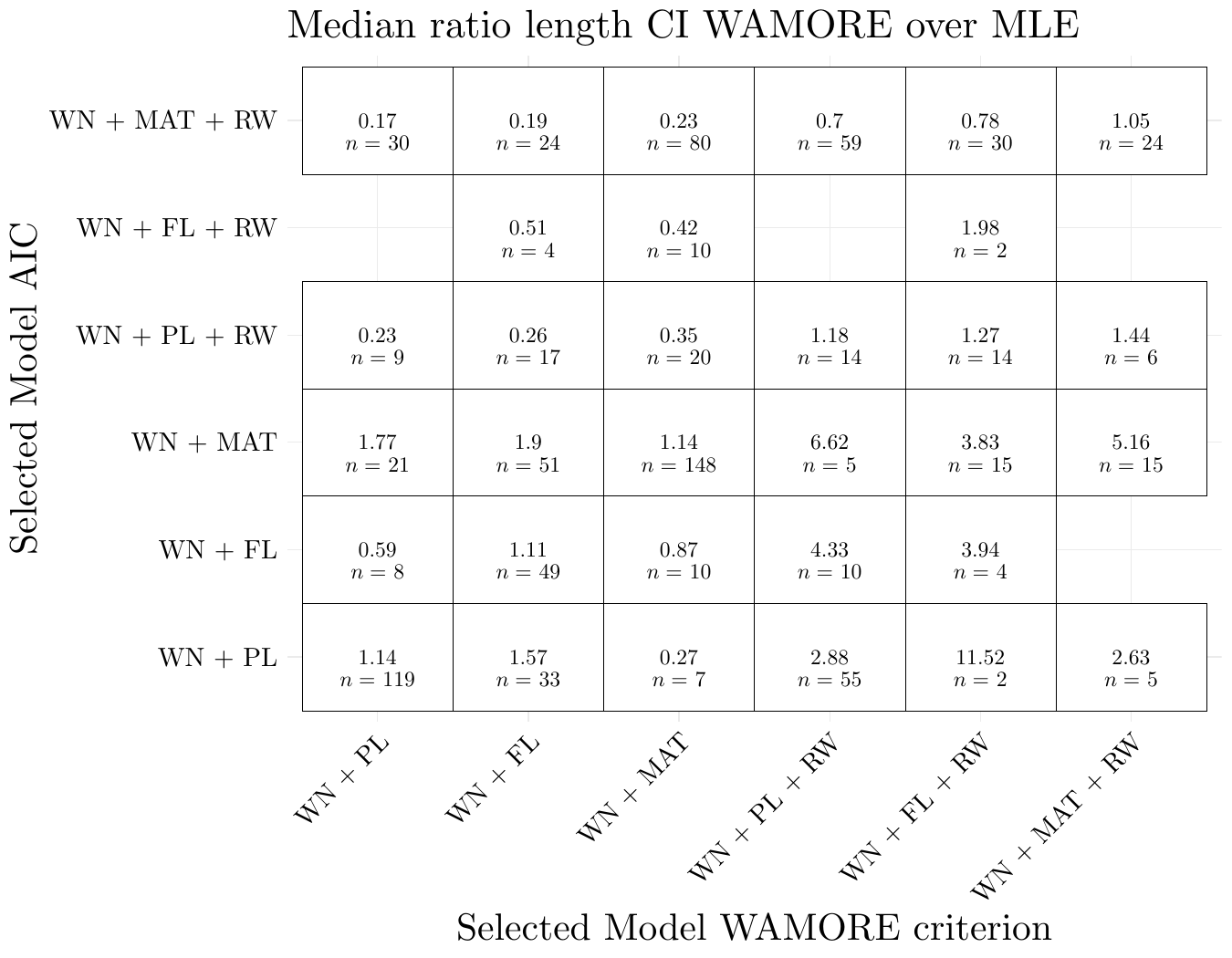}
    \caption{Ratio of median CI length for all combinations of selected models for the proposed criterion and the AIC (WAMORE CI length over MLE CI length).}
    \label{fig:median_ratio_length_ci}
\end{figure}

\end{center}

\subsubsection{Emulation Study}
\label{app:emulation_study}

We design an emulation study based on one of the stations analyzed in the case study for which both model selection criteria selected the model composed of a WN and a MAT process. More precisely, we consider the ACU5 station near Boston which, as of the drafting of this work, records more than $18$ years of daily data and reporting $10.77\%$ of missing observations.

\noindent We first estimate the model considering a WN and a MAT process as the stochastic model. We then fix all estimated parameters and generate $1,000$ simulated times series data from this model and do so also with larger signal lengths compared to the observed one to which we add $5$, $10$ and $15$ years of daily data. We omit showing the results here since, in this case, the WAMORE and the MLE always have good empirical coverage for the trend parameter (within simulation error) no matter whether we estimate this model directly or if we run a selection procedure versus the model with an additional RW. This performance is confirmed by the simulations in Appendix~\ref{app:main_simu_results}.

\begin{wrapfigure}[12]{r}{0.55\textwidth}
\centering

\includegraphics[width=0.5\textwidth]{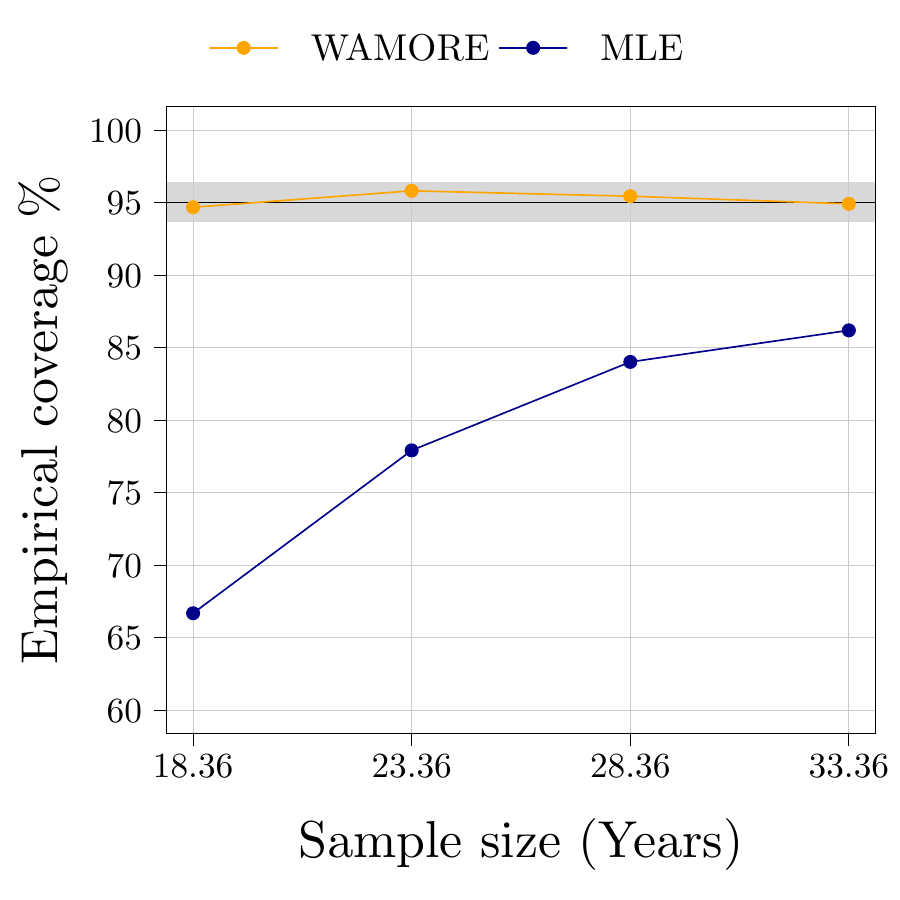}
\vspace{-10pt}
\caption{Empirical coverages of the WAMORE and MLE when the true model is WN + MAT + RW.}
\label{fig:emp_coverage_wn_matern_rw_w_missing}
\end{wrapfigure}

We then run another emulation study by adding a RW to the true model, where the innovation variance of the RW is within range of those estimated when this model is selected in the signals analyzed in the case study. Also in this case we generate $1,000$ signals and consider the original signal length to which we add $5$, $10$ and $15$ years of daily data. For each of these lengths, we compare the empirical coverages obtained (i) when we directly estimate the true model (i.e., WN + MAT + RW) and then (ii) when we proceed to model selection using the WAMORE criterion proposed in Section~\ref{sec:model_selection} and the AIC implemented in \texttt{Hector}. The empirical coverage of the WAMORE and the MLE when we directly estimate the true model (i.e., WN + MAT + RW) is shown in Figure~\ref{fig:emp_coverage_wn_matern_rw_w_missing} where it can be observed that the WAMORE estimator consistently provides standard errors that ensure valid inference for the trend parameter across all signal lengths. In contrast, the MLE's estimated standard errors deliver an empirical coverage which is too liberal and falls outside simulation error ranges. The WAMORE hence appears to provide slightly larger standard errors than the MLE when a RW is present and consequently appears to better target the nominal inferential level.

We now consider the empirical coverage conditional on the outcome of the model selection procedure between the candidate models WN + MAT and WN + MAT + RW. For these results, in Figure~\ref{fig:emp_coverage_wn_matern_rw_w_missing_w_mdl_selection_and_selection_rate} we observe (i) the empirical coverage at the selected model for both methods as well as (ii) the empirical selection rate of the correct model. As can be seen, both criteria exhibit similar empirical selection rates of the correct model and select the correct model more frequently as the signal grows in size, resulting in empirical coverages improving towards the nominal level. In particular, the WAMORE appears to approach the required nominal level faster than the MLE as the signal length increases.

\begin{figure}[H]
    \centering
    \includegraphics[width=1\linewidth]{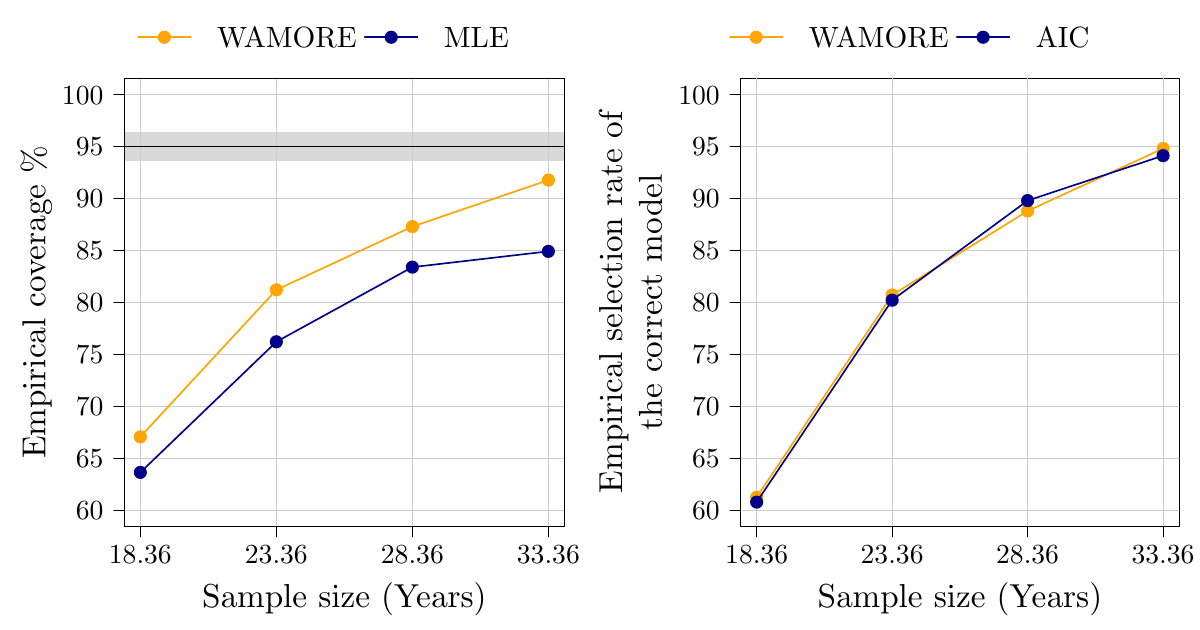}
    \caption{Empirical coverages of the WAMORE and MLE at the selected model and empirical selection rate of the correct model when the true model is WN + MAT + RW.}
    \label{fig:emp_coverage_wn_matern_rw_w_missing_w_mdl_selection_and_selection_rate}
\end{figure}

These results help to better understand and shed light on the differences observed in the case study discussed in Section~\ref{sec:case_study}. Indeed, these findings suggest that the WAMORE estimator appears to provide slightly larger but more reliable standard errors with respect to the MLE when the stochastic model includes a RW. Also, the proposed model selection criterion appears to have similar performance to the AIC implemented in \texttt{Hector}.

\end{document}